\def\llncs{0}
\def\fullpage{1}
\def\anonymous{0}
\def\authnote{0}
\def\notxfont{0}
\def\submission{0}
\def\reply{0}
\def\cameraready{0}
\def\noaux{1}
\ifnum\submission=1
\def\anonymous{1}
\def\llncs{1}
\else
\fi

\ifnum\cameraready=1
\def\llncs{1}
\def\anonymous{0}
\def\authnote{0}
\else
\fi

\ifnum\anonymous=1
\def\authnote{0}
\else
\fi

\def\mac{0}

\ifnum\mac=0
\ifnum\llncs=1
	\documentclass[envcountsect,a4paper,runningheads]{llncs}
\else
	\documentclass[letterpaper,hmargin=1.05in,vmargin=1.05in]{article}
			\ifnum\fullpage=1
		\usepackage{fullpage}
		\fi
\fi
\else
\ifnum\llncs=1
	\documentclass[dvipdfmx,envcountsect,a4paper,runningheads]{llncs}
\else
	\documentclass[dvipdfmx,letterpaper,hmargin=1.05in,vmargin=1.05in]{article}
			\ifnum\fullpage=1
		\usepackage{fullpage}
		\fi
\fi
\fi

\ifnum\reply=1
\usepackage{ulem}
\renewcommand{\emph}{\textit}
\else
\fi

\usepackage[%
  colorlinks=true,
  citecolor=darkgreen,
  linkcolor=darkblue,
  urlcolor=darkblue,
  pagebackref=true
]{hyperref}

\usepackage{amsmath, amsfonts, amssymb, mathtools,amscd}

\usepackage{amsthm}

\usepackage{lmodern}
\usepackage[T1]{fontenc}
\usepackage[utf8]{inputenc}

\usepackage{etex}

\usepackage{arydshln} 
\usepackage{url}
\usepackage{ifthen}
\usepackage{bm}
\usepackage{multirow}
\usepackage[dvips]{graphicx}
\usepackage[usenames]{color}
\usepackage{xcolor,colortbl} 
\usepackage{threeparttable}
\usepackage{comment}
\usepackage{paralist,verbatim}
\usepackage{cases}
\usepackage{booktabs}
\usepackage{braket}
\usepackage{cancel} 
\usepackage{framed}
\usepackage{authblk}
\usepackage{pifont}
\usepackage{physics}
\definecolor{darkblue}{rgb}{0,0,0.6}
\definecolor{darkgreen}{rgb}{0,0.5,0}
\definecolor{maroon}{rgb}{0.5,0.1,0.1}
\definecolor{dpurple}{rgb}{0.2,0,0.65}
\definecolor{chocolate}{rgb}{0.8,0.4,0.1}

\usepackage[capitalise,noabbrev]{cleveref}
\usepackage[absolute]{textpos}
\usepackage[final]{microtype}
\usepackage[absolute]{textpos}
\usepackage{autonum}

\usepackage{dsfont}
\DeclareMathAlphabet{\mathpzc}{OT1}{pzc}{m}{it}

\ifnum\submission=1
\renewcommand*{\backref}[1]{}
\def\notxfont{1}
\else
\fi

\ifnum\llncs=1
\renewcommand{\subparagraph}{\paragraph}
\else
\ifnum\notxfont=1
\else
\usepackage{mathpazo}
\usepackage{newtxtext}
\usepackage{helvet}
\fi
\fi

\newtheoremstyle{thicktheorem}%
{\topsep}
{\topsep}
{\itshape}{}%
{\bfseries}%
{.}
{ }%
{\thmname{#1}\thmnumber{ #2}%
		\thmnote{ (#3)}%
}

\newtheoremstyle{remark}
{\topsep}
{\topsep}
	{}
	{}
	{}
	{.}
	{ }
	{\textit{\thmname{#1}}\thmnumber{ #2}
			\thmnote{ (#3)}%
	}

\ifnum\llncs=0
	\theoremstyle{thicktheorem}
	\newtheorem{theorem}{Theorem}[section]
	\newtheorem{lemma}[theorem]{Lemma}

	\newtheorem{definition}[theorem]{Definition}

	\theoremstyle{remark}
	
	\newtheorem{remark}[theorem]{Remark}

\fi

	\crefname{theorem}{Theorem}{Theorems}
	\crefname{assumption}{Assumption}{Assumptions}
	\crefname{construction}{Construction}{Constructions}
	\crefname{corollary}{Corollary}{Corollaries}
	\crefname{conjecture}{Conjecture}{Conjectures}
	\crefname{definition}{Definition}{Definitions}
	\crefname{exmaple}{Example}{Examples}
	\crefname{experiment}{Experiment}{Experiments}
	\crefname{counterexample}{Counterexample}{Counterexamples}
	\crefname{lemma}{Lemma}{Lemmata}
	\crefname{observation}{Observation}{Observations}
	\crefname{proposition}{Proposition}{Propositions}
	\crefname{remark}{Remark}{Remarks}
	\crefname{claim}{Claim}{Claims}
	\crefname{fact}{Fact}{Facts}
	\crefname{note}{Note}{Notes}

\ifnum\llncs=1
 \crefname{appendix}{App.}{Appendices}
 \crefname{section}{Sec.}{Sections}
\else
\fi

\ifnum\llncs=1
\renewcommand*{\backref}[1]{}
\else
	\renewcommand*{\backref}[1]{(Cited on page~#1.)}
	\ifnum\notxfont=1
	\else
		\usepackage{newtxtext}
	\fi
\fi

\newcommand*{\keys}[1]{\mathsf{#1}}

\newcommand{\Oracle}[1]{O_{\mathtt{#1}}}

\newcommand*{\algo}[1]{\ensuremath{\mathsf{#1}}}
\newcommand*{\qalgo}[1]{\ensuremath{\mathpzc{#1}}}
\newcommand*{\qstate}[1]{\mathpzc{#1}}
\newcommand*{\qreg}[1]{{\color{gray}{\mathsf{#1}}}}

\newcounter{expitem}

\usepackage{mathtools}

\newcommand{\chosen}{\leftarrow}

\newcommand{\lrun}{\leftarrow}

\newcommand{\la}{\leftarrow}
\newcommand{\ra}{\rightarrow}
\renewcommand{\gets}{\leftarrow}

\newcommand{\seteq}{\coloneqq}

\newcommand{\tensor}{\otimes}
\newcommand{\concat}{\|}

\newcommand{\Ch}{\mathsf{Ch}}
\newcommand{\qCh}{\qalgo{Ch}}

\newcommand{\cM}{\mathcal{M}}

\newcommand{\cO}{\mathcal{O}}

\newcommand{\cQ}{\mathcal{Q}}

\newcommand{\cU}{\mathcal{U}}

\newcommand{\cX}{\mathcal{X}}
\newcommand{\cY}{\mathcal{Y}}
\newcommand{\cZ}{\mathcal{Z}}

\newcommand{\qA}{\qalgo{A}}
\newcommand{\qB}{\qalgo{B}}

\newcommand{\qD}{\qalgo{D}}

\def\makeuppercase#1{
\expandafter\newcommand\csname sf#1\endcsname{\mathsf{#1}}
\expandafter\newcommand\csname frak#1\endcsname{\mathfrak{#1}}
\expandafter\newcommand\csname bb#1\endcsname{\mathbb{#1}}
\expandafter\newcommand\csname bf#1\endcsname{\textbf{#1}}
}

\def\makelowercase#1{
\expandafter\newcommand\csname frak#1\endcsname{\mathfrak{#1}}
\expandafter\newcommand\csname bf#1\endcsname{\textbf{#1}}
}

\newcounter{char}
\loop
   \edef\letter{\alph{char}}
   \edef\Letter{\Alph{char}}
   \expandafter\makelowercase\letter
   \expandafter\makeuppercase\Letter
   \stepcounter{char}
   \unless\ifnum\thechar>26
\repeat

\def\makeuppercase#1{
\expandafter\newcommand\csname tl#1\endcsname{\widetilde{#1}}
}

\def\makelowercase#1{
\expandafter\newcommand\csname tl#1\endcsname{\widetilde{#1}}
}

\newcommand{\bit}{\{0,1\}}

\newcommand{\Ms}{\mathcal{M}}

\newcommand{\Ks}{\mathcal{K}}

\newcommand{\secp}{\lambda}

\newcommand{\coin}{\keys{coin}}

\newcommand{\cert}{\keys{cert}}

\newcommand{\aux}{\mathsf{aux}}

\newcommand{\advt}[2]{\mathsf{Adv}_{#1}^{\mathsf{#2}}}

\newcommand{\advb}[3]{\mathsf{Adv}_{#1}^{\mathsf{#2} \mbox{-} \mathsf{#3}}}
\newcommand{\advc}[4]{\mathsf{Adv}_{#1}^{\mathsf{#2} \mbox{-} \mathsf{#3} \mbox{-} \mathsf{#4}}}

\newcommand{\expa}[2]{\mathsf{Expt}_{#1}^{\mathsf{#2}}}
\newcommand{\expb}[3]{\mathsf{Exp}_{#1}^{ \mathsf{#2} \mbox{-} \mathsf{#3}}}
\newcommand{\expc}[4]{\mathsf{Exp}_{#1}^{ \mathsf{#2} \mbox{-} \mathsf{#3} \mbox{-} \mathsf{#4}}}

\newcommand{\Hyb}{\mathsf{Hyb}}

\newcommand{\hybi}[1]{\mathsf{Hyb}_{#1}}

\newcommand*{\pk}{\keys{pk}}
\newcommand*{\sk}{\keys{sk}}
\newcommand*{\dk}{\keys{dk}}
\newcommand*{\ek}{\keys{ek}}
\newcommand*{\vk}{\keys{vk}}

\newcommand*{\key}{\keys{k}}
\newcommand*{\msk}{\keys{msk}}

\newcommand*{\pp}{\keys{pp}}

\newcommand*{\tk}{\keys{tk}}

\newcommand*{\ct}{\keys{ct}}

\newcommand*{\msg}{\keys{m}}

\newcommand{\qct}{\qstate{ct}}
\newcommand{\qsk}{\qstate{sk}}

\newcommand{\qdk}{\qstate{dk}}

\newcommand{\SD}{\mathsf{SD}}

\newcommand{\oh}{\mathsf{o2h}}

\newenvironment{boxfig}[2]{\begin{figure}[#1]\fbox{\begin{minipage}{0.97\linewidth}
                        \vspace{0.2em}
                        \makebox[0.025\linewidth]{}
                        \begin{minipage}{0.95\linewidth}
            {{
                        #2 }}
                        \end{minipage}
                        \vspace{0.2em}
                        \end{minipage}}
                        }
                        {\end{figure}}

\newcommand{\Setup}{\algo{Setup}}

\newcommand{\KeyGen}{\algo{KeyGen}}

\newcommand{\KG}{\algo{KG}}
\newcommand{\Enc}{\algo{Enc}}
\newcommand{\Dec}{\algo{Dec}}

\newcommand{\Vrfy}{\algo{Vrfy}}

\newcommand{\qKG}{\qalgo{KG}}

\newcommand{\qEnc}{\qalgo{Enc}}
\newcommand{\qDec}{\qalgo{Dec}}

\newcommand{\qVrfy}{\qalgo{Vrfy}}
\newcommand{\qKGt}{\qalgo{\widetilde{KG}}}
\newcommand{\qR}{\qalgo{R}}

\newcommand\ABE{\algo{ABE}}

\newcommand\SKFE{\algo{SKFE}}

\newcommand{\ske}{\algo{ske}}

\newcommand{\Sim}{\algo{Sim}}

\newcommand{\PRF}{\algo{PRF}}

\newcommand{\negl}{{\mathsf{negl}}}

\newcommand{\poly}{{\mathrm{poly}}}

\newcommand{\zo}[1]{\{0,1\}^{#1}}
\newcommand{\bin}{\{0,1\}}

\newcommand{\xor}{\oplus}

\newcommand{\class}[1]{\mathsf{#1}}

\newcommand{\Ppoly}{\class{P/poly}}

\newcommand{\NCone}{\class{NC}^1}

\newcommand{\calO}{\mathcal{O}}

\newcommand{\calY}{\mathcal{Y}}

\newcommand{\SKFESKL}{\algo{SKFE}\textrm{-}\algo{CR}\textrm{-}\algo{SKL}}
\newcommand{\ABESKL}{\algo{ABE}\textrm{-}\algo{CR}\textrm{-}\algo{SKL}}

\newcommand{\abe}{\mathsf{abe}}

\newcommand{\tlC}{\widetilde{C}}

\newcommand{\ABECRCRSKL}{\mathsf{ABE}\textrm{-}\mathsf{CR^2}\textrm{-}\mathsf{SKL}}

\newcommand{\qaux}{\qstate{aux}}

\newcommand{\skfe}{\algo{skfe}}
\newcommand{\PKFESKL}{\mathsf{PKFE\textrm{-}SKL}}
\newcommand{\List}[1]{L_{\mathtt{#1}}}

\newcommand{\MIABE}{\algo{MI}\textrm{-}\algo{ABE}}

\newcommand{\decision}{d}

\newcommand{\authornote}[3]{\textcolor{#3}{[\textsc{#1:} {#2}]}}
\ifnum\authnote=1
\newcommand{\fuyuki}[1]{\authornote{Fuyuki}{#1}{chocolate}}
\newcommand{\ryo}[1]{\authornote{Ryo}{#1}{darkblue}}
\newcommand{\nikhil}[1]{\authornote{Nikhil}{#1}{red}}
\else
\newcommand{\fuyuki}[1]{}
\newcommand{\ryo}[1]{}
\newcommand{\nikhil}[1]{}
\fi

\ifnum\llncs=1
\let\oldvec\vec
\let\vec\oldvec
\makeatletter
\renewcommand*\l@author[2]{}
\renewcommand*\l@title[2]{}
\makeatletter
\fi    

\theoremstyle{remark}

\ifnum\submission=1
\title{
\textbf{PKE and ABE with \\ Collusion-Resistant Secure Key Leasing}
}
\else
\title{
\textbf{PKE and ABE with \\ Collusion-Resistant Secure Key Leasing}
}
\fi

\begin{document}

\ifnum\anonymous=1 
\ifnum\llncs=1
\author{\empty}\institute{\empty}
\else
\author{}
\fi
\else
%
%
\ifnum\llncs=1
\author{
	Fuyuki Kitagawa\inst{1,2} \and Ryo Nishimaki\inst{1,2} \and Nikhil Pappu\inst{3}
}
\institute{
	NTT Social Informatics Laboratories, Tokyo, Japan \and NTT Research Center for Theoretical Quantum Information, Atsugi, Japan \and Portland State University, USA
}
\else
%
%
\author[$\dagger$ $\diamondsuit$]{\hskip 1em Fuyuki Kitagawa}
\author[$\dagger$ $\diamondsuit$]{\hskip 1em Ryo Nishimaki}
\author[$\star$]{\hskip 1em Nikhil Pappu\thanks{Part of this work was done while visiting NTT Social Informatics Laboratories as an internship.}
\thanks{Supported by the US National Science Foundation (NSF) via Fang Song's Career Award (CCF-2054758). }}
\affil[$\dagger$]{{\small NTT Social Informatics Laboratories, Tokyo, Japan}\authorcr{\small \{fuyuki.kitagawa,ryo.nishimaki\}@ntt.com}}
\affil[$\diamondsuit$]{{\small NTT Research Center for Theoretical Quantum Information, Atsugi, Japan}}
\affil[$\star$]{{\small Portland State University, USA}\authorcr{\small nikpappu@pdx.edu}}
\renewcommand\Authands{, }
\fi 
\fi

\ifnum\llncs=1
\date{}
\else
\ifnum\anonymous=0
\date{\today}
\else
\date{}
\fi
\fi

\maketitle

\begin{abstract}
Secure key leasing (SKL) is an advanced encryption functionality that allows a secret key holder to generate a quantum decryption key and securely lease it to a user.
Once the user returns the quantum decryption key (or provides a classical certificate confirming its deletion), they lose their decryption capability.
Previous works on public key encryption with SKL (PKE-SKL) have only considered the single-key security model, where the adversary receives at most one quantum decryption key. However, this model does not accurately reflect real-world applications of PKE-SKL.
To address this limitation, we introduce \emph{collusion-resistant security} for PKE-SKL (denoted as PKE-CR-SKL). In this model, the adversary can adaptively obtain multiple quantum decryption keys and access a verification oracle which validates the correctness of queried quantum decryption keys. Importantly, the size of the public key and ciphertexts must remain independent of the total number of generated quantum decryption keys.
We present the following constructions:
\begin{itemize}
	\item A PKE-CR-SKL scheme based on the learning with errors (LWE) assumption.
	\item An attribute-based encryption scheme with collusion-resistant SKL (ABE-CR-SKL), also based on the LWE assumption.
	\item An ABE-CR-SKL scheme with classical certificates, relying on multi-input ABE with polynomial arity.
\end{itemize}
\end{abstract}

\ifnum\llncs=1
\else
\newpage
\setcounter{tocdepth}{2}
\tableofcontents

\newpage
\fi


\section{Introduction}\label{sec:intro}

\paragraph{Secure key leasing.}
Encryption with secure key leasing (SKL) enables a secret key holder to generate a quantum decryption key and lease it securely to another party.
Once the lessee returns the quantum decryption key, they lose their ability to decrypt ciphertexts.
Since its introduction by Kitagawa and Nishimaki for secret key functional encryption (SKFE)~\cite{AC:KitNis22}, SKL has been extensively studied~\cite{EC:AKNYY23,TCC:AnaPorVai23,EPRINT:CGJL23,EPRINT:MorPorYam23,TCC:AnaHuHua24,EC:BGKMRR24,myEPRINT:KitMorYam24} due to its strong security guarantee and practical applications.

\paragraph{Collusion-resistant SKL.}
Most prior works~\cite{EC:AKNYY23,TCC:AnaPorVai23,EPRINT:CGJL23,EPRINT:MorPorYam23,TCC:AnaHuHua24} study how to achieve public-key encryption (PKE) with SKL schemes from standard cryptographic assumptions. All prior works on PKE-SKL have explored the setting where an adversary can obtain \emph{only one} quantum decryption key.
However, this single-key security model does not accurately reflect real-world scenarios. In practice, once a lessee returns their decryption key and it is verified (i.e., after revocation), the lessor may lease another decryption key, even to the same lessee. Moreover, in realistic settings, a single secret key holder may need to generate and lease \emph{multiple} quantum decryption keys to various entities. To accurately capture this setting, we consider adversaries capable of obtaining \emph{an unbounded number of quantum decryption keys}, even in the standard PKE setting.
We define this model as \emph{collusion-resistant} SKL.

Previous works~\cite{TCC:AnaPorVai23,TCC:AnaHuHua24} presented delegation tasks as a compelling application like the following:
Consider a scenario where a system administrator unexpectedly needs to take leave and must temporarily assign their responsibilities---including access to encrypted data---to a colleague by providing decryption keys.
In such cases, the single-key security model is inadequate, as the administrator may need to take leave multiple times or assign their responsibilities to different colleagues.
Another potential application of collusion-resistant PKE-SKL is in streaming services.
Encrypted videos are made accessible to subscribers through quantum decryption keys.
When a user unsubscribes, they return their leased keys, losing access to the content, and their subscription fees are canceled. By utilizing an attribute-based encryption~\cite{EC:SahWat05} variant of collusion-resistant PKE-SKL, it becomes possible to precisely control video access based on user attributes, such as location-restrictions or premium and basic subscription plans.

\paragraph{Our goal: collusion-resistant SKL from weaker assumptions.}
The notion of collusion-resistant SKL is both natural and compelling. However, it has yet to be achieved from well-established assumptions.  
All known PKE with SKL schemes based on standard assumptions~\cite{EC:AKNYY23,TCC:AnaPorVai23,EPRINT:CGJL23,TCC:AnaHuHua24,myEPRINT:KitMorYam24} do not remain secure in the collusion-resistant setting.  
While some existing constructions seem to imply collusion-resistant SKL, such as public-key functional encryption (PKFE) schemes with SKL~\cite{EC:AKNYY23,EC:BGKMRR24} and collusion-resistant single decryptor encryption (SDE)\footnote{In short, SDE is PKE where the decryption keys are copy-protected (i.e., unclonable). In general, SDE implies encryption with SKL (see the discussion by Agrawal et al.~\cite{EC:AKNYY23} for details).}~\cite{TCC:CakGoy24}, they rely on strong assumptions. These include post-quantum secure indistinguishability obfuscation (IO) or collusion-resistant PKFE, which in turn implies IO, albeit with a sub-exponential security loss~\cite{JACM:BitVai18,C:AnaJai15,EPRINT:AnaJaiSah15a}. Achieving them from well-established assumptions still remains elusive.
In this work, we aim to construct the first collusion-resistant SKL schemes based on weaker assumptions.


%

\subsection{Our Results}
Our main contributions are summarized as follows:

\begin{enumerate}
    \item \emph{Definition of collusion-resistant PKE-SKL (PKE-CR-SKL):} We formally define PKE-CR-SKL, ensuring that if all quantum decryption keys are returned and successfully verified, users lose decryption capabilities. We extend the indistinguishability against key leasing attacks (IND-KLA) security definition~\cite{EC:AKNYY23} to the collusion-resistant setting. One notable feature of this definition is that the adversary can send multiple queries to the verification oracle, which confirms the validity of returned decryption keys. Another important feature is that the public key and ciphertext size are independent of the number of leased decryption keys (up to logarithmic factors).

    \item \emph{Construction of IND-KLA secure PKE with collusion-resistant SKL:} We propose an IND-KLA secure PKE-CR-SKL scheme based on the learning with errors (LWE) assumption.

     \item \emph{Attribute-based encryption with collusion-resistant SKL:} We construct an ABE-CR-SKL scheme, also based on the LWE assumption. Since PKE is a special case of ABE, ABE-CR-SKL trivially implies PKE-CR-SKL. We first present the PKE-CR-SKL scheme separately to provide a clearer foundation for understanding the ABE-CR-SKL construction.

     \item \emph{PKE and ABE with collusion-resistant SKL and
         classical certificates:} We also propose an IND-KLA secure
         ABE-CR-SKL scheme that utilizes \emph{classical} certificates, whereas the constructions above rely on quantum
         certificates.\footnote{Quantum decryption keys function as quantum certificates.} In this model, a classical certificate can be derived from a leased quantum decryption key, and successful verification guarantees security. Our scheme is based on multi-input ABE (MI-ABE), which is a potentially weaker assumption than collusion-resistant PKFE. We specify required properties for MI-ABE and discuss the relationship between MI-ABE and other primitives in \cref{sec:def-mi-abe}. ABE-CR-SKL with classical certificates trivially implies PKE-CR-SKL with classical certificates.
\end{enumerate}

We introduce fascinating techniques to achieve our results. These
include the classical decryption property and the notion of key-testability for secret key encryption with collusion-resistant SKL (SKE-CR-SKL) and secret key functional encryption with collusion-resistant SKL (SKFE-CR-SKL). Then, we transform SKE-CR-SKL (resp. SKFE-CR-SKL) into PKE-CR-SKL (resp. ABE-CR-SKL) by using classical ABE and compute-and-compare obfuscation~\cite{FOCS:WicZir17,FOCS:GoyKopWat17}. In these transformations, we use decryption keys of SKE-CR-SKL (resp. SKFE-CR-SKL) as key attributes of ABE in a superposition way.
See~\cref{sec:technical_overview} for the details.

\begin{table*}[!t]
\setlength\tabcolsep{0.5eM}
\begin{center}
\begin{minipage}[c]{\textwidth} \scriptsize
\begin{center}
 \begin{threeparttable}

\caption{{\scriptsize Comparison of SKL. VO means security in the presence of the verification oracle for certificates.
}
}
\label{tbl:comparison_PKE-SKL}
\begin{tabular}{rccccc}\toprule
 & Primitive & Collusion-resistant SKL & VO & Certificate & Assumption\\
\midrule
\cite{EC:AKNYY23} & PKE & $-$ & $\checkmark$ & quantum  & PKE\\
\cite{EC:AKNYY23} & ABE & bounded & $\checkmark$ & quantum  & ABE\\
 \cite{EC:AKNYY23} & bCR PKFE\tnote{a} & bounded & $\checkmark$ & quantum  & PKE \\
  \cite{EC:AKNYY23} & PKFE & $\checkmark$ & $\checkmark$ & quantum & PKFE\tnote{b} \\
\cite{EC:BGKMRR24} & PKFE\tnote{c} & $\checkmark$ & $\checkmark$\tnote{d} & classical  & IO \\
 \cite{TCC:AnaPorVai23,TCC:AnaHuHua24} & PKE (FHE) & $-$ &  $-$ & classical & LWE\\
 \cite{EPRINT:CGJL23} & PKE (FHE) & $-$ & $-$ & classical  & LWE\\
 \cite{myEPRINT:KitMorYam24} & PKE & $-$ & $\checkmark$\tnote{e}  & classical  & PKE\\
 \midrule
  Ours1 & PKE & \colorbox[rgb]{1,0.9,0}{$\checkmark$}&  \colorbox[rgb]{1,0.9,0}{$\checkmark$} & quantum  &\colorbox[rgb]{1,0.9,0}{LWE}\\
 Ours2 & ABE\tnote{c} & \colorbox[rgb]{1,0.9,0}{$\checkmark$}  & \colorbox[rgb]{1,0.9,0}{$\checkmark$} & quantum & \colorbox[rgb]{1,0.9,0}{LWE}\\
  Ours3 & PKE & \colorbox[rgb]{1,0.9,0}{$\checkmark$} & \colorbox[rgb]{1,0.9,0}{$\checkmark$} &  \colorbox[rgb]{1,0.9,0}{classical}  & MI-ABE\\
  Ours4 & ABE\tnote{c} & \colorbox[rgb]{1,0.9,0}{$\checkmark$}  & \colorbox[rgb]{1,0.9,0}{$\checkmark$} & \colorbox[rgb]{1,0.9,0}{classical} & MI-ABE\\  
\bottomrule
\end{tabular}
 \begin{tablenotes}[flushleft,online,normal] 
 \item[a] bCR denotes bounded collusion-resistant. This scheme only satisfies the bounded collusion resistance of standard PKFE.
 \item[b] Collusion-resistant PKFE implies IO up to sub-exponential security loss.
 \item[c] These schemes are selectively secure, where adversaries must
     declare the target plaintexts (PKFE case) or attributes (ABE case) at the beginning.\nikhil{Same 'c' is used for ABE-CR-SKL and Bartusek et al. Our selective-security means ciphertext attribute is declared right?}\ryo{Yes. \cite{EC:BGKMRR24} considers FE, so plaintexts are declared. Ours considers ABE, so I added ``attributes''.}
 \item[d] This scheme has public verifiability (the verification key for certificates is public).
 \item[e] This scheme is secure even if the verification key is revealed to the adversary \emph{after} the adversary outputs a valid certificate.
 \end{tablenotes}
 \end{threeparttable}
\end{center}
\end{minipage}
\end{center}
\end{table*}

\subsection{Related Work}\label{sec:related_work}
\paragraph{Encryption with SKL.}
Agrawal et al.~\cite{EC:AKNYY23} presented PKE-SKL, ABE-SKL, and PKFE-SKL schemes.
Their PKE-SKL scheme is based on standard PKE, and its certificate is quantum. This scheme is not collusion-resistant.
Their ABE-SKL scheme is based on PKE-SKL and standard (collusion-resistant) ABE, and its certificates are ones of the underlying PKE-SKL.
This scheme is bounded collusion-resistant, that is, the adversary can
obtain an a-priori bounded number of quantum decryption keys that can decrypt target ciphertexts.
Their PKFE-SKL scheme is based on PKE-SKL and standard (collusion-resistant) PKFE, and its certificates are ones of the underlying PKE-SKL.
This scheme is collusion-resistant. (If we instantiate the PKFE-SKL scheme with bounded collusion-resistant PKFE, the scheme is bounded collusion-resistant with respect to both standard PKFE and SKL.)
All schemes mentioned above are secure in the presence of the verification oracle for certificates.

Bartusek et al.~\cite{EC:BGKMRR24} also presented a PKFE-SKL scheme.
Their scheme is collusion-resistant in the selective model, where the target plaintext is declared at the beginning of the game. In addition, their certificates are classical and publicly verifiable. However, their scheme relies on IO.

Ananth, Poremba, and Vaikuntanathan~\cite{TCC:AnaPorVai23} presented a PKE-SKL scheme with classical certificates based on the LWE assumption and an unproven conjecture.
Since the encryption algorithm of this scheme is essentially the same as the Dual-Regev PKE scheme~\cite{STOC:GenPeiVai08}\nikhil{Needs citation?}\ryo{done.}, their scheme achieves fully homomorphic encryption (FHE) with SKL. Later, Ananth, Hu, and Huang~\cite{TCC:AnaHuHua24} present a new security analysis for Ananth et al.'s scheme and removed the conjecture.
Chardouvelis, Goyal, Jain, and Liu~\cite{EPRINT:CGJL23} presented a PKE-SKL scheme with classical certificates based on the LWE assumption. Their scheme is based on the Regev PKE scheme~\cite{JACM:Regev09}\nikhil{Needs citation?}\ryo{done.}, so it also achieves FHE with SKL. In addition, their scheme also achieves classical communication, where all messages between senders and receivers are classical.
Kitagawa, Morimae, and Yamakawa~\cite{myEPRINT:KitMorYam24} presented
a PKE-SKL scheme with classical certificates based on PKE.
These three works do not achieve collusion resistance.
Ananth et al.~\cite{TCC:AnaPorVai23,TCC:AnaHuHua24} and Chardouvelis et al.~\cite{EPRINT:CGJL23} do not consider the verification oracle in their security definitions, while Kitagawa et al.~\cite{myEPRINT:KitMorYam24} does.\footnote{More precisely, the work considers adversaries that receive a verification key after they output a valid certificate. Kitagawa et al. show that their scheme can be converted to satisfy IND-KLA by Agrawal et al.~\cite{EC:AKNYY23}.}
We summarize the works on encryption with SKL in~\cref{tbl:comparison_PKE-SKL}.

\paragraph{Secure software leasing.}
Ananth and La Placa~\cite{EC:AnaLaP21} introduced secure software
leasing, which encodes classical programs into quantum programs that
we can securely lease. We can view SKL as secure software leasing for
decryption functions. However, previous works on secure software
leasing consider a sub-class of evasive
functions~\cite{EC:AnaLaP21,ARXIV:ColMajPor20,TCC:KitNisYam21,TCC:BJLPS21}
or PRFs~\cite{TCC:KitNisYam21}, which do not support decryption
functions. Moreover, they consider a weak security model in which
pirated programs use \emph{honest} evaluation
algorithms~\cite{EC:AnaLaP21,TCC:KitNis23,TCC:BJLPS21} or rely on the
quantum random oracle model~\cite{ARXIV:ColMajPor20}. Bartusek et
al.~\cite{EC:BGKMRR24} achieve secure software leasing supporting all
differing inputs circuits\footnote{Roughly speaking, a pair of
circuits $(C_0,C_1)$ is differing input if it is hard to find an input
$y$ such that $C_0(y)\ne C_1(y)$.} in a strong security model where
pirated programs can use arbitrary evaluation algorithms. However, their scheme relies on IO.

\paragraph{Multi-input ABE.}
Roughly speaking, MI-ABE is ABE that can support multiple ciphertext-attributes (or multiple key-attributes).
To achieve our classical certificates scheme, we need MI-ABE for polynomial-size circuits where the number of slots is polynomial, and we can generate ciphertexts for one slot using a public key. See~\cref{def:miabe} for the definition.
However, as we review previous works on MI-ABE for general circuits\footnote{Francati, Fior, Malavolta, and Venturi~\cite{EC:FFMV23} and Agrawal, Tomida, and Yadav~\cite{C:AgrTomYad23} proposed MI-ABE for restricted functionalities.} below, none of them achieves what we need without using IO.\footnote{IO implies multi-input functional encryption~\cite{EC:GGGJKL14}, which implies MI-ABE.}

Agrawal, Yadav, and Yamada~\cite{C:AgrYadYam22} proposed \emph{two-input} ABE for polynomial-size circuits based on lattices. However, the scheme is heuristic (no reduction-based security proof) and needs a master \emph{secret key to generate ciphertexts for all slots}.
Agrawal, Rossi, Yadav, and Yamada~\cite{C:ARYY23} proposed MI-ABE for
$\NCone$ from the evasive LWE assumption and MI-ABE for
polynomial-size circuits from the evasive LWE and tensor LWE
assumptions. Their scheme allows to generate ciphertexts for one slot using a public key. However, the number of slots is \emph{constant}.
Agrawal, Kumari, and Yamada~\cite{myEPRINT:AgrKumYam24a} proposed MI-ABE\footnote{Precisely speaking, their scheme is predicate encryption, which satisfies privacy for attributes.} for polynomial-size circuits based on the evasive LWE assumption.
In their scheme, the number of slots is polynomial. However, we need a master \emph{secret key to generate ciphertexts for all input-attributes}.

\paragraph{Broadcast encryption.}
Broadcast encryption~\cite{C:FiaNao93} enables a sender to generate ciphertexts intended for a specific subset of users.
Only the designated users can decrypt the ciphertexts, while even if all other users collude, they cannot recover the message. A key performance metric for broadcast encryption is the size of the public key and ciphertexts.
Several works have proposed optimal broadcast encryption schemes, where these sizes are $\poly(\log{N})$ and $N$ is the total number of users.
The constructions by Agrawal and Yamada~\cite{EC:AgrYam20} and Agrawal, Wichs, and Yamada~\cite{TCC:AgrWicYam20} rely on the LWE assumption and \emph{pairings}, which are not post-quantum secure. The construction by Wee~\cite{EC:Wee22} relies on the \emph{evasive LWE assumption}, which is a non-falsifiable assumption and has known counterexamples~\cite{AC:BrzUnaWoo24}. The construction by Brakerski and Vaikuntanathan~\cite{ITCS:BraVai22} relies on \emph{heuristics} and lacks a reduction-based proof.

While broadcast encryption is particularly well-suited for streaming services, PKE-CR-SKL can also be applied in this domain. Each approach has its own advantages and limitations. One advantage of broadcast encryption is that it allows the sender to specify recipients at the encryption phase, whereas PKE-CR-SKL does not.
However, PKE-CR-SKL offers three notable advantages over optimal broadcast encryption.
\begin{enumerate}
    \item \textbf{Efficient revocation:} Since decryption keys in broadcast encryption are classical, the system requires maintaining a revocation list, which senders use to revoke users, whereas PKE-CR-SKL eliminates the need for such lists.
    \item \textbf{Seamless user expansion:} In broadcast encryption, the user set is fixed during the setup phase, meaning that adding new users requires updating the encryption key. In contrast, PKE-CR-SKL allows new users to be added without requiring any key updates.
    \item \textbf{Weaker cryptographic assumptions:} Our PKE-CR-SKL scheme is based on the standard LWE assumption, whereas post-quantum secure optimal broadcast encryption relies on the evasive LWE assumption, which has counterexamples~\cite{AC:BrzUnaWoo24}.
\end{enumerate}
In terms of  asymptotic efficiency, both optimal broadcast encryption and PKE-CR-SKL achieve the same public key and ciphertext sizes.

\ifnum\submission=1
We discuss more related works in~\cref{sec:other-rel}.
\else

\ifnum\submission=1
\section{More on Related Work} 
\label{sec:other-rel}
\else
\fi

\paragraph{Certified deletion.}
Broadbent and Islam~\cite{TCC:BroIsl20} introduced encryption with certified deletion, where we can generate classical certificates to guarantee that \emph{ciphertexts} were deleted. Subsequent works improved Broadbent and Islam's work and achieved advanced encryption with certified deletion~\cite{AC:HMNY21,ITCS:Poremba23,C:BarKhu23,EC:HKMNPY24,EC:BGKMRR24} and publicly verifiable deletion~\cite{AC:HMNY21,EC:BGKMRR24,TCC:KitNisYam23,TCC:BKMPW23}.
Compute-and-compare obfuscation with certified deletion introduced by Hiroka et al.~\cite{EC:HKMNPY24} is essentially the same as secure software leasing in the strong security model.

\paragraph{Single decryptor encryption.}
Georgiou and Zhandry~\cite{EPRINT:GeoZha20} introduced the notion of SDE. They constructed a public-key SDE scheme from one-shot signatures~\cite{STOC:AGKZ20} and extractable witness encryption with quantum auxiliary information~\cite{STOC:GGSW13,C:GKPVZ13}. Coladangelo, Liu, Liu, and Zhandry~\cite{C:CLLZ21} constructed a public-key SDE scheme from IO~\cite{JACM:BGIRSVY12} and extractable witness encryption or from subexponentially secure IO, subexponentially secure OWF, and LWE by combining the results by Culf and Vidick~\cite{Quantum:CulVid22}. Kitagawa and Nishimaki~\cite{AC:KitNis22} introduced the notion of single-decryptor functional encryption (SDFE), where each functional decryption key is copy-protected and constructed single decryptor PKFE for $\Ppoly$ from the subexponential hardness of IO and LWE.
These works consider the setting where the adversary receives only one copy-protected decryption key.
Liu, Liu, Qian, and Zhandry~\cite{TCC:LLQZ22} study SDE in the collusion-resistant setting, where the adversary receives multiple copy-protected decryption keys.
They constructed a public-key SDE scheme with bounded collusion-resistant copy-protected keys from subexponentially secure IO and subexponentially secure LWE.

\paragraph{Multi-copy revocable encryption.}
Ananth, Mutreja, and Poremba~\cite{myEPRINT:AnaMutPor24} introduced multi-copy revocable encryption.
This notion considers the setting where we can revoke \emph{ciphertexts} (not decryption keys), and the adversary receives multiple copies of the target \emph{ciphertext} (they are pure states). Hence, this notion is different from secure key leasing (or key-revocable cryptography).

\fi



\ifnum\submission=1
\subsection{Organization of the paper}
In~\cref{sec:technical_overview}, we present a technical overview.
In~\cref{sec:prelim}, we define the notation and preliminaries that we require in this work.  In~\cref{sec:cr-skl-defs}, we define the notion of secret key and public key encryption with collusion-resistant secure key leasing (SKE-CR-SKL and PKE-CR-SKL) and its security notions. In~\cref{sec:SKECRSKL-KT}, we construct an SKE-CR-SKL scheme with key testability. In~\cref{sec:PKE-CR-SKL}, we provide our PKE-CR-SKL scheme. Due to the page limitation, we present our ABE-CR-SKL and classical certificates ABE-CR-SKL schemes in~\cref{sec:ABE-SKL,sec:ABR-CR2-SKL}, respectively.
\else
\fi


\section{Technical Overview}\label{sec:technical_overview}

\newcommand{\qDel}{\qalgo{Del}}
\newcommand{\cnc}[2]{\mathbf{CC}[#1,#2]}
\newcommand{\lock}{\mathsf{lock}}
\newcommand{\CCObf}{\mathsf{CC}.\mathsf{Obf}}
\newcommand{\CCSim}{\mathsf{CC}.\mathsf{Sim}}
\newcommand{\tlP}{\widetilde{P}}
\newcommand{\CDSKE}{\mathsf{SKECD}}
\newcommand{\SKECD}{\mathsf{SKECD}}
\newcommand{\CDec}{\mathsf{CDec}}
\newcommand{\ctlen}{\ell_{\ct}}
\newcommand{\NMSKECD}{\mathsf{NMSKECD}}
\newcommand{\sig}{\mathsf{sig}}
\newcommand{\sgn}{\mathsf{sgn}}
\newcommand{\skecd}{\mathsf{skecd}}
\renewcommand{\Check}{\mathsf{Check}}
\newcommand{\msglen}{\ell_{\msg}}
\newcommand{\PKECRSKL}{\mathsf{PKE}\textrm{-}\mathsf{CR}\textrm{-}\mathsf{SKL}}
\newcommand{\SKECRSKL}{\mathsf{SKE}\textrm{-}\mathsf{CR}\textrm{-}\mathsf{SKL}}
\newcommand{\KeyTest}{\mathsf{KeyTest}}

We now provide a technical overview of this work.  
Our primary focus here is on constructing PKE-CR-SKL and ABE-CR-SKL based on the LWE assumption.  
For an overview of our variants with classical certificates based on MI-ABE, please refer to \cref{sec-ABE-CR-SKL-classical-certificate}. 

\subsection{Defining PKE with Collusion-Resistant SKL}

We will begin by describing the definition of PKE with Secure
Key-Leasing (PKE-SKL) \cite{EC:AKNYY23,TCC:AnaPorVai23}, and then get into our collusion-resistant generalization
of it. PKE-SKL is a cryptographic primitive consisting of five
algorithms: $\Setup, \qKG, \Enc, \qDec$ and $\qVrfy$. The algorithm
$\Setup$ samples a public encryption-key $\ek$ and a ``master''
secret-key $\msk$. The encryption algorithm $\Enc$ takes a message
$\msg$ and $\ek$ as inputs and produces a corresponding ciphertext
$\ct$. Both these algorithms are classical and similar to their
counterparts in standard PKE. On the other hand, the key-generation
algorithm $\qKG$ is quantum, and produces a pair of keys $(\qdk, \vk)$
given $\msk$ as input. Here, $\qdk$ is a quantum decryption-key using
which $\qDec$ can decrypt arbitrarily many ciphertexts encrypted under
$\ek$. The other key $\vk$ is a classical verification-key, the
purpose of which will be clear in a moment.

The setting to consider is one where an adversary is given a
decryption-key $\qdk$, and is later asked to return it back.
Intuitively, we wish to guarantee that if $\qdk$ is returned, the
adversary can no longer decrypt ciphertexts
encrypted under $\ek$. Since a malicious adversary may
even send a malformed key $\widetilde{\qdk}$, we need a way to tell
whether the decryption-key has been correctly returned.  This
is the purpose of the algorithm $\qVrfy$, which performs such a check
using the corresponding (private) verification-key $\vk$.  It is
required that if the state $\qdk$ is sent back undisturbed, then
$\qVrfy$ must accept. On the other hand, if $\qVrfy$ accepts (even for
a possibly malformed state), then the adversary must lose the ability
to decrypt. This loss in the ability to decrypt is captured formally
by a cryptographic game, where an adversary is asked to distinguish
between ciphertexts of different messages, after passing the
verification check.

So far, we have described the notion of PKE-SKL. In this work, we
study the notion of PKE with Collusion-Resistant SKL (PKE-CR-SKL).
This primitive has the same syntax as PKE-SKL, but the aforementioned
security requirement is now stronger. Specifically, an adversary that
obtains polynomially-many decryption keys and (verifiably) returns
them all, should no longer be able to decrypt. This is defined formally
in our notion of IND-KLA (Indistinguishability under Key-Leasing
Attacks) security (Definition \ref{def:IND-CPA_PKECRSKL}). As standard
in cryptography, it is characterized by a game between a challenger
$\qCh$ and a QPT adversary $\qA$. An informal description of the game
follows:

\begin{description}
\item \textbf{IND-KLA Game in the Collusion-Resistant Setting}
\begin{enumerate}
\item $\qCh$ samples $(\ek, \msk) \gets \Setup(1^\secp)$ and sends
$\ek$ to $\qA$.
\item Then, $\qA$ requests $q$ decryption keys corresponding to $\ek$,
where $q$ is some arbitrary polynomial in $\secp$.\footnote{Without loss of generality, we can assume that $\qA$ asks for sufficiently many keys at the beginning itself. Hence, even adversaries that can request additional keys after accessing the verification oracle are covered by this definition.}
\item $\qCh$ generates $(\qdk_i, \vk_i)_{i \in [q]}$ using $q$
independent
invocations of $\qKG(\msk)$. It sends $\{\qdk_i\}_{i\in[q]}$ to $\qA$.

\item Corresponding to each index $i \in [q]$, $\qA$ is allowed
oracle access to the algorithm $\qVrfy(\vk_i, \cdot)$ in the following
sense: For a quantum state $\widetilde{\qdk}$ queried by $\qA$, the
oracle evaluates $\qVrfy(\vk_i, \widetilde{\qdk})$ and measures the
verification result. It then returns the obtained classical outcome
(which indicates accept/reject). We emphasize that $\qA$ is allowed to
interleave queries corresponding to indices $i \in [q]$, and can also
make its queries adaptively.

\item If at-least one query of $\qA$ to $\qVrfy(\vk_i, \cdot)$
produces an accept output for every $i \in [q]$, the game proceeds to
the challenge phase. Otherwise, the game aborts.
\item In the challenge phase, $\qA$ specifies a pair of messages
$(\msg_0, \msg_1)$. $\qCh$ sends $\msg_\coin$ to $\qA$ for a random
bit $\coin$.
\item $\qA$ outputs $\coin'$.
\end{enumerate}
\end{description}

$\qA$ wins the game whenever $\coin = \coin'$. The security
requirement is that for every QPT adversary, the winning probability
conditioned on the game not aborting is negligibly close to $1/2$.
Observe that the
adversary is allowed to make polynomially-many attempts in order to
verifiably return a decryption-key $\qdk_i$, and unsuccessful attempts
are not penalized.  We also emphasize that in the definition, $q$ is
an unbounded polynomial, i.e., the construction is not allowed to
depend on $q$ in any way.

\subsection{Insecurity of Direct Extensions of Prior Work}

Next, we will provide some intuition regarding the insecurity of
direct extensions of prior works to this stronger setting.

Let us consider the PKE-SKL scheme due to Agrawal et al.
\cite{EC:AKNYY23} for demonstration. The decryption-keys in their
scheme are of the following form\footnote{In actuality, they use a
parallel repetition to amplify security.}:

$$\frac{1}{\sqrt{2}}\big(\ket{0}\ket{\sk_0} + \ket{1}\ket{\sk_1}\big)$$

Here, $\sk_0, \sk_1$ are secret-keys corresponding to public-keys
$\pk_0, \pk_1$ respectively of a standard PKE scheme. Specifically,
the pairs $(\pk_0, \sk_0)$ and $(\pk_1, \sk_1)$ are generated using
independent invocations of the setup algorithm of PKE. The PKE-SKL
public-key consists of the pair $(\pk_0, \pk_1)$. The encryption
algorithm outputs ciphertexts of the form $(\ct_0, \ct_1)$, where for
each $i \in \bit$, $\ct_i$ encrypts the plaintext under $\pk_i$.  The
verification procedure requires the adversary to return the
decryption-key, and checks if it is the same as the above state. The intuition is that if the adversary retains the ability to decrypt, then it cannot pass the verification check with probability close to 1. For instance, measuring the state provides the adversary with $\sk_0$ or
$\sk_1$ which is sufficient for decryption, but it clearly destroys
the above quantum state.

Consider now the scenario where the public-key $(\pk_0, \pk_1)$ is
fixed, and $n = \poly(\secp)$ copies of the above decryption-key are
given to an adversary. In this case, it is easy to see that the
adversary can simply measure the states to obtain both $\sk_0$ and
$\sk_1$. Even though this process is destructive, since $\sk_0$ and
$\sk_1$ completely describe the state, the adversary can just recreate
many copies of it and pass the deletion checks. Consequently, one
might be tempted to encode other secret information in the Hadamard
basis by introducing random phases. However, we observe that such
approaches also fail due to simple gentle-measurement attacks.  For
example, consider the combined state of the $n$ decryption-keys in
such a case:

$$\frac{1}{2^{n/2}}\Big(\ket{0\ldots0}\ket{\sk_0 \ldots \sk_0} -
\ket{00\ldots1}\ket{\sk_0\sk_0\ldots\sk_1} + \cdots -
\ket{1\ldots1}\ket{\sk_1\ldots\sk_1}\Big)$$

Clearly, only the last term of the superposition does not contain
$\sk_0$, and the term only has negligible amplitude.
Hence, one can compute $\sk_0$ on another register and measure
it, without disturbing the state more than a negligible amount. As a result, the verification checks can be passed while retaining the ability to decrypt. We note that all existing constructions of encryption with SKL can be broken with similar collusion attacks. 

The reason our scheme does not run into such an attack is because we
rely on the notion of Attribute-Based Encryption (ABE), which
enables exponentially-many secret-keys for every public-key.
Consequently, different decryption-keys can be generated as
superpositions of different ABE secret-keys. However, it is not
clear how this intuition alone can be used to establish security in
a provable manner, and we require additional ideas. We now describe
these ideas at a high-level.

\subsection{Idea Behind the PKE-CR-SKL Scheme}

In order to explain the basic idea behind our PKE-CR-SKL scheme, we
will first introduce a key-building block, a primitive called
SKE-CR-SKL. This is basically a secret-key variant of PKE-CR-SKL,
i.e., the setup algorithm only samples a master secret-key $\ske.\msk$, and
the encryption algorithm encrypts plaintexts under $\ske.\msk$. The
security requirement is similar. An adversary that receives polynomially
many decryption-keys and returns them all, should not be able to
distinguish between ciphertexts of different messages. We refer to
this security notion as one-time IND-KLA (OT-IND-KLA) security
(Definition \ref{def:OT-IND-CPA_SKECRSKL}).
The ``one-time'' prefix refers to the fact that unlike in PKE-CR-SKL,
the adversary does not have the ability to perform chosen plaintext
attacks. In other words, the adversary does not see any ciphertexts
before it is required to return its decryption-keys. Although this is
a weak security guarantee, it suffices for our PKE-CR-SKL scheme. The
description of the PKE-CR-SKL scheme now follows.

The key-generation procedure involves first generating an SKE-CR-SKL
decryption-key, which is represented in the computational basis as
$\ske.\qdk = \sum_u \alpha_u \ket{u}$.
Actually, our SKE-CR-SKL scheme needs to
satisfy another crucial property, which we call the classical
decryption property. This property requires the existence of a
classical deterministic algorithm $\CDec$ with the following
guarantee.  For any SKE-CR-SKL ciphertext $\ske.\ct$, $\CDec(u, \ske.\ct)$
correctly decrypts $\ske.\ct$ for every string $u$ in the superposition of
$\ske.\qdk$. Our construction exploits this fact with the help of an
ABE scheme as follows.

The actual decryption key is generated as $\qdk \seteq \sum_u
\alpha_u \ket{u} \otimes \ket{\abe.\sk_u}$ where $\abe.\sk_u$ is an ABE secret-key
corresponding to the key-attribute $u$. The idea is to now have the
encryption algorithm encrypt the plaintext using the ABE scheme,
under a carefully chosen ciphertext-policy. Specifically, we wish to
embed an SKE-CR-SKL ciphertext $\ske.\ct^*$ as part of the
policy-circuit, such that the outcome of $\CDec(u, \ske.\ct^*)$
determines whether the key $\abe.\sk_u$ can decrypt the ABE
ciphertext or not. This allows us to consider two ciphertexts
$\ske.\ct_0^*, \ske.\ct_1^*$ of different plaintexts, such that when
$\ske.\ct_0^*$
is embedded in the policy, then every ABE key $\abe.\sk_u$ satisfies
the ABE relation. On the other hand, no ABE key satisfies the relation
when $\ske.\ct_1^*$ is
embedded. Observe that in the former case, $\qdk$ allows for
decryption without disturbing the state (by gentle measurement) while
in the latter case, the security of ABE ensures that no adversary can
distinguish ciphertexts of different messages\footnote{This requires
that the ABE scheme is secure even given superposition access to the
key-generation oracle.}. Hence, if we were to undetectably switch the
policy-circuit from one corresponding to $\ske.\ct_0^*$ to one with
$\ske.\ct_1^*$, we are done.

While we cannot argue this directly, our main idea is that this switch
is undetectable using OT-IND-KLA security, given that all the
SKE-CR-SKL keys $\ske.\qdk$ are returned. This can be enforced because
all the leased decryption keys $\qdk$ are required to be returned.
Consequently, the verification procedure uncomputes the ABE secret-key
register of the returned keys, followed by verifying all the obtained
SKE-CR-SKL keys $\ske.\qdk'$. However, there is one problem with the
template as we have described it so far. This is that both
$\ske.\ct_0^*, \ske.\ct_1^*$ are SKE-CR-SKL ciphertexts, which
inherently depend on the master secret-key. In order to remove this
dependence and achieve public-key encryption, the actual encryption
algorithm uses a dummy ciphertext-policy $\tlC \gets \CCSim(1^\secp)$
where $\CCSim$ is the simulator of a compute-and-compare obfuscation
scheme, a notion we will explain shortly.  We will then rely on the
security of this obfuscation scheme to switch $\tlC$ in the security proof, to
an appropriate obfuscated circuit with $\ske.\ct_0^*$ embedded in
it.

In more detail, the ABE scheme allows a key with attribute $u$ to
decrypt if and only if the ciphertext policy-circuit $\tlC$ satisfies
$\tlC(u) = \bot$. Observe that in the construction, $\tlC$ is
generated as $\tlC \gets
\CCSim(1^\secp)$, which outputs $\bot$ on every input $u$. Consider
now a circuit $\tlC^*$ that is an obfuscation of the circuit
$\cnc{D[\ske.\ct^*_0]}{\lock, 0}$ which is described as follows:

\begin{description}
\item {\bf Description of $\cnc{D[\ske.\ct^*_0]}{\lock, 0}:$}
\begin{itemize}
\item $\ske.\ct_0^*$ is an SKE-CR-SKL encryption of the (dummy) message $0^\secp$.
\item $D[\ske.\ct_0^*]$ is a circuit with $\ske.\ct^*_0$ hardwired. It
is defined as $D[\ske.\ct_0^*](x) = \CDec(x, \ske.\ct_0^*)$.
\item $\lock$ is a value chosen uniformly at random, independently of all
other values.
\item On input $x$, the circuit outputs $0$ if $D[\ske.\ct_0^*](x) = \lock$. Otherwise, it outputs $\bot$.
\end{itemize}
\end{description}

The above circuit belongs to a sub-class of circuits known as
compute-and-compare circuits. Recall that our goal was to avoid the use
of IO. We can get away with using IO for this sub-class of
circuits, as these so-called compute-and-compare obfuscation schemes are known
from LWE \cite{FOCS:GoyKopWat17,FOCS:WicZir17}.

The security of the obfuscation can now be used to argue that the
switch from $\tlC$ to $\tlC^*$ is indistinguishable. Note that to
invoke this security guarantee, $\lock$ must be a
uniform value that is independent of all other values. Next, we
can rely on the OT-IND-KLA security of SKE-CR-SKL to switch the ciphertext
$\ske.\ct_0^*$ embedded in the circuit $D$ to some other ciphertext
$\ske.\ct_1^*$. The switch would be indistinguishable given that the
SKE-CR-SKL decryption-keys are revoked, which we can enforce as
mentioned previously. Crucially, we will generate $\ske.\ct_1^*$ as an
encryption of the value $\lock$ corresponding to the above
compute-and-compare circuit.\footnote{Although the compute-and-compare
circuit depends on $\lock$ in this hybrid, the switch is still
justified by OT-IND-KLA security.} This ensures that for every attribute $u$
in the superposition of an SKE-CR-SKL decryption-key $\ske.\qdk =
\sum_u \alpha_u \ket{u}$, the algorithm $\CDec(u, \ske.\ct_1^*)$
outputs the value $\lock$. As a consequence, the circuit $\tlC^*$
will output $0$ instead of $\bot$, meaning the key $\abe.\sk_u$ does
not satisfy the relation in this hybrid, as desired.

It will be clear in the next subsection that our SKE-CR-SKL scheme is implied
by OWFs. Since compute-and-compare obfuscation is known from LWE
\cite{FOCS:GoyKopWat17,FOCS:WicZir17}, and so is Attribute-Based Encryption for
polynomial-size circuits
\cite{EC:BGGHNS14}\footnote{We show that their ABE scheme is secure
even with superposition access to the key-generation oracle.}, we have
the following theorem:

\begin{theorem}
There exists a PKE-CR-SKL scheme satisfying IND-KLA security, assuming
the polynomial hardness of the LWE assumption.
\end{theorem}

We also observe that using similar ideas as in PKE-CR-SKL, one can
obtain an analogous notion of selectively-secure ABE with
collusion-resistant secure-key
leasing (ABE-CR-SKL) for arbitrary polynomial-time computable circuits. Intuitively, this
primitive allows the adversary to declare a target ciphertext-attribute and then make arbitrarily many key-queries
adaptively, even ones which satisfy the ABE relation. Then, as long as
the adversary verifiably returns all the keys that satisfy the
relation, it must lose the ability to decrypt. This is captured
formally in the notion of selective IND-KLA security (Definition
\ref{def:sel_ind_ABE_SKL}). To realize this primitive, we introduce a
notion called secret-key functional-encryption with
collusion-resistant secure key-leasing
(SKFE-CR-SKL), which is a functional encryption analogue of SKE-CR-SKL.
From this primitive, we require a notion called selective
single-ciphertext security (Definition \ref{def:sel-1ct-SKFE}), that
is similar to the OT-IND-KLA security of SKE-CR-SKL. Like SKE-CR-SKL,
we observe that SKFE-CR-SKL is also implied by OWFs.
The other elements of the ABE-CR-SKL construction are the
same as in PKE-CR-SKL, namely compute-and-compare obfuscation and an
ABE scheme. Consequently, we have the following theorem:

\begin{theorem}
There exists an ABE-CR-SKL scheme satisfying selective IND-KLA security, assuming
the polynomial hardness of the LWE assumption.
\end{theorem}

\subsection{Constructing SKE-CR-SKL}

Next, we will describe how we realize the aforementioned
building-block of SKE-CR-SKL satisfying the classical decryption
property (Definition \ref{def:CDEC_SKE-CR-SKL}). Our construction will
make use of a BB84-based secret-key encryption with certified-deletion
(SKE-CD) scheme, a brief description of which is as follows. This is
an encryption scheme where the ciphertexts are quantum BB84 states
\cite{wiesner1983conjugate,BB84}.  Given such a ciphertext, an
adversary is later asked to provide a certificate of deletion. If a
certificate is provided and verified to be correct, it is guaranteed
that the adversary learns nothing about the plaintext even if the
secret-key is later revealed.  Crucially, we require that
the SKE-CD scheme also satisfies a classical decryption property
(Definition \ref{def:bb84}). Intuitively, this requires that if the ciphertext
is of the form $\skecd.\qct = \sum_u \alpha_u \ket{u}$, then every string $u$ in the
superposition can be used to decrypt correctly. Specifically, there
exists a classical deterministic algorithm $\SKECD.\CDec$ such that
$\SKECD.\CDec(\skecd.\sk, u)$ correctly decrypts $\skecd.\qct$, where
$\skecd.\sk$ is the secret-key.

Let us now recall the functionality offered by SKE-CR-SKL.
The encryption algorithm $\Enc$ takes as input a master secret-key
$\ske.\msk$
and a plaintext $\msg$ and outputs a classical ciphertext
$\ske.\ct$. The
key-generation algorithm $\qKG$ takes as input $\ske.\msk$ and produces a
quantum decryption key $\ske.\qdk$ along with a corresponding
verification-key $\ske.\vk$. Decryption of $\ske.\ct$ can be performed by $\qDec$ using
$\ske.\qdk$ without disturbing the state by more than a negligible amount. Furthermore, an
adversary that receives $q$ (unbounded polynomially many)
decryption-keys can be asked to return all of them, before which it
does not get to see any ciphertext encrypted under $\ske.\msk$. Each
returned key can be verified using the verification algorithm and
the corresponding verification key. If all $q$ keys are
verifiably returned, then it is guaranteed that the adversary cannot
distinguish a pair of ciphertexts (of different messages) encrypted
under $\ske.\msk$. This requirement, termed as OT-IND-KLA security, is
captured formally in Definition \ref{def:OT-IND-CPA_SKECRSKL}. The
intuition behind the construction is now described as follows:

Let us begin with the simple encryption algorithm. $\Enc(\ske.\msk, \msg)$
produces a classical output $\ske.\ct= (\skecd.\sk, z \seteq \msg \xor r)$, where
$\skecd.\sk$ and $r$ are values specified by $\ske.\msk$. The value
$\skecd.\sk$ is a secret key of an SKE-CD scheme $\SKECD$, while $r$
is a string chosen uniformly at random. Clearly, one can retrieve
$\msg$ from $\ske.\ct$ given $r$. Consequently, each decryption key
$\ske.\qdk$
is essentially an $\SKECD$ encryption of $r$.  The idea is that the
secret-key $\skecd.\sk$ can be obtained from $\ske.\ct$, which can then be
used to retrieve $r$ from $\ske.\qdk$. As a consequence,
collusion-resistance (OT-IND-KLA security) can be argued based on the
security of $\SKECD$ by utilizing the following observations:

\begin{itemize}
\item Each decryption-key contains an $\SKECD$ encryption of $r$ using
independent randomness.
\item The adversary must return all the decryption keys (containing
$\SKECD$ ciphertexts) before it receives the challenge ciphertext
(containing the $\SKECD$ secret-key).
\end{itemize}
Furthermore, since $\ske.\qdk$ is essentially an
$\SKECD$ ciphertext, the classical decryption property of SKE-CR-SKL
follows easily from the analogous classical decryption property of
BB84-based SKE-CD (Definition \ref{def:bb84}).

\subsection{Handling Verification Queries}

In our previous discussion about the idea behind the PKE-CR-SKL
scheme, we left out some details regarding the following. We did not discuss how the key-generation oracle of the ABE scheme can be used to
simulate the adversary's view, in the hybrid where $\ske.\ct_1^*$ is
embedded in the circuit $\tlC$. Firstly, we note that the ABE scheme
can handle superposition key-queries, which we establish by a
straightforward argument about the LWE-based ABE scheme of Boneh et
al.
\cite{EC:BGGHNS14}.
Recall now that in this hybrid, for every leased
decryption-key $\qdk = \sum_u \alpha_u \ket{u} \otimes
\ket{\abe.\sk_u}$, each
$\abe.\sk_u$ can be obtained (in superposition) by querying $\ske.\qdk
= \sum_u
\alpha_u \ket{u}$ to the oracle of ABE. However, we observe that responses to
verification queries made by the adversary are not so straightforward
to simulate. This is because for each verification query, the ABE
secret-key register needs to be uncomputed, for which we will again
rely on the ABE key-generation oracle. Specifically, the problem is
that for a malformed key $\widetilde{\qdk}$, there may exist some
$\widetilde{u}$ in the superposition which actually satisfies the ABE
relation. Recall that by definition of this relation, it follows that
$\CDec(\widetilde{u}, \ske.\ct_1^*)$ incorrectly decrypts the 
ciphertext $\ske.\ct_1^*$.
To fix this issue, we upgrade our SKE-CR-SKL
scheme to satisfy another property called key-testability. This
involves the existence of a classical algorithm $\KeyTest$ which
accepts or rejects. The property requires that when $\KeyTest$ is
applied in superposition followed by post-selecting on it accepting,
every $\widetilde{u}$ in the superposition decrypts correctly. As a
result, we are able to apply this key-testing procedure to the
register of the SKE-CR-SKL key, followed by simulating the adversary's
view using the ABE oracle. Note that we will now consider
$\qKG(\ske.\msk)$ to output an additional testing-key $\ske.\tk$ along
with $\ske.\qdk$ and $\ske.\vk$. The key-testability requirements are
specified in more detail as follows:


\begin{itemize}

\item \textbf{Security:} There exists an algorithm $\KeyTest$ such
that no QPT adversary can produce a classical value $\dk$ and a message
$\msg$ such that:
\begin{itemize}
\item $\CDec(\dk, \ske.\ct) \neq \msg$, where $\ske.\ct \gets
\Enc(\ske.\msk, \msg)$.
\item $\KeyTest(\ske.\tk, \dk) = 1$.
\end{itemize}

\item \textbf{Correctness:} 
If $\KeyTest$ is applied in superposition to $\ske.\qdk$ and the output is
measured to obtain outcome $1$, $\ske.\qdk$ should be almost
undisturbed.
\end{itemize}

\subsection{Upgrading SKE-CR-SKL with Key-Testability}

We will now explain how the SKE-CR-SKL construction is modified to
satisfy the aforementioned key-testability property. First, we mention
a classical decryption property of $\SKECD$ (Definition
\ref{def:bb84}) that we crucially rely on. A ciphertext $\skecd.\qct$
(corresponding to some message $\msg$) of $\SKECD$ is essentially a
BB84 state $\ket{x}_\theta$. The property guarantees the existence of
an algorithm $\SKECD.\CDec$ such that for any string $u$ that matches
$x$ at all computational basis positions specified by $\theta$,
$\SKECD.\CDec(\skecd.\sk, u) = \msg$ with overwhelming probability.
Recall now that in our SKE-CR-SKL construction, ciphertexts are of the
form $\ske.\ct = (\skecd.\sk, z = r \xor \msg)$ and decryption-keys
are essentially $\SKECD$ encryptions of the value $r$. Consequently,
the algorithm $\CDec$ of $\SKECRSKL$ works as follows:

\begin{description}
\item $\underline{\CDec(u, \ske.\ct)}:$
\begin{itemize}
\item Parse $\ske.\ct = (\skecd.\sk, z)$.
\item Output $z \xor \SKECD.\CDec(\skecd.\sk, u)$.
\end{itemize}
\end{description}

As a result, it is sufficient for us to ensure that no QPT adversary
can output a value $\dk$ such that
$\SKECD.\CDec(\skecd.\sk,\allowbreak \dk)$
produces a value different from $r$. By the aforementioned
classical-decryption property of $\SKECD$, it suffices to bind the
adversary to the computational basis values of a ciphertext
$\skecd.\qct = \ket{x}_\theta$. For this, we employ a technique
reminiscent of the Lamport-signature scheme \cite{Lamport79}. Thereby, an additional ``signature'' register that is entangled with
the $\SKECD$ ciphertext $\skecd.\qct = \ket{x}_\theta$ is utilized. We
note that similar
techniques for signing BB84 states were employed in previous works on
certified deletion \cite{EC:HKMNPY24,TCC:KitNisYam23}.
Specifically, let $\qreg{SKECD.CT_i}$ denote the
register holding the $i$-th qubit of $\skecd.\qct$ and $s_{i,0}, s_{i,1}$ be
randomly chosen pre-images from the domain of an OWF $f$.  Then, the
following map is performed on registers $\qreg{SKECD.CT_i}$ and
$\qreg{S_i}$ where the latter is initialized to $\ket{0 \ldots 0}$:

$$\ket{u_i}_{\qreg{SKECD.CT_i}} \otimes \ket{v_i}_{\qreg{S_i}}
\ra \ket{u_i}_{\qreg{SKECD.CT_i}} \otimes \ket{v_i \xor
s_{i,u_i}}_\qreg{S_i}$$

Let $\rho_i$ be the resulting state. The SKE-CR-SKL decryption-key
$\ske.\qdk$ is
set to be the state $\rho_1 \otimes \ldots \otimes \rho_{\ctlen}$
where $\ctlen$ is the length of $\skecd.\qct$. Thereby, the testing-key
$\ske.\tk$ will consist of the values $f(s_{i, 0}), f(s_{i, 1})$ for each
$i \in [\ctlen]$. Observe now that for a returned (possibly altered)
decryption-key $\widetilde{\qdk}$ (or $\ske.\widetilde{\qdk}$), it is possible to check for each
qubit whether the superposition term $u_i$ is associated with the
correct pre-image $s_{i,u_i}$. This can be done by forward
evaluating the pre-image register and comparing it with the value
$f(s_{i, u_i})$ that is specified in $\ske.\tk$. This is essentially the
$\KeyTest$ algorithm. It is easy to see that this procedure does
not disturb the state when applied to the unaltered
decryption key $\qdk$ (or $\ske.\qdk$). Moreover, observe that the adversary does not
receive the pre-image $s_{i, 1-x[i]}$ for any computational basis
position $i$. Consequently, we show that the adversary cannot
produce a value $\dk$ whose computational basis values are
inconsistent with those of $x$, unless it can invert outputs of $f$.
From the previous discussion, it follows that if the computational
basis values of $\dk$ are consistent with $x$, then $\dk$ cannot
result in the incorrect decryption of $\ske.\ct$.


\section{Preliminaries}\label{sec:prelim}
\ifnum\submission=1
\else
\paragraph{Notations and conventions.}
In this paper, standard math or sans serif font stands for classical algorithms (e.g., $C$ or $\algo{Gen}$) and classical variables (e.g., $x$ or $\keys{pk}$).
Calligraphic font stands for quantum algorithms (e.g., $\qalgo{Gen}$) and calligraphic font and/or the bracket notation for (mixed) quantum states (e.g., $\qstate{q}$ or $\ket{\psi}$).

Let $[\ell]$ denote the set of integers $\{1, \cdots, \ell \}$, $\secp$ denote a security parameter, and $y \seteq z$ denote that $y$ is set, defined, or substituted by $z$.
For a finite set $X$ and a distribution $D$, $x \chosen X$ denotes selecting an element from $X$ uniformly at random, and $x \chosen D$ denotes sampling an element $x$ according to $D$. Let $y \gets \algo{A}(x)$ and $y \gets \qalgo{A}(\qstate{x})$ denote assigning to $y$ the output of a probabilistic or deterministic algorithm $\algo{A}$ and a quantum algorithm $\qalgo{A}$ on an input $x$ and $\qstate{x}$, respectively.
PPT and QPT algorithms stand for probabilistic polynomial-time algorithms and polynomial-time quantum algorithms, respectively.
Let $\negl$ denote a negligible function.
For strings $x,y\in \bit^n$, $x\cdot y$ denotes $\bigoplus_{i\in[n]} x_i y_i$ where $x_i$ and $y_i$ denote the $i$th bit of $x$ and $y$, respectively.
For random variables $X$ and $Y$, we use the notation $X \approx Y$ to denote that these are computationally indistinguishable. Likewise, $X \approx_s Y$ denotes that they are statistically indistinguishable.
\fi

\subsection{One Way to Hiding Lemmas}

\begin{lemma}[O2H Lemma~\cite{C:AmbHamUnr19}]\label{lem:O2Hprev}
Let $G,H:X\ra Y$ be any functions, $z$ be a random value, and $S\subseteq X$ be a random set such that $G(x)=H(x)$ holds for every $x\notin S$.
The tuple $(G,H,S,z)$ may have arbitrary joint distribution.
Furthermore, let $\qA$ be a quantum oracle algorithm that makes at most $q$ quantum queries.
Let $\qB$ be an algorithm such that $\qB^H$ on input $z$ chooses $i\gets[q]$, runs $\qA^H(z)$, measures $\qA$'s $i$-th query, and outputs the measurement outcome.
Then, we have:
\begin{align}
\abs{\Pr[\qA^{H}(z)=1]-\Pr[\qA^{G}(z)=1]} \leq 2q\cdot\sqrt{\Pr[\qB^{H}(z)\in S]}
\enspace.
\end{align}
\end{lemma}

We require a generalization of this lemma, where $\qA$ receives an
additional quantum oracle $\cQ$ in both worlds. Consequently, we
consider $\qB$ to be given oracle access to $\cQ$, which it uses to
simulate the oracle calls of $\qA$ to $\cQ$. Notice that if the
outputs of $\cQ$ were classical, we could have simply defined
augmented oracles $G'$ (likewise $H'$) based on $G$ (likewise $H$) and
$\cQ$.  However, the oracles $\cQ$ we consider will have classical
inputs and quantum state outputs. Consequently, the lemma we require
is stated as follows:

\begin{lemma}[O2H Lemma with Auxiliary Quantum Oracle]\label{lem:O2H}
Let $G, H: X \ra Y$ be any functions, $z$ be a random value, and $S
\subseteq X$ be a random set such that $G(x) = H(x)$ holds for every
$x \notin S$. The tuple $(G, H, S, z)$ may have arbitrary joint
distribution. Furthermore, let $\cQ$ be a quantum oracle that is
arbitrarily correlated with the tuple $(G, H, S, z)$, takes
classical input and produces a (possibly mixed) quantum state as
output. Let $\qA$ be a quantum oracle algorithm that makes at most $q$
quantum queries to the oracles $H$ or $G$, and arbitrarily many
queries to the oracle $\cQ$. Let $\qB$ be an algorithm such that
$\qB^{\cQ, H}$ on input $z$, chooses $i \gets [q]$, runs $\qA^{\cQ,
H}(z)$, measures $\qA$'s $i$-th query to $H$, and outputs the measurement
outcome. Then, we have:

\begin{align}
\abs{\Pr[\qA^{\cQ, H}(z)=1]-\Pr[\qA^{\cQ, G}(z)=1]} \leq
2q\cdot\sqrt{\Pr[\qB^{\cQ, H}(z)\in S]}
\enspace.
\end{align}

\end{lemma}

\ifnum\submission=1
We prove this lemma in~\cref{appsec:O2H_oracle}.
\else

\ifnum\submission=1
\section{Proof of O2H Lemma with Auxiliary Quantum Oracle}\label{appsec:O2H_oracle}
\else
\fi
\begin{proof}[Proof of~\cref{lem:O2H}]
Let us consider an adversary $\widetilde{\qA}$ that receives as input
the description $\langle \cQ\rangle$ of $\cQ$\footnote{The O2H Lemma (Lemma \ref{lem:O2Hprev}) holds even if $z$ is exponentially large, so the description of $\cQ$ need not be concise.}, along with the input
$z$ used by $\qA$. Given oracle access to $H$ (likewise $G$) and $(z,
\langle \cQ\rangle)$ as input, $\widetilde{\qA}$ simply runs
$\qA^{H, \cQ}(z)$ (likewise $\qA^{G, \cQ}(z)$) by simulating its queries to $\cQ$
using the description $\langle \cQ\rangle$. Then, the O2H Lemma (Lemma
\ref{lem:O2Hprev}) implies the existence of an algorithm
$\widetilde{\qB}^{H}(z, \langle \cQ\rangle)$ that chooses $i \gets
[q]$, runs $\widetilde{\qA}^H(z, \langle \cQ\rangle)$, measures its
$i$-th query to $H$ and outputs the measurement outcome. Observe that
the algorithms $\widetilde{A}$ and $\widetilde{B}$ do not make use of
the description $\langle \cQ\rangle$ except for simulating the queries
made by $\qA$.  Consequently, there exists an algorithm
$\qB^{\cQ, H}(z)$ equivalent to $\widetilde{\qB}^{H}(z, \langle \cQ\rangle)$ that directly runs $\qA$ (instead of $\widetilde{\qA}$)
and simulates its oracle queries to $\cQ$ using its own access to
$\cQ$.
\end{proof}

\fi

\begin{remark}
We assume that $\qB$ also outputs the measured index $i$. However,
this output is not taken into account for notation such as $\qB^{\cQ,
H}(z) \in S$ for the sake of simplicity.
\end{remark}

\subsection{Standard Cryptographic Tools}\label{sec:standard_crypto}

\subsubsection*{Attribute-Based Encryption.}

\begin{definition}[Attribute-Based Encryption]\label{def:ABE}
An ABE scheme $\ABE$ is a tuple of four PPT algorithms $(\Setup, \KG,
\Enc,\allowbreak \Dec)$. 
Below, let $\cX=\{\cX_\secp\}_\secp$, $\cY=\{\cY_\secp\}_\secp$, and $R=\{R_\secp:\cX_\secp \times \cY_\secp \ra \bin \}_\secp$ be the ciphertext attribute space, key attribute space, and the relation associated with $\ABE$, respectively.
We note that we will abuse the notation and occasionally drop the subscript for these spaces for notational simplicity.
We also note that the message space is set to be $\bin^\ell$ below. 
\begin{description}
\item[$\Setup(1^\secp)\ra(\pk,\msk)$:] The setup algorithm takes a security parameter $1^\secp$ and outputs a public key $\pk$ and master secret key $\msk$.

\item[$\KG(\msk,y, r)\ra\sk_y$:] The key generation algorithm $\KG$
takes a master secret key $\msk$, a key attribute $y \in
\cY$, and explicit randomness $r$. It outputs a decryption key
$\sk_y$.
Note that $\KG$ is deterministic.\footnote{In the standard syntax, $\KG$ does not take explicit randomness, and is probabilistic. This change is just for notational convention in our schemes.}

\item[$\Enc(\pk,x,\msg)\ra\ct$:] The encryption algorithm takes a
public key $\pk$, a ciphertext attribute $x \in \cX$, and a
message $\msg$, and outputs a ciphertext $\ct$.

\item[$\Dec(\sk_y,\ct)\ra z$:] The decryption algorithm takes a
secret key $\sk_y$ and a ciphertext $\ct$ and outputs
$z \in \{ \bot \} \cup \bin^\ell$.

\item[Correctness:] We require that
\[
\Pr\left[
\Dec(\sk_y, \ct) = \msg
 \ :
\begin{array}{rl}
 &(\pk,\msk) \la \Setup(1^\secp),\\
 &r \leftarrow \bit^{\poly(\secp)},\\
 & \sk_y \gets \KG(\msk,y, r), \\
 &\ct \gets \Enc(\pk,x,\msg)
\end{array}
\right] \ge 1 -\negl(\secp).
\]
holds for all $x\in \cX$ and $y\in \cY$ such that $R(x,y)=0$ and $m\in \bin^\ell$.
\end{description}
\end{definition}

By setting $\cX$, $\cY$, and $R$ appropriately, we can recover
important classes of ABE. In particular, if we set
$\cX_\secp=\cY_\secp=\bin^*$ and define $R_\secp$ so that
$R_\secp(x,y)=0$ if $x=y$ and $R_\secp(x,y)=1$ otherwise, we recover
the definition of identity-based encryption (IBE). 
If we set $\cX_\secp=\bin^{n(\secp)}$ and $\cY_\secp$ to be the set of all circuits with input space $\bin^{n(\secp)}$ and size at most $s(\secp)$, where $n$ and $s$ are some polynomials, and define $R$ so that $R(x,y)=y(x)$, we recover the definition of (key policy) ABE for circuits. 

We introduce a new security notion for ABE that we call quantum selective-security for ABE where the adversary is allowed to get access to the key generation oracle in super-position.

\begin{definition}[Quantum Selective-Security for ABE]\label{def:qsel_ind_ABE}
We say that $\ABE$ is a \emph{selective-secure} ABE scheme for
relation $R:\cX\times \cY \to \bin$, if it satisfies the following
requirement, formalized by the experiment
$\expc{\ABE, \qA}{q}{sel}{ind}(1^\secp,\coin)$ between an adversary
$\qA$ and a challenger $\qCh$:
        \begin{enumerate}
            \item $\qA$ declares the challenge ciphertext attribute
                $x^*$. $\qCh$ runs $(\pk,\msk)\gets\Setup(1^\secp)$ and sends $\pk$ to $\qA$.
            \item $\qA$ can get access to the following quantum key generation oracle.
            \begin{description}

    \item[$\Oracle{qkg}(\qreg{Y},\qreg{Z})$:] Given two registers
        $\qreg{Y}$ and $\qreg{Z}$, it first applies the map
        $\ket{y}_{\qreg{Y}}\ket{b}_{\qreg{B}}\ra\ket{y}_{\qreg{Y}}\ket{b\oplus
        R(x^*,y)}_{\qreg{B}}$ and measures the register $\qreg{B}$, where $\qreg{B}$ is initialized to $\ket{0}_\qreg{B}$.
    If the result is $0$, it returns $\bot$. Otherwise, it
    chooses $r \leftarrow \bit^{\poly(\secp)}$, applies
    the map
    $\ket{y}_{\qreg{Y}}\ket{z}_{\qreg{Z}}\ra\ket{y}_{\qreg{Y}}\ket{z\oplus
    \KG(\msk,y,r)}_{\qreg{Z}}$ and returns the registers $\qreg{Y}$
    and $\qreg{Z}$.
    \end{description}

\item At some point, $\qA$ sends $(\msg_0,\msg_1)$ to $\qCh$. Then, $\qCh$
generates $\ct^*\gets\Enc(\pk,x^*,\allowbreak\msg_\coin)$ and sends $\ct^*$ to
$\qA$.

\item Again, $\qA$ can get access to the oracle $\Oracle{qkg}$.
\item $\qA$ outputs a guess $\coin^\prime$ for $\coin$ and the
experiment outputs $\coin'$.
\end{enumerate}
We say that $\ABE$ satisfies quantum selective security if, for all
QPT $\qA$, it holds that

\begin{align}
\advc{\ABE,\qA}{q}{sel}{ind}(1^\secp) \seteq \abs{\Pr[\expc{\ABE,\qA}{q}{sel}{ind} (1^\secp,0) \ra 1] - \Pr[\expc{\ABE,\qA}{q}{sel}{ind} (1^\secp,1) \ra 1] }\\\leq \negl(\secp).
\end{align}
\end{definition}

Boneh and Zhandry~\cite{C:BonZha13} introduced a similar quantum security notion for IBE and argued that it is straightforward to prove the quantum security of the IBE scheme by \cite{EC:AgrBonBoy10}, by leveraging the lattice trapdoor based proof technique.
It is easy to prove the quantum selective security of the ABE scheme for circuits by Boneh et al.~\cite{EC:BGGHNS14}, which relies on the lattice trapdoor based proof technique as well.
Formally, we have the following theorem.
\begin{theorem}
Assuming the hardness of the LWE problem, there exists a quantum selectively secure ABE scheme for all relations computable in polynomial time.
\end{theorem}
We elaborate on this in \cref{sec-quantum-secure-ABE}.

\ifnum\submission=1
\else
\paragraph{Compute-and-Compare Obfuscation.}
We define a class of circuits called compute-and-compare circuits as
follows:

\begin{definition}[Compute-and-Compare Circuits]\label{def:cc_circuits_searchability}
A compute-and-compare circuit $\cnc{P}{\lock,\msg}$ is of the form
\[
\cnc{P}{\lock,\msg}(x)\left\{
\begin{array}{ll}
    \msg&\textrm{if}\; P(x)=\lock\\
\bot&\text{otherwise}~
\end{array}
\right.
\]
where $P$ is a circuit, $\lock$ is a string called the lock value,
and $\msg$ is a message.
\end{definition}

We now introduce the definition of compute-and-compare obfuscation.
We assume that a program $P$ has an associated set of parameters $\pp_P$ (input size, output size, circuit size) which we do not need to hide.
\begin{definition}[Compute-and-Compare Obfuscation]\label{def:CCObf}
A PPT algorithm $\CCObf$ is an obfuscator for the family of distributions $D=\{D_\secp\}$ if the following holds:
\begin{description}
\item[Functionality Preserving:] There exists a negligible function
$\negl$ such that for all programs $P$, all lock values $\lock$, and
all messages $\msg$, it holds that

\begin{align}
\Pr[\forall x, \tlP(x)=\cnc{P}{\lock,\msg}(x) :
\tlP\la\CCObf(1^\secp,P,\lock,\msg)] \ge 1-\negl(\secp).
\end{align}
\item[Distributional Indistinguishability:] There exists an
efficient simulator $\Sim$ such that for all messages $\msg$, we have
\begin{align}
(\CCObf(1^\secp,P,\lock,\msg),\qaux)\approx(\CCSim(1^\secp,\pp_P,\abs{\msg}),\qaux),
\end{align}
where $(P,\lock,\qaux)\la D_\secp$.
\end{description}
\end{definition}

\begin{theorem}[\cite{FOCS:GoyKopWat17,FOCS:WicZir17}]
If the LWE assumption holds, there exists compute-and-compare obfuscation for all families of distributions $D=\{D_\secp\}$, where each $D_\secp$ outputs uniformly random lock value $\lock$ independent of $P$ and $\qaux$.
\end{theorem}
\fi
 
We present the definitions for SKE with certified deletion.
\begin{definition}[SKE-CD (Syntax)]\label{def:SKE-CD}
An SKE-CD scheme is a tuple of algorithms $(\KG,\qEnc,\qDec,\qDel,\Vrfy)$ with plaintext space $\Ms$ and key space $\Ks$.
\begin{description}
    \item[$\KG (1^\secp) \ra \sk$:] The key generation algorithm takes as input the security parameter $1^\secp$ and outputs a secret key $\sk \in \Ks$.
    \item[$\qEnc(\sk,\msg) \ra (\qct,\vk)$:] The encryption algorithm
    takes as input $\sk$ and a plaintext $\msg\in\Ms$ and outputs a
    ciphertext $\qct$ and a verification key $\vk$.

    \item[$\qDec(\sk,\qct) \ra \msg^\prime$:] The decryption algorithm
        takes as input $\sk$ and $\qct$ and outputs a plaintext $\msg^\prime \in \Ms$ or $\bot$.
    \item[$\qDel(\qct) \ra \cert$:] The deletion algorithm takes as
        input $\qct$ and outputs a certificate $\cert$.
    \item[$\Vrfy(\vk,\cert)\ra \top/\bot$:] The verification
        algorithm takes $\vk$ and $\cert$ as input and outputs
        $\top$ or $\bot$.

\item[Decryption correctness:] There exists a negligible function
    $\negl$ such that for any $\msg\in\Ms$, 
\begin{align}
\Pr\left[
\qDec(\sk,\qct)= \msg
\ :
\begin{array}{ll}
\sk\lrun \KG(1^\secp)\\
(\qct, \vk) \lrun \qEnc(\sk,\msg)
\end{array}
\right] 
\ge 1-\negl(\secp).
\end{align}

\item[Verification correctness:] There exists a negligible function
    $\negl$ such that for any $\msg\in\Ms$, 
\begin{align}
\Pr\left[
\Vrfy(\vk,\cert)=\top
\ :
\begin{array}{ll}
\sk\lrun \KG(1^\secp)\\
(\qct, \vk) \lrun \qEnc(\sk,\msg)\\
\cert \lrun \qDel(\qct)
\end{array}
\right] 
\ge 1-\negl(\secp).
\end{align}
\end{description}
\end{definition}

We introduce indistinguishability against Chosen Verification Attacks (CVA).
\begin{definition}[IND-CVA-CD Security]\label{def:reusable_sk-vo_certified_del}
We consider the following security experiment
$\expc{\CDSKE,\qA}{ind}{cva}{cd}(1^\secp,\coin)$.

\begin{enumerate}
    \item The challenger computes $\sk \la \KG(1^\secp)$.
    \item Thoughout the experiment, $\qA$ can get access to the following oracle.
    \begin{description}
    \item[$\Oracle{\qEnc}(\msg)$:] On input $\msg$, it generates
    $(\qct, \vk)\gets\qEnc(\sk,\msg)$ and returns $(\vk,\qct)$.  
    \end{description}
    \item $\qA$ sends $(\msg_0,\msg_1)\in\cM^2$ to the challenger. 
    \item The challenger computes $(\qct^*,\vk^*) \la
        \qEnc(\sk,\msg_\coin)$ and sends $\qct^*$ to $\qA$.
    \item Hereafter, $\qA$ can get access to the following oracle, where $V$ is initialized to $\bot$.
    \begin{description}
        \item[$\Oracle{\Vrfy}(\cert)$:] On input $\cert$, it returns $\sk$ and updates $V$ to $\top$ if $\Vrfy(\vk^*,\cert)=\top$. Otherwise, it returns $\bot$.
    \end{description}
    \item When $\qA$ outputs $\coin'\in \bit$, the experiment outputs $\coin^\prime$ if $V=\top$ and otherwise outputs $0$.
\end{enumerate}
We say that $\CDSKE$ is IND-CVA-CD secure if for any QPT $\qA$, it holds that
\begin{align}
\advc{\CDSKE,\qA}{ind}{cva}{cd}(1^\secp)\seteq \abs{\Pr[
\expc{\CDSKE,\qA}{ind}{cva}{cd}(1^\secp, 0)=1] - \Pr[
\expc{\CDSKE,\qA}{ind}{cva}{cd}(1^\secp, 1)=1] }\\\le \negl(\secp).
\end{align}
\end{definition}

\begin{definition}[BB84-Based SKE-CD]\label{def:bb84}
We say that an SKE-CD scheme $\CDSKE=(\KG,\qEnc,\qDec,\qDel,\Vrfy)$
is a BB84-based SKE-CD scheme if it satisfies the following
conditions.

\begin{itemize}
   \item Let $(\qct,\vk)\gets\qEnc(\sk,\msg)$. $\vk$ is of the form
    $(x,\theta)\in\bit^{\ctlen}\times\bit^{\ctlen}$,
    and $\qct$ is of the form $\ket{\psi_1}\tensor
    \cdots \tensor\ket{\psi_{\ctlen}}$, where
    \begin{align}
    \ket{\psi_i}=
    \begin{cases}
        \ket{x[i]} & if~~ \theta[i]=0\\
        \ket{0}+(-1)^{x[i]}\ket{1} & if~~ \theta[i]=1.
    \end{cases}
    \end{align}
    Moreover, there exists $n<\ctlen$ such that $\theta[i]=0$ for every $i\in[n+1,\ctlen]$. We call $x[n+1]\|\cdots\|x[\ctlen]$ a classical part of $\qct$. The parameter $n$ is specified by a construction. The classical part has information of $\theta$, and we can compute $\theta$ from it and $\sk$.

        \item $\qDel(\qct)$ measures each qubit of $\qct$ in the
    Hadamard basis and outputs the measurement result
    $\cert\in\bit^{\ctlen}$.

    \item $\Vrfy(\vk,\cert)$ outputs $\top$ if $\cert[i]=x[i]$ holds
    for every $i\in[n]$ such that $\theta[i]=1$, and $0$
    otherwise.

\item \textbf{Classical Decryption Property:} There exists an
additional deterministic polynomial time algorithm $\CDec$ with the
following property. Let $(\qct,\vk)\gets\qEnc(\sk,\msg)$, where $\vk=(x,\theta)$. Let $u\in\bit^{\ctlen}$ be any string such that $u[i] =x[i]$ for all $i : \theta[i] = 0$. Then, the following holds:

$$\Pr\Big[\CDec\big(\sk, u\big)= \msg\Big] \ge 1 - \negl(\secp)$$ 

\end{itemize}
\end{definition}

\begin{theorem}\label{thm:SKECD-BB84}
There exists a BB84-based SKE-CD scheme satisfying IND-CVA-CD
security, assuming the existence of a CPA-secure Secret-Key
Encryption scheme.
\end{theorem}

Kitagawa and Nishimaki~\cite{AC:KitNis22} claimed the same statement as~\cref{thm:SKECD-BB84}.
However, their proof has a gap because known BB84-based SKE-CD schemes do not satisfy the unique certificate property, which they introduced. Hence, we prove \cref{thm:SKECD-BB84} in \cref{sec:SKECD-BB84}.

%


\section{Encryption with Collusion-Resistant SKL} 
\label{sec:cr-skl-defs}
In this section, we define the notions of public-key and secret-key
encryption with collusion-resistant secure key-leasing.
\subsection{Definitions of PKE-CR-SKL}
The syntax of PKE-CR-SKL is defined as follows. 
\begin{definition}[PKE-CR-SKL]
    A PKE-CR-SKL scheme $\PKECRSKL$ is a tuple of five algorithms
    $(\Setup,\qKG, \Enc, \allowbreak \qDec,\qVrfy)$. 
Below, let $\cM$  be the message space of $\PKECRSKL$. 
\begin{description}
\item[$\Setup(1^\secp)\ra(\ek,\msk)$:] The setup algorithm takes a
    security parameter $1^\lambda$, and outputs an encryption key
    $\ek$ and a master secret-key $\msk$.

\item[$\qKG(\msk)\ra(\qdk,\vk)$:] The key generation algorithm takes
the master secret-key $\msk$ as input, and outputs a decryption key
$\qdk$ and a verification key $\vk$.

\item[$\Enc(\ek,\msg)\ra\ct$:] The encryption algorithm takes an
    encryption key $\ek$ and a message $\msg \in \cM$, and outputs a ciphertext $\ct$.

\item[$\qDec(\qdk,\ct)\ra\widetilde{\msg}/\bot$:] The decryption
    algorithm takes a decryption key $\qdk$ and a ciphertext $\ct$,
    and outputs a value $\widetilde{\msg}$ or $\bot$.


\item[$\qVrfy(\vk,\widetilde{\qdk})\ra\top/\bot$:] The verification algorithm takes a verification key $\vk$ and a (possibly malformed) decryption key $\widetilde{\qdk}$, and outputs $\top$ or $\bot$.

\item[Decryption correctness:]For every $\msg \in \cM$, we have
\begin{align}
\Pr\left[
\qDec(\qdk, \ct) \allowbreak = \msg
\ :
\begin{array}{ll}
(\ek,\msk)\gets\Setup(1^\secp)\\
(\qdk,\vk)\gets\qKG(\msk)\\
\ct\gets\Enc(\ek,\msg)
\end{array}
\right] 
\ge 1-\negl(\secp).
\end{align}

\item[Verification correctness:] We have 
\begin{align}
\Pr\left[
\qVrfy(\vk,\qdk)=\top
\ :
\begin{array}{ll}
(\ek,\msk)\gets\Setup(1^\secp)\\
(\qdk,\vk)\gets\qKG(\msk)
\end{array}
\right] 
\ge 1-\negl(\secp).
\end{align}
\end{description}
\end{definition}
\begin{remark}\label{rem:reusability}
We can assume without loss of generality that a decryption key of a PKE-CR-SKL scheme is reusable, i.e., it can be reused to decrypt (polynomially) many ciphertexts. In particular, we can assume that 
for honestly generated $\ct$ and $\qdk$, if we decrypt $\ct$ by using $\qdk$, the state of the decryption key after the decryption is negligibly close to that before the decryption in terms of trace distance. 
This is because the output of the decryption is almost deterministic by decryption correctness, and thus such an operation can be done without almost disturbing the input state by the gentle measurement lemma~\cite{Winter99}.    
\end{remark}

\begin{definition}[IND-KLA Security]\label{def:IND-CPA_PKECRSKL}
We say that a PKE-CR-SKL scheme $\PKECRSKL$  with the message space
$\cM$ is IND-KLA secure, if it satisfies the following requirement,
formalized by the experiment
$\expb{\PKECRSKL,\qA}{ind}{kla}(1^\secp,\coin)$ between an adversary
$\qA$ and a challenger $\qCh$:
\begin{enumerate}
\item $\qCh$ runs $(\ek,\msk)\gets\Setup(1^\secp)$ and sends $\ek$
to $\qA$. 

\item $\qA$ requests $q$ decryption keys for some polynomial $q$.
$\qCh$ generates $(\qdk_i,\vk_i)\gets\qKG(\msk)$ for
every $i\in[q]$ and sends $\qdk_1,\ldots,\qdk_q$ to $\qA$.

\item $\qA$ can get access to the following (stateful) verification
oracle $\Oracle{\qVrfy}$ where $V_i$ is initialized to $\bot$ for all $i\in [q]$:

\begin{description}
\item[ $\Oracle{\qVrfy}(i,\widetilde{\qdk})$:] It runs $d \gets \qVrfy(\vk_i,\widetilde{\qdk})$ and returns $d$.  

If $V_i=\bot$ and $d=\top$, it updates $V_i\seteq \top$. 

\end{description}

\item $\qA$ sends $(\msg_0^*,\msg_1^*)\in \cM^2$ to the challenger.
If $V_i=\bot$ for some $i\in[q]$, the challenger outputs $0$ as the
final output of this experiment. Otherwise, the challenger generates
$\ct^*\la\Enc(\ek,\msg_\coin^*)$ and sends $\ct^*$ to $\qA$.

\item $\qA$ outputs a guess $\coin^\prime$ for $\coin$. $\qCh$
outputs $\coin'$ as the final output of the experiment.

\end{enumerate}

For any QPT $\qA$, it holds that
\begin{align}
&\advb{\PKECRSKL,\qA}{ind}{kla}(1^\secp) \seteq \\
&\abs{\Pr[\expb{\PKECRSKL,\qA}{ind}{kla} (1^\secp,0) \ra 1] - \Pr[\expb{\PKECRSKL,\qA}{ind}{kla} (1^\secp,1) \ra 1] }\leq \negl(\secp).
\end{align}
\end{definition}

\subsection{Definitions of SKE-CR-SKL}\label{def:ske_cr_skl}
The syntax of SKE-CR-SKL is defined as follows. 
\begin{definition}[SKE-CR-SKL]
    An SKE-CR-SKL scheme $\SKECRSKL$ is a tuple of five algorithms
$(\Setup,\qKG, \Enc, \allowbreak \qDec, \qVrfy)$. 
Below, let $\cM$  be the message space of $\SKECRSKL$. 
\begin{description}
\item[$\Setup(1^\secp)\ra\msk$:] The setup algorithm takes a
security parameter $1^\lambda$ and outputs a master secret-key
$\msk$.

\item[$\qKG(\msk)\ra(\qdk,\vk,\tk)$:] The key generation algorithm
takes the master secret-key $\msk$ as input. It outputs a decryption
key $\qdk$, a verification key $\vk$, and a testing key
$\tk$.

\item[$\Enc(\msk,\msg)\ra\ct$:] The encryption algorithm takes the
master secret-key $\msk$ and a message $\msg \in \cM$, and outputs a
ciphertext $\ct$.

\item[$\qDec(\qdk,\ct)\ra\widetilde{\msg}$:] The decryption
algorithm takes a decryption key $\qdk$ and a ciphertext $\ct$,
and outputs a value $\widetilde{\msg}$.


\item[$\qVrfy(\vk, \widetilde{\qdk})\ra\top/\bot$:] The verification
    algorithm takes a verification key $\vk$ and a (possibly
    malformed) decryption key $\widetilde{\qdk}$,
and outputs $\top$ or $\bot$.


\item[Decryption correctness:] For all $\msg \in \cM$, we have
\begin{align}
\Pr\left[
\qDec(\qdk, \ct) \allowbreak = \msg
\ :
\begin{array}{ll}
\msk\gets\Setup(1^\secp)\\
(\qdk,\vk,\tk)\gets\qKG(\msk)\\
\ct\gets\Enc(\msk,\msg)
\end{array}
\right] 
\ge 1-\negl(\secp).
\end{align}

\item[Verification correctness:] We have 
\begin{align}
\Pr\left[
\qVrfy(\vk,\qdk)=\top
\ :
\begin{array}{ll}
\msk\gets\Setup(1^\secp)\\
(\qdk,\vk,\tk)\gets\qKG(\msk)\\
\end{array}
\right] 
\ge 1-\negl(\secp).
\end{align}


\end{description}
\end{definition}

\begin{definition}[Classical Decryption
    Property]\label{def:CDEC_SKE-CR-SKL}
We say that $\SKECRSKL\allowbreak=(\Setup,\qKG,\Enc,\qDec,\qVrfy)$ has
the classical decryption property if there exists a deterministic
polynomial time algorithm $\CDec$ such that given $\qdk$ in the
register $\qreg{DK}$ and ciphertext $\ct$, $\qDec$ applies the map
$\ket{u}_{\qreg{DK}}\ket{v}_{\qreg{MSG}}\ra\ket{u}_{\qreg{DK}}\allowbreak\ket{v\oplus\CDec(u,\ct)}_{\qreg{MSG}}$
and outputs the measurement result of the register $\qreg{MSG}$ in
the computational basis, where $\qreg{MSG}$ is initialized to
$\ket{0\cdots0}_{\qreg{MSG}}$.  \end{definition}

\begin{definition}[Key Testability]\label{def:key-testability-SKE}
We say that an SKE-CR-SKL scheme $\SKECRSKL$ with the classical
decryption property satisfies key testability, if there exists a
classical deterministic algorithm $\KeyTest$ that satisfies the
following conditions:

\begin{itemize}
\item \textbf{Syntax:} $\KeyTest$ takes as input a testing key $\tk$
and a classical string $\dk$ as input. It outputs $0$ or $1$.

\item \textbf{Correctness:} Let $\msk\gets\Setup(1^\secp)$ and
$(\qdk,\vk,\tk)\gets\qKG(\msk)$. We denote the register holding
$\qdk$ as $\qreg{DK}$. Let $\qreg{KT}$ be a register that is
initialized to $\ket{0}_{\qreg{KT}}$. If we apply the map
$\ket{u}_{\qreg{DK}}\ket{\beta}_{\qreg{KT}}\ra\ket{u}_{\qreg{DK}}\ket{\beta\oplus\KeyTest(\tk,u)}_{\qreg{KT}}$
to the registers $\qreg{DK}$ and $\qreg{KT}$ and then measure
$\qreg{KT}$ in the computational basis, we obtain $1$ with
overwhelming probability.

\item \textbf{Security:} Consider the following experiment
$\expb{\SKECRSKL,\qA}{key}{test}(1^\secp)$.

\begin{enumerate}
\item The challenger $\qCh$ runs $\msk\gets\Setup(1^\secp)$ and
initializes $\qA$ with input $\msk$. 

\item $\qA$ requests $q$ decryption keys for some polynomial $q$.
$\qCh$ generates $(\qdk_i,\vk_i,\allowbreak\tk_i)\gets\qKG(\msk)$
for every $i\in[q]$ and sends $(\qdk_i,\vk_i,\tk_i)_{i\in[q]}$ to
$\qA$.

\item $\qA$ sends $(k, \dk, \msg)$ to $\qCh$, where $k$ is
an index, $\dk$ is a classical string and $\msg$ is a message.
$\qCh$ generates $\ct\gets\Enc(\msk,\msg)$. $\qCh$ outputs $\top$ if
$\KeyTest(\tk_k,\dk)=1$ and $\CDec(\dk,\ct)\ne\msg$.
Otherwise, $\qCh$ outputs $\bot$.
\end{enumerate}

For all QPT $\qA$, the following must hold:

\begin{align}
\advb{\SKECRSKL,\qA}{key}{test}(1^\secp) \seteq
\Pr[\expb{\SKECRSKL,\qA}{key}{test}(1^\secp) \ra \top] \le
\negl(\secp).
\end{align} 
\end{itemize}
\end{definition}

\begin{definition}[OT-IND-KLA Security]\label{def:OT-IND-CPA_SKECRSKL}
We say that an SKE-CR-SKL scheme with key testability $\SKECRSKL$
is OT-IND-KLA secure, if it satisfies
the following requirement, formalized by the experiment
$\expc{\SKECRSKL,\qA}{ot}{ind}{kla}(1^\secp,\coin)$ between an
adversary $\qA$ and a challenger $\qCh$:

\begin{enumerate}
\item $\qCh$ runs $\msk\gets\Setup(1^\secp)$ and initializes
$\qA$ with the security parameter $1^\secp$.

\item $\qA$ requests $q$ decryption keys for some polynomial $q$.
The challenger generates $(\qdk_i,\vk_i,\tk_i)\gets\qKG(\msk)$
for every $i\in[q]$ and sends $(\qdk_i,\tk_i)_{i\in[q]}$ to $\qA$.

\item $\qA$ can get access to the following (stateful) verification
oracle $\Oracle{\qVrfy}$ where $V_i$ is initialized to be $\bot$:

\begin{description}
\item[ $\Oracle{\qVrfy}(i, \widetilde{\qdk})$:] It runs $d \gets
\qVrfy(\vk_i, \widetilde{\qdk})$ and returns $d$.  

If $V_i=\bot$ and $d=\top$, it updates $V_i\seteq \top$. 
\end{description}
\item $\qA$ sends $(\msg_0^*,\msg_1^*)\in \cM^2$ to the challenger.
If $V_i=\bot$ for some $i\in[q]$, $\qCh$ outputs
$0$ as the final output of this experiment. Otherwise, $\qCh$
generates $\ct^*\la\Enc(\msk,\msg_\coin^*)$ and sends
$\ct^*$ to $\qA$.

\item $\qA$ outputs a guess $\coin^\prime$ for $\coin$.
$\qCh$ outputs $\coin'$ as the final output of the
experiment.
\end{enumerate}

For all QPT $\qA$, it holds that:
\begin{align}
&\advc{\SKECRSKL,\qA}{ot}{ind}{kla}(1^\secp) \seteq\\
&\abs{\Pr[\expc{\SKECRSKL,\qA}{ot}{ind}{kla} (1^\secp,0) \ra 1] - \Pr[\expc{\SKECRSKL,\qA}{ot}{ind}{kla} (1^\secp,1) \ra 1] }\leq \negl(\secp).
\end{align} 
\end{definition}

\section{SKE-CR-SKL with Key Testability}\label{sec:SKECRSKL-KT}
In this section, we show how to achieve SKE-CR-SKL introduced in~\cref{def:ske_cr_skl}.
\subsection{Construction}\label{sec:SKECRSKL-KT-construction}
We construct an SKE-CR-SKL scheme with key testability $\SKECRSKL=
\SKECRSKL.(\Setup,\qKG,\Enc,\qDec,\allowbreak\qVrfy)$ having the
additional algorithms $\CDec$ and $\KeyTest$, using the
following building blocks.

\begin{itemize}
\item BB84-based SKE-CD scheme (Definition \ref{def:bb84}) $\SKECD =
\SKECD.(\KG,\qEnc,\qDec,\allowbreak\qDel,\Vrfy)$ having the classical
decryption algorithm $\SKECD.\CDec$.

\item OWF $f:\bit^\secp\ra\bit^{p(\secp)}$ for some polynomial $p$.
\end{itemize}

Let $\cM \seteq \bit^{\msglen}$ be the plaintext space. The construction is as follows:

\begin{description}

\item[$\SKECRSKL.\Setup(1^\secp)$:] $ $
\begin{enumerate}
    \item Generate $r\gets\bit^{\msglen}$.
    \item Generate $\skecd.\sk\gets\SKECD.\KG(1^\secp)$.
    \item Output $\msk\seteq(\skecd.\sk, r)$.
\end{enumerate}

\item[$\SKECRSKL.\qKG(\msk)$:] $ $
\begin{enumerate}
    \item Parse $\msk=(\skecd.\sk, r)$.
    \item Generate
        $(\skecd.\qct,\skecd.\vk)\gets\SKECD.\qEnc(\skecd.\sk,r)$.
        $\skecd.\vk$ is of the form
        $(x,\theta)\in\bit^{\ctlen}\times\bit^{\ctlen}$, and
        $\skecd.\qct$ is of the form
        $\ket{\psi_1}_{\qreg{SKECD.CT_1}}\tensor\cdots\tensor\ket{\psi_{\ctlen}}_{\qreg{SKECD.CT_{\ctlen}}}$.

    \item Generate $s_{i,b}\la\bit^\secp$ and compute $t_{i,b}\la f(s_{i,b})$ for every $i\in[\ctlen]$ and $b\in\bit$. 
    Set $T\seteq
    t_{1,0}\|t_{1,1}\|\cdots\|t_{\ctlen,0}\|t_{\ctlen,1}$ and $S =
    \{s_{i,0} \xor s_{i, 1}\}_{i \in [\ctlen] \; : \;\theta[i] =
    1}$.
    \item Prepare a register $\qreg{S_i}$ that is initialized to
    $\ket{0^\secp}_{\qreg{S_i}}$ for every $i\in[\ctlen]$. 
    \item For every $i\in[\ctlen]$, apply the map
    \begin{align}
    \ket{u_i}_{\qreg{SKECD.CT_i}}\tensor\ket{v_i}_{\qreg{S_i}}
    \ra
    \ket{u_i}_{\qreg{SKECD.CT_i}}\tensor\ket{v_i\oplus s_{i,u_i}}_{\qreg{S_i}}
    \end{align}
    to the registers $\qreg{SKECD.CT_i}$ and $\qreg{S_i}$ and obtain the resulting state $\rho_i$.
    \item Output $\qdk = (\rho_i)_{i\in{[\ctlen]}}$,
    $\vk=(x,\theta,S)$, and $\tk=T$.
\end{enumerate}

\item[$\SKECRSKL.\Enc(\msk, \msg)$:] $ $
\begin{enumerate}
    \item Parse $\msk = (\skecd.\sk, r)$.
    \item Output $\ct\seteq(\skecd.\sk, r\oplus\msg)$.
    
\end{enumerate}

\item[$\SKECRSKL.\CDec(\dk, \ct)$:] $ $
\begin{enumerate}
\item Parse $\ct = (\skecd.\sk, z)$. Let
$\widetilde{\dk}$ be the sub-string of $\dk$ on register
$\qreg{SKECD.CT} = \qreg{SKECD.CT_1} \otimes \cdots \otimes \qreg{SKECD.CT_{\ctlen}}$.
\item Output $z \xor \SKECD.\CDec(\skecd.\sk, \widetilde{\dk})$.
\end{enumerate}

\item[$\SKECRSKL.\qDec(\qdk, \ct)$:] $ $
\begin{enumerate}
\item Parse $(\rho_i)_{i \in [\ctlen]}$. We denote the register
holding $\rho_i$ as $\qreg{SKECD.CT_i}\tensor\qreg{S_i}$ for
every $i\in[\ctlen]$.

\item Prepare a register $\qreg{MSG}$ of $\msglen$ qubits that is
initialized to $\ket{0\cdots0}_{\qreg{MSG}}$.

\item Apply the map
\begin{align}
&\ket{u}_{\bigotimes_{i\in[\ctlen]}\qreg{SKECD.
CT_i}} \tensor\ket{w}_{\qreg{MSG}} \ra\\
&\ket{u}_{\bigotimes_{i\in[\ctlen]}\qreg{SKECD.CT_i}}\tensor\ket{w\oplus
\SKECRSKL.\CDec(u, \ct)}_{\qreg{MSG}}
\end{align}

to the registers
$\bigotimes_{i\in[\ctlen]}\qreg{SKECD.CT_i}$ and
$\qreg{MSG}$.
\item Measure $\qreg{MSG}$ in the computational basis and output the result $\msg^\prime$.
\end{enumerate}

\item[$\SKECRSKL.\qVrfy(\vk,\widetilde{\qdk})$:] $ $
\begin{enumerate}
\item Parse $\vk = (x,\theta,S=\{s_{i,0} \xor
    s_{i,1}\}_{i\in[\ctlen]\; : \; \theta[i]=1})$ and
    $\qdk = (\rho_i)_{i\in[\ctlen]}$ where $\rho_i$ is a state on
    the registers $\qreg{SKECD.CT_i}$ and $\qreg{S_i}$.
\item For every $i \in [\ctlen]$, measure $\rho_i$ in the Hadamard
basis to get outcomes $c_i, d_i$ corresponding to 
the registers $\qreg{SKECD.CT_i}$ and $\qreg{S_i}$ respectively.

\item Output $\top$ if $x[i]=c_i \oplus d_i\cdot(s_{i,0}\oplus
        s_{i,1})$ holds for every $i\in[\ctlen]$ such that $\theta[i]=1$.
    Otherwise, output $\bot$.
\end{enumerate}

\item[$\SKECRSKL.\KeyTest(\tk,\dk)$:] $ $
\begin{enumerate}
\item Parse $\dk$ as a string over the registers
$\qreg{SKECD.CT} = \qreg{SKECD.CT_1} \otimes \cdots \otimes
\qreg{SKECD.CT_{\ctlen}}$ and $\qreg{S} = \qreg{S_1} \otimes
\cdots \otimes \qreg{S_{\ctlen}}$.
Let $u_i$ denote the value on
$\qreg{SKECD.CT_i}$ and $v_i$ the value on $\qreg{S_i}$. Parse $\tk$
as $T=t_{1,0}\|t_{1,1}\|\cdots \|t_{\ctlen,0}\|t_{\ctlen,1}$.

\item Let $\Check[t_{i,0},t_{i_1}](u_i,v_i)$ be the deterministic
algorithm that outputs $1$ if $f(v_i)=t_{i,u_i}$ holds and $0$
otherwise.

\item Output $\Check[t_{1,0},t_{1,1}](u_1,v_1) \land
\cdots \land
\Check[t_{\ctlen,0},t_{\ctlen,1}](u_{\ctlen},v_{\ctlen})$.
\end{enumerate}

\end{description}

\paragraph{Decryption correctness:} For a ciphertext $\ct =
(\skecd.\sk, r \xor \msg)$, the decryption algorithm $\qDec$ performs
the following computation:

\begin{align}
&\ket{u}_{\bigotimes_{i\in[\ctlen]}\qreg{SKECD.
CT_i}} \tensor\ket{w}_{\qreg{MSG}} \ra\\
&\ket{u}_{\bigotimes_{i\in[\ctlen]}\qreg{SKECD.CT_i}}\tensor\ket{w
\xor r \xor \msg \xor \SKECD.\CDec(\skecd.\sk, u)}_{\qreg{MSG}}
\end{align}

Recall that $\qreg{SKECD.CT} = \qreg{SKECD.CT_1} \otimes \cdots
\otimes \qreg{SKECD.CT_{\ctlen}}$ is a register corresponding to the
ciphertext of a BB84-based SKE-CD scheme $\SKECD$. Hence, from the Classical
Decryption property of $\SKECD$ (Definition \ref{def:bb84}), it must be that
$\SKECD.\CDec(\skecd.\sk, u) = r$ for every $u$ in the superposition
of $\skecd.\qct$. Consequently, $\msg$ is written onto the
$\qreg{MSG}$ register in each term of the superposition and decryption
correctness follows.

\paragraph{Verification correctness:}

Observe that for the Hadamard basis positions ($i \in [\ctlen]$ such that
$\theta[i] = 1$), $\rho_i$ is of the form:

$$\rho[i] = \ket{0}_{\qreg{SKECD.CT_i}}\ket{s_{i,0}}_{\qreg{S_i}}+
(-1)^{x[i]}\ket{1}_{\qreg{SKECD.CT_i}}\ket{s_{i,1}}_{\qreg{S_i}}$$

It is easy to see that measuring the state in the Hadamard basis gives
outcomes $c_i, d_i$ (on registers $\qreg{SKECD.CT_i}$ and
$\qreg{S_i}$ respectively) satisfying $x[i] = c_i \xor d_i \cdot
(s_{i,0} \xor s_{i,1})$. Hence, the verification correctness follows.

\begin{theorem} There exists an SKE-CR-SKL scheme satisfying
OT-IND-KLA security and Key-Testability, assuming the existence of a
BB84-based SKE-CD scheme and the existence of an OWF.
\end{theorem}

\ifnum\submission=1
We prove this theorem in~\cref{appsec:security_proofs_for_ske_cr_skl}.
\else
We prove this theorem in the subsequent sections.
\ifnum\submission=1
\section{Security Proofs for SKE-CR-SKL}\label{appsec:security_proofs_for_ske_cr_skl}
\fi
\subsection{Proof of OT-IND-KLA Security}\label{proof:ot-ind}\nikhil{Move to Appendix? Commented out to check Pg Limit.}
\ryo{I set the submission flag. If we switch the flag to 1, these proofs are moved to Appendix. See the main.tex.}

Let $\qA$ be an adversary for the OT-IND-KLA security of the
construction $\SKECRSKL$ that makes use of a BB84-based
SKE-CD scheme $\SKECD$. Consider the hybrid $\Hyb_j^\coin$ defined
as  follows:

\begin{description}
\item[$\hybi{j}^\coin$:] $ $
\begin{enumerate}
\item The challenger $\qCh$ runs $\msk \gets
\SKECRSKL.\Setup(1^\secp)$ 
and initializes $\qA$ with input $1^\secp$.

\item $\qA$ requests $q$ decryption keys for some polynomial 
$q$. For each $k \in [j]$, $\qCh$ generates
$(\qdk_i,\vk_i,\tk_i) \gets \qKGt(\msk)$ where $\qKGt$ is defined 
as follows (the difference from $\SKECRSKL.\qKG$ is colored 
in red):

\item[$\qKGt(\msk)$:] $ $
\begin{enumerate}
    \item Parse $\msk=(\skecd.\sk, r)$.
    \item \textcolor{red}{Sample $\widetilde{r} \gets
        \bit^{\msglen}$.}
    \item Generate
$(\skecd.\qct,\skecd.\vk)\gets\SKECD.\qEnc(\skecd.\sk,
\textcolor{red}{\widetilde{r}})$.
$\skecd.\vk$ is of the form
$(x,\theta)\in\bit^{\ctlen}\times\bit^{\ctlen}$, and $\skecd.\qct$
can be described as
$\ket{\psi_1}_{\qreg{SKECD.CT_1}}\tensor\cdots\tensor\ket{\psi_{\ctlen}}_{\qreg{SKECD.CT_{\ctlen}}}$.

    \item Generate $s_{i,b}\la\bit^\secp$ and compute $t_{i,b}\la
        f(s_{i,b})$ for every $i\in[\ctlen]$ and $b\in\bit$. 
    Set $T\seteq t_{1,0}\|t_{1,1}\|\cdots\|t_{\ctlen,0}\|t_{\ctlen,1}$ and $S =
    \{s_{i,0} \xor s_{i, 1}\}_{i \in [\ctlen] \; : \;\theta[i] = 1}$.
    \item Prepare a register $\qreg{S_i}$ that is initialized to $\ket{0\cdots0}_{\qreg{S_i}}$ for every $i\in[\ctlen]$. 
    \item For every $i\in[\ctlen]$, apply the map
    \begin{align}
    \ket{u_i}_{\qreg{SKECD.CT_i}}\tensor\ket{v_i}_{\qreg{S_i}}
    \ra
    \ket{u_i}_{\qreg{SKECD.CT_i}}\tensor\ket{v_i\oplus s_{i,u_i}}_{\qreg{S_1}}
    \end{align}
    to the registers $\qreg{SKECD.CT_i}$ and $\qreg{S_i}$ and obtain the resulting state $\rho_i$.
\item Output $\qdk=(\rho_i)_{i\in[\ctlen]}$, $\vk=(x,\theta,S)$, and $\tk=T$.
\end{enumerate}

\item On the other hand, for $k = j+1, \ldots, q$, $\qCh$ generates
$(\qdk_k, \vk_k,\tk_k) \gets \SKECRSKL.\qKG(\msk)$. Then,
$\qCh$ sends $(\qdk_k,\tk_k)_{k\in[q]}$ to $\qA$.

\item $\qA$ can
get access to the following (stateful) verification oracle
$\Oracle{\qVrfy}$ where $V_i$ is initialized to $\bot$:

\begin{description} \item[
        $\Oracle{\qVrfy}(i, \widetilde{\qdk})$:] It runs $d
        \gets \SKECRSKL.\qVrfy(\vk_i, \widetilde{\qdk})$ and returns $d$.
        If $V_i=\bot$ and $d=\top$, it updates $V_i\seteq
        \top$. 
\end{description}

\item $\qA$ sends $(\msg_0^*,\msg_1^*)\in 
\bit^{\msglen} \times \bit^{\msglen}$ to $\qCh$. If $V_i=\bot$ for
some $i\in[q]$, $\qCh$ outputs $0$ as the final output of this
experiment. Otherwise, $\qCh$ generates
$\ct^*\la\SKECRSKL.\Enc(\msk,\msg_\coin^*)$ and sends $\ct^*$ to
$\qA$.

\item $\qA$ outputs a guess $\coin^\prime$ for $\coin$. The 
challenger outputs $\coin'$ as the final output of the 
experiment.
\end{enumerate}
\end{description}

We will now prove the following lemma:

\begin{lemma}
$\forall j \in \{0, \ldots, q-1\}$ and $\coin \in \bit: \Hyb_j^\coin \approx \Hyb_{j+1}^\coin$.
\end{lemma}
\begin{proof}
Suppose $\Hyb_j^\coin \not \approx \Hyb_{j+1}^\coin$. Let $\qD$ be
a corresponding distinguisher. We will
construct a reduction $\qR$ that breaks the IND-CVA-CD security of
the BB84-based SKE-CD scheme $\SKECD$. The execution of $\qR^\qD$ in
the experiment $\expc{\SKECD,\qR}{ind}{cva}{cd}(1^\secp,b)$ proceeds
as follows:

\begin{description}
\item Execution of $\qR^\qD$ in
$\expc{\SKECD,\qR}{ind}{cva}{cd}(1^\secp,b)$:

\begin{enumerate}
\item The challenger $\qCh$ computes $\skecd.\sk \leftarrow
\SKECD.\KG(1^\lambda)$.
\item $\qR$ samples $(r_0, r_1) \leftarrow \bit^{\msglen} \times
    \bit^{\msglen}$ and
sends it to $\qCh$.
\item $\qCh$ computes $(\skecd.\qct^\star, \skecd.\vk^\star)
\leftarrow
\SKECD.\qEnc(\skecd.\sk, r_b)$ and sends $\skecd.\qct^\star$ to $\qR$.
\item $\qR$ initializes $\qD$ with $1^\secp$. $\qD$ requests $q$ keys for some polynomial $q$.

\item For each $k \in [j]$, $\qR$ computes $\qdk_k$ as follows:
\begin{itemize}
\item Sample a random value $\widetilde{r} \gets
    \bit^{\msglen}$.
\item Compute $(\skecd.\widetilde{\qct}, \skecd.\widetilde{\vk}) \leftarrow
    \Oracle{\qEnc}(\widetilde{r})$.
\item Compute $\qdk_k$ by executing Steps 2.(c)-2.(g) as in
$\Hyb_j^0$, but using $\skecd.\widetilde{\qct}$ in place of
$\skecd.\qct$.
\end{itemize}

\item $\qR$ computes $\qdk_{j+1}$ by executing Steps 2.(c)-2.(g) as
in $\Hyb_j^0$, but using $\skecd.\qct^\star$ in place of
$\skecd.\qct$.

\item For each $k \in [j+2,q]$, $\qR$ computes $\qdk_k$ as follows:
\begin{itemize}
\item Compute $(\skecd.\widetilde{\qct}, \skecd.\widetilde{\vk}) \leftarrow
\Oracle{\qEnc}(r_1)$.
\item Compute $\qdk_k$ by executing Steps 2.(c)-2.(g) as in
$\Hyb_j^0$, but using $\skecd.\widetilde{\qct}$ in place of
$\skecd.\qct$.
\end{itemize}
\item $\qR$ sends $\qdk_1, \ldots, \qdk_q$ to $\qD$ and initializes
$V_k = \bot$ for every $k \in [q]$.
\item If $k \neq j+1$, $\qR$ simulates
the response  of oracle $\Oracle{\qVrfy}(k, \widetilde{\qdk})$ as follows:

\begin{itemize}
\item Parse $\widetilde{\vk} =
(x,\theta,S=\{s_{i,0} \xor
s_{i,1}\}_{i\in[\ctlen]\;:\;\theta[i]=1})$ and $\widetilde{\qdk} =
(\rho_i)_{i\in[\ctlen]}$.

\item For every $i \in [\ctlen]$, measure $\rho_i$ in the Hadamard
basis to get outcomes $c_i, d_i$ corresponding to the registers
$\qreg{SKECD.CT_i}$ and $\qreg{S_i}$ respectively.

\item Compute $\cert[i] = c_i \xor d_i \cdot (s_{i,0} \xor
s_{i,1})$ for every $i \in [\ctlen]$.

\item If $x[i] = \cert[i]$ holds
for every $i \in [\ctlen] : \theta[i] = 1$, then update $V_k = \top$
and send $\top$ to $\qD$. Else, send $\bot$.
\end{itemize}

\item If $k = j+1$, $\qR$ simulates
the response  of oracle $\Oracle{\qVrfy}(k, \widetilde{\qdk})$ as
follows:

\begin{itemize}
\item Compute $\cert = \cert[1] \| \ldots \| \cert[\ctlen]$,
where each $\cert[i]$ is computed as in Step 9.

\item Send $\cert$ to $\qCh$. If $\qCh$ returns $\skecd.\sk$, send $\top$ to $\qD$ and update $V_{j+1}= \top$. Else if $\qCh$ returns $\bot$, send $\bot$ to $\qD$.
\end{itemize}

\item $\qD$ sends $(\msg_0^\star, \msg_1^\star) \in \bit^{\msglen} \times
\bit^{\msglen}$ to $\qR$. If $V_i = \bot$ for any $i \in [q]$, $\qR$
sends $0$ to $\qCh$.
$\qR$ computes $\ct^\star = (\skecd.\sk, r_1 \xor \msg^\star_\coin)$,
where $\skecd.\sk$ is obtained from $\qCh$ in Step 10. $\qR$ sends $\ct^*$ to $\qD$.

\item $\qD$ outputs a guess $b'$ which $\qR$ forwards to $\qCh$.
$\qCh$ outputs $b'$ as the final output of the experiment.
\end{enumerate}
\end{description}

We will first argue that when $b=1$, the view of $\qD$ is exactly
the same as its view in the hybrid $\Hyb_j^\coin$. Notice that the
reduction computes the first $j$ decryption keys by querying the
encryption oracle on random plaintexts. $\Hyb_j$ on the other hand,
directly computes them but there is no difference in the
output ciphertexts. A similar argument holds for the keys
$\qdk_{j+2}, \ldots, \qdk_q$, which contain encryptions of the same
random value $r_1$. Moreover, if $b=1$, the value encrypted as part
of the key $\qdk_{j+1}$ is also $r_1$. This is the same as in
$\Hyb_j^\coin$. As for the verification oracle queries, notice that
they are answered similarly by the reduction and $\Hyb_j^\coin$ for
all but the $j+1$-th key. For the $j+1$-th key, the reduction works
differently in that it forwards the certificate $\cert$ to the
verification oracle.  However, the verification procedure of the
BB84-based SKE-CD scheme checks the validity of the value $\cert$ in
the same way as the reduction, so there is no difference.

Finally, notice that when $b=0$, the encrypted value is random and
independent of $r_1$, similar to the hybrid $\Hyb_{j+1}^\coin$.
Consequently, $\qR$ breaks the IND-CVA-CD security of $\SKECD$ with
non-negligible probability, a contradiction.
\end{proof}

Notice now that the hybrid $\Hyb_0^\coin$ is the same as the
experiment $\expc{\SKECRSKL,\qA}{ot}{ind}{kla}\allowbreak(1^\secp,\coin)$. From
the previous lemma, we have that $\Hyb_0^\coin \approx
\Hyb_q^\coin$. However, we have that $\Hyb_q^0 \approx \Hyb_q^1$
because $\Hyb_q^0$ and $\Hyb_q^1$ do not encrypt $r$ at all as part
of the decryption keys, but they mask the plaintext with $r$.
Consequently, we have that $\Hyb_0^0 \approx \Hyb_0^1$, which
completes the proof. \qed

\subsection{Proof of Key-Testability}\label{proof:kt}\nikhil{Move to Appendix? Commented out to check Pg Limit.}

First, we will argue the correctness requirement. Recall that
$\SKECRSKL.\qKG$ applies the following map to a BB84 state
$\ket{x}_\theta$, where $(x, \theta) \in
\bit^{\ctlen}\times\bit^{\ctlen}$, for every $i \in [\ctlen]$:

\begin{align}
\ket{u_i}_{\qreg{SKECD.CT_i}}\tensor\ket{v_i}_{\qreg{S_i}}
\ra
\ket{u_i}_{\qreg{SKECD.CT_i}}\tensor\ket{v_i\oplus
s_{i,u_i}}_{\qreg{S_i}}
\end{align}
where $\qreg{SKECD.CT_i}$ denotes the register holding the $i$-th
qubit of $\ket{x}_\theta$ and $\qreg{S_i}$ is a register initialized
to $\ket{0\ldots0}_{\qreg{S_i}}$.

Consider applying the algorithm $\SKECRSKL.\KeyTest$ in
superposition to the resulting state, i.e., performing the following
map, where $\qreg{\SKECD.CT} = \qreg{\SKECD.CT_1} \otimes \cdots
\otimes \qreg{\SKECD.CT_{\ctlen}}$ and $\qreg{S} =
\qreg{S_1} \otimes \cdots \otimes \qreg{S_{\ctlen}}$, and
$\qreg{KT}$ is initialized to $\ket{0}$:

\begin{align}
\ket{u}_{\qreg{SKECD.
CT}}\tensor\ket{v}_{\qreg{S}}\tensor\ket{\beta}_{\qreg{KT}} \ra
\ket{u}_{\qreg{SKECD.
CT}}\tensor\ket{v}_{\qreg{S}}\tensor\ket{\beta\oplus\SKECRSKL.
\KeyTest(\tk, u\|v)}_{\qreg{KT}}
\end{align}

where $\tk = T = t_{1, 0}\|t_{1, 1} \| \cdots \| t_{\ctlen,
0}\|t_{\ctlen, 1}$. Recall that $\SKECRSKL.\KeyTest$
outputs 1 if
and only if $\Check[t_{i,0}, t_{i, 1}](u_i,
v_i) = 1$ for every $i
\in [\ctlen]$, where $u_i, v_i$ denote the states of the registers
$\qreg{SKECD.CT_i}$ and $\qreg{S_i}$ respectively.
Recall that $\Check[t_{i,0},t_{i,1}](u_i, v_i)$ computes $f(v_i)$
and checks if it equals $t_{i, u_i}$. Since the construction chooses
$t_{i, u_i}$ such that $f(s_{i, u_i}) = t_{i, u_i}$, this check
always passes. Consequently, measuring register $\qreg{KT}$ always
produces outcome $1$.

We will now argue that the security requirement holds by showing the
following reduction to the security of the OWF $f$. Let $\qA$ be an
adversary that breaks the key-testability of $\SKECRSKL$. Consider
a QPT reduction $\qR$ that works as follows in the OWF security
experiment:

\begin{description}
\item Execution of $\qR^\qA$ in
$\expa{f,\qR}{owf}(1^\secp)$:

\begin{enumerate}
\item The challenger chooses $s \leftarrow \bit^\lambda$ and sends
$y \seteq f(s)$ to $\qR$.

\item $\qR$ runs $\SKECRSKL.\Setup(1^\secp)$ and initializes $\qA$
with input $\msk$.

\item $\qA$ requests $q$ decryption keys for some polynomial $q$.
$\qR$ picks a random index $k^\star \in [q]$. For every $k \neq
k^\star$, $\qR$ generates $(\qdk_k, \vk_k, \tk_k)$ by computing the
function $f$ as needed. For the index $k^\star$, $\qR$ computes
$(\qdk_{k^\star}, \vk_{k^\star}, \tk_{k^\star})$ as follows (the
difference is colored in \textcolor{red}{red}):

\begin{enumerate}
\item Parse $\msk=(\skecd.\sk, r)$.
\item Generate
$(\skecd.\qct,\skecd.\vk)\gets\SKECD.\qEnc(\skecd.\sk,r)$.
$\skecd.\vk$ is of the form
$(x,\theta)\in\bit^{\ctlen}\times\bit^{\ctlen}$, 
$\skecd.\qct$ is of the form
$\ket{\psi_1}_{\qreg{SKECD.
CT_1}}\tensor\cdots\tensor\ket{\psi_{\ctlen}}_{\qreg{SKECD.
CT_{\ctlen}}}$.
\item \textcolor{red}{Choose an index $i^\star \in [\ctlen]$ such
that $\theta[i^\star] = 0$. For every $i \in [\ctlen]$ such
that $i \neq i^\star$, generate $s_{i,b}\la\bit^\secp$ and compute
$t_{i,b}\la f(s_{i,b})$ for every $b\in\bit$. For $i = i^\star$,
set $t_{i^\star, 1 - x[i^\star]} \seteq y$. Then, generate $s_{i^\star,
x[i^\star]} \leftarrow \bit^\lambda$ and compute $t_{i^\star,
x[i^\star]} = f(s_{i^\star, x[i^\star]})$.} Set $T\seteq
t_{1,0}\|t_{1,1}\|\cdots\|t_{\ctlen,0}\|t_{\ctlen,1}$ and $S =
\{s_{i,0} \xor s_{i, 1}\}_{i \in [\ctlen] \; : \;\theta[i] = 1}$.

\item Prepare a register $\qreg{S_i}$ that is initialized to
$\ket{0\cdots0}_{\qreg{S_i}}$ for every $i\in[\ctlen]$. 

\item For every $i\in[\ctlen]$, apply the map
\begin{align}
    \ket{u_i}_{\qreg{SKECD.CT_i}}\tensor\ket{v_i}_{\qreg{S_i}}
    \ra
    \ket{u_i}_{\qreg{SKECD.CT_i}}\tensor\ket{v_i\oplus
    s_{i,u_i}}_{\qreg{S_i}}
\end{align}

    to the registers $\qreg{SKECD.CT_i}$ and $\qreg{S_i}$ and obtain the resulting state $\rho_i$.
\item Compute $\qdk_{k^\star}=(\rho_i)_{i\in{[\ctlen]}}$,
$\vk_{k^\star}=(x,\theta,S)$, and $\tk_{k^\star}=T$.
\end{enumerate}

\item $\qR$ sends $(\qdk_i, \vk_i, \tk_i)$ to $\qA$ for every $i \in
[q]$.

\item $\qA$ sends $(k, \dk, \msg)$ to $\qR$. If $k \neq k^\star$, $\qR$
aborts. Otherwise, $\qR$ parses $\dk$ 
as a string over the registers $\qreg{SKECD.CT} = \qreg{SKECD.CT_1}
\otimes \cdots \otimes
\qreg{SKECD.CT_{\ctlen}}$ and $\qreg{S} = \qreg{S_1} \otimes
\cdots \otimes \qreg{S_{\ctlen}}$ and measures the register
$\qreg{S_{i^\star}}$ to obtain an outcome $s_{i^\star}$. $\qR$
then sends $s_{i^\star}$ to the challenger.
\end{enumerate}
\end{description}

Notice that the view of $\qA$ is the same as its view in the
key-testability experiment, as only the value $t_{i^\star,
1-x[i^\star]}$ is generated differently by forwarding the value $y$,
but this value is distributed identically to the original value.
Note that in both cases, $\qA$ receives no information about a
pre-image of $t_{i^\star, 1-x[i^\star]}$.
Now, $\qR$ guesses the index $k$ that $\qA$ targets with probability
$\frac1q$. By assumption, we have that $\CDec(\dk, \ct)
\neq \msg$ where $\ct = \Enc(\msk, \msg)$. The value $\dk$ can be
parsed as a string over the registers $\qreg{SKECD.CT}$ and
$\qreg{S}$. Let $\widetilde{\dk}$ be the sub-string of $\dk$ on the
register $\qreg{SKECD.CT}$. Recall that $\CDec$ invokes the
algorithm $\SKECD.\CDec$ on input $\widetilde{\dk}$. We will
now recall a property of $\SKECD.\CDec$ that was specified in
Definition \ref{def:bb84}:

Let $(\qct, \vk = (x, \theta)) \gets \SKECD.\Enc(\skecd.\sk, r)$
where $\skecd.\sk \gets \SKECD.\KG(1^\secp)$. Now, let $u$ be any
arbitrary value such that $u[i] = x[i]$ for all $i : \theta[i] = 0$.
Then, the following holds:

$$\Pr\Big[\SKECD.\CDec(\skecd.\sk, u) = r\Big] \ge 1 -
\negl(\secp)$$

Consequently, if $\widetilde{\dk}$ is such that $\widetilde{\dk}[i]
= x[i]$ for all $i: \theta[i] = 0$, where $(x, \theta)$ are
specified by $\vk_{k^\star}$, then $\SKECD.\CDec(\skecd.\sk,
\widetilde{\dk})$ outputs the value $r$ with high probability. Since
$\CDec(\dk, \ct = (\skecd.\sk, r \xor \msg))$ outputs $r \xor \msg \xor
\SKECD.\CDec(\skecd.\sk, \widetilde{\dk})$, we have that
$\CDec(\dk, \ct = (\skecd.\sk, r \xor \msg)) = \msg$. Therefore, it must
be the case that there exists some index $i$ for which
$\widetilde{\dk}[i] \neq x[i]$. With probability $\frac{1}{\ctlen}$,
this happens to be the guessed value $i^\star$.  In this case, $\qA$
must output $s_{i^\star}$ on register $\qreg{S_i}$ such that
$f(s_{i^\star}) = t_{i^\star, 1 - x[i^\star]} = y$. This concludes
the proof. \qed

Since we have proved OT-IND-KLA security (Section
\ref{proof:ot-ind}) and Key-Testability (Section
\ref{proof:kt}), we can now state the following theorem:

\fi


\section{PKE-CR-SKL from LWE}\label{sec:PKE-CR-SKL}

In this section, we show how to achieve PKE-CR-SKL from SKE-CR-SKL, standard ABE, and compute-and-compare obfuscation.

\subsection{Construction}
We construct a PKE-CR-SKL scheme
$\PKECRSKL=\PKECRSKL.(\Setup,\qKG,\allowbreak\Enc,\qDec,\qVrfy)$ with
message space $\cM = \bit^{\ell}$ using the following building blocks.

\begin{itemize}
\item ABE scheme $\ABE=\ABE.(\Setup,\KG,\Enc,\Dec)$ for the following relation $R$.
\begin{description}
\item[$R(x,y)$:] Interpret $x$ as a circuit. Then, output $0$ (decryptable) if $\bot=x(y)$ and otherwise $1$.
\end{description}
\item Compute-and-Compare Obfuscation $\CCObf$ with the simulator $\CCSim$.


\item SKE-CR-SKL scheme with Key Testability
$\SKECRSKL=\SKECRSKL.(\Setup,\allowbreak\qKG,\Enc,\qDec,\qVrfy,\KeyTest)$. It also has the classical decryption algorithm $\SKECRSKL.\CDec$.
\end{itemize} 

The construction is as follows.

\begin{description}
\item[$\PKECRSKL.\Setup(1^\secp)$:] $ $
\begin{itemize}
    \item Generate $(\abe.\pk,\abe.\msk)\gets\ABE.\Setup(1^\secp)$.
    \item Generate $\ske.\msk\gets\SKECRSKL.\Setup(1^\secp)$.
    \item Output $\ek\seteq\abe.\pk$ and $\msk\seteq(\abe.\msk,\ske.\msk)$.
\end{itemize}

\item[$\PKECRSKL.\qKG(\msk)$:] $ $
\begin{itemize}
    \item Parse $\msk=(\abe.\msk,\ske.\msk)$.
    \item Generate $(\ske.\qdk,\ske.\vk,\ske.\tk)\gets\SKECRSKL.\qKG(\ske.\msk)$. We denote the register holding $\ske.\qdk$ as $\qreg{SKE.DK}$.
    \item Prepare a register $\qreg{ABE.SK}$ that is initialized to $\ket{0\cdots0}_{\qreg{ABE.SK}}$.
    \item Choose explicit randomness $\key \leftarrow \bit^{\secp}$.
    \item Apply the map
        $\ket{u}_{\qreg{SKE.DK}}\ket{v}_{\qreg{ABE.SK}}\ra\ket{u}_{\qreg{SKE.DK}}\ket{v\oplus\ABE.\KG(\abe.\msk,u,
        \key)}_{\qreg{ABE.SK}}$ to the registers $\qreg{SKE.DK}$ and
        $\qreg{ABE.SK}$, and obtain $\qdk$ over the registers
        $\qreg{SKE.DK}$ and $\qreg{ABE.SK}$.

    \item Output $\qdk$ and
        $\vk\seteq(\abe.\msk,\ske.\vk,\ske.\tk, \key)$.
\end{itemize}

\item[$\PKECRSKL.\Enc(\ek,\msg)$:] $ $
\begin{itemize}
    \item Parse $\ek=\abe.\pk$.
    \item Generate $\tlC\gets\CCSim(1^\secp,\pp_D,1)$, where $\pp_D$ consists of circuit parameters of $D$ defined in the security proof.
    \item Generate $\abe.\ct\gets\ABE.\Enc(\abe.\pk,\tlC,\msg)$.
    \item Output $\ct\seteq\abe.\ct$.
\end{itemize}
 
\item[$\PKECRSKL.\qDec(\qdk,\ct)$:] $ $
\begin{itemize}
   \item Parse $\ct=\abe.\ct$. We denote the register holding $\qdk$ as $\qreg{SKE.DK}\tensor\qreg{ABE.SK}$.
   \item Prepare a register $\qreg{MSG}$ that is initialized to $\ket{0\cdots0}_{\qreg{MSG}}$
   \item Apply the map $\ket{v}_{\qreg{ABE.SK}}\ket{w}_{\qreg{MSG}}\ra\ket{v}_{\qreg{ABE.SK}}\ket{w\oplus\ABE.\Dec(v,\abe.\ct)}_{\qreg{MSG}}$ to the registers $\qreg{ABE.SK}$ and $\qreg{MSG}$.
   \item Measure the register $\qreg{MSG}$ in the computational basis and output the result $\msg^\prime$.
\end{itemize}

\item[$\PKECRSKL.\qVrfy(\vk,\qdk^\prime)$:] $ $
\begin{itemize}
    \item Parse $\vk=(\abe.\msk,\ske.\vk,\ske.\tk, \key)$. We denote the register holding $\qdk^\prime$ as $\qreg{SKE.DK}\tensor\qreg{ABE.SK}$.
    \item Prepare a register $\qreg{SKE.KT}$ that is initialized to $\ket{0}_{\qreg{SKE.KT}}$.
    \item Apply the map $\ket{u}_{\qreg{SKE.DK}}\ket{\beta}_{\qreg{SKE.KT}}\ra\ket{u}_{\qreg{SKE.DK}}\ket{\beta\oplus\SKECRSKL.\KeyTest(\ske.\tk,u)}_{\qreg{SKE.KT}}$ to the registers $\qreg{SKE.DK}$ and $\qreg{SKE.KT}$.
    \item Measure $\qreg{SKE.KT}$ in the computational basis and output $\bot$ if the result is $0$. Otherwise, go to the next step.
    \item Apply the map
        $\ket{u}_{\qreg{SKE.DK}}\ket{v}_{\qreg{ABE.SK}}\ra\ket{u}_{\qreg{SKE.DK}}\ket{v\oplus\ABE.\KG(\abe.\msk,u,
        \key)}_{\qreg{ABE.SK}}$ to the registers $\qreg{SKE.DK}$ and $\qreg{ABE.SK}$.
    \item Trace out the register $\qreg{ABE.SK}$ and obtain $\ske.\qdk^\prime$ over $\qreg{SKE.DK}$.
    \item Output $\top$ if $\top=\SKECRSKL.\qVrfy(\ske.\vk,\ske.\qdk')$
        and $\bot$ otherwise.

\end{itemize}
\end{description}

\paragraph{Decryption correctness.}
The key
$\qdk$ output by $\PKECRSKL.\qKG$ is of the form $\sum_u \alpha_u
\ket{u}_{\qreg{SKE.DK}}\ket{\abe.\sk_u}_{\qreg{ABE.SK}}$, where
$\abe.\sk_u\gets\ABE.\KG(\abe.\msk,u, \key)$.
Let $\ct\gets\PKECRSKL.\Enc(\ek,\msg)$.
On applying $\ket{v}_{\qreg{ABE.SK}}\ket{w}_{\qreg{MSG}}\ra\ket{v}_{\qreg{ABE.SK}}\ket{w\oplus\ABE.\Dec(v,\abe.\ct)}_{\qreg{MSG}}$ to $\sum_u \alpha_u \ket{u}_{\qreg{SKE.DK}}\ket{\abe.\sk_u}_{\qreg{ABE.SK}} \tensor \ket{0\cdots0}_{\qreg{MSG}}$, with overwhelming probability, the result is negligibly close to
\begin{align}
\sum_u \alpha_u \ket{u}_{\qreg{SKE.DK}}\ket{\abe.\sk_u}_{\qreg{ABE.SK}} \tensor \ket{\msg}_{\qreg{MSG}}
\end{align}
since $\tlC(u)=\bot$ and thus $R(\tlC,u)=0$ for $\tlC\gets\CCSim(1^\secp,\pp_D,1)$ and almost every string $u$.
Therefore, we see that $\PKECRSKL$ satisfies decryption correctness.

\paragraph{Verification correctness.}
Let $\qdk\gets\PKECRSKL.\qKG(\msk)$.
It is clear that the state $\ske.\qdk^\prime$ obtained when computing $\PKECRSKL.\qVrfy(\vk,\qdk)$ is the same as $\ske.\qdk$ generated when generating $\qdk$.
Therefore, the verification correctness of $\PKECRSKL$ follows from that of $\SKECRSKL$.

\subsection{Proof of IND-KLA Security}
Let $\qA$ be an adversary for the IND-KLA security of $\PKECRSKL$.
We consider the following sequence of experiments.
\begin{description}
\item[$\hybi{0}^\coin$:]This is $\expb{\PKECRSKL,\qA}{ind}{kla}(1^\secp,\coin)$.
\begin{enumerate}
\item The challenger $\qCh$ generates $(\abe.\pk,\abe.\msk)\gets\ABE.\Setup(1^\secp)$ and $\ske.\msk\gets\SKECRSKL.\Setup(1^\secp)$, and sends $\ek\seteq\abe.\pk$ to $\qA$.

\item $\qA$ requests $q$ decryption keys for some polynomial $q$.
$\qCh$ generates $\qdk_i$ as follows for every $i\in[q]$:

\begin{itemize}
    \item Generate $(\ske.\qdk_i,\ske.\vk_i,\ske.\tk_i)\gets\SKECRSKL.\qKG(\ske.\msk)$. We denote the register holding $\ske.\qdk_i$ as $\qreg{SKE.DK_i}$.
    \item Prepare a register $\qreg{ABE.SK_i}$ that is initialized to $\ket{0\cdots0}_{\qreg{ABE.SK_i}}$.
    \item Choose explicit randomness $\key_i \leftarrow \bit^\secp$.
    \item Apply the map
        $\ket{u}_{\qreg{SKE.DK_i}}\ket{v}_{\qreg{ABE.SK_i}}\ra\ket{u}_{\qreg{SKE.DK_i}}\ket{v\oplus\ABE.\KG(\abe.\msk,u,
        \key_i)}_{\qreg{ABE.SK_i}}$ to the registers $\qreg{SKE.DK_i}$ and $\qreg{ABE.SK_i}$, and obtain $\qdk_i$ over the registers $\qreg{SKE.DK_i}$ and $\qreg{ABE.SK_i}$.
\end{itemize}
$\qCh$ sends $\qdk_1,\ldots,\qdk_q$ to $\qA$.
\item $\qA$ can get access to the following (stateful) verification
oracle $\Oracle{\qVrfy}$ where $V_i$ is initialized to be $\bot$:

\begin{description}
\item[ $\Oracle{\qVrfy}(i,\widetilde{\qdk})$:] It computes $d$ as follows.  

\begin{enumerate}
    \item Let the register holding $\widetilde{\qdk}$ be $\qreg{SKE.DK_i}\tensor\qreg{ABE.SK_i}$.
    \item Prepare a register $\qreg{SKE.KT_i}$ that is initialized to $\ket{0}_{\qreg{SKE.KT_i}}$.
    \item Apply the map $\ket{u}_{\qreg{SKE.DK_i}}\ket{\beta}_{\qreg{SKE.KT_i}}\ra\ket{u}_{\qreg{SKE.DK_i}}\ket{\beta\oplus\SKECRSKL.\KeyTest(\ske.\tk_i,u)}_{\qreg{SKE.KT_i}}$ to the registers $\qreg{SKE.DK_i}$ and $\qreg{SKE.KT_i}$.
    \item Measure $\qreg{SKE.KT_i}$ in the computational basis and
        set $d\seteq\bot$ if the result is $0$. Otherwise, go to the next step.
    \item Apply the map
        $\ket{u}_{\qreg{SKE.DK_i}}\ket{v}_{\qreg{ABE.SK_i}}\ra\ket{u}_{\qreg{SKE.DK_i}}\ket{v\oplus\ABE.\KG(\abe.\msk,u,\key_i)}_{\qreg{ABE.SK_i}}$ to the registers $\qreg{SKE.DK_i}$ and $\qreg{ABE.SK_i}$.
    \item Trace out the register $\qreg{ABE.SK_i}$ and obtain $\ske.\qdk^\prime$ over $\qreg{SKE.DK_i}$.
    \item Set $d\seteq\top$ if
        $\top=\SKECRSKL.\qVrfy(\ske.\vk_i,\ske.\qdk')$ and set $d\seteq\bot$ otherwise.
It returns $d$ to $\qA$. Finally, if $V_i=\bot$ and $d=\top$, it updates $V_i\seteq \top$. 
\end{enumerate}
            \end{description}
            \item $\qA$ sends $(\msg_0^*,\msg_1^*)\in \cM^2$ to
                $\qCh$. If $V_i=\bot$ for some $i\in[q]$,
                $\qCh$ outputs $0$ as the final output of this
                experiment. Otherwise, $\qCh$ generates
                $\tlC^*\gets\CCSim(1^\secp, \pp_D, 1)$ and $\abe.\ct^*\gets\ABE.\Enc(\abe.\pk,\tlC^*,\msg_\coin^*)$, and sends $\ct^*\seteq\abe.\ct^*$ to $\qA$.
            \item $\qA$ outputs $\coin^\prime$. $\qCh$ outputs $\coin'$ as the final output of the experiment.
        \end{enumerate}
        
\item[$\hybi{1}^\coin$:]This is the same as $\hybi{0}^\coin$ except
that $\tlC^*$ is generated as
$\tlC^*\gets\CCObf(1^\secp,D[\ske.\ct^*], \lock, 0)$, where
$\ske.\ct^*\gets\SKECRSKL.\Enc(\ske.\msk,0^\secp)$,
$\lock\gets\bit^\secp$, and $D[\ske.\ct^*](x)$ is a circuit that has
$\ske.\ct^*$ hardwired and outputs $\SKECRSKL.\CDec(x,\ske.\ct^*)$.

\end{description}

We pick $\lock$ as a uniformly random string that is completely independent of other variables such as $\ske.\ct^*$.
Thus, from the security of $\CCObf$, we have $\Hyb_0^\coin \approx
\Hyb_1^\coin$.

\begin{description}
\item[$\hybi{2}^\coin$:]This is the same as $\hybi{1}^\coin$ except
that $\ske.\ct^*$ hardwired into the obfuscated circuit $\tlC^*$ is
generated as $\ske.\ct^*\gets\SKECRSKL.\Enc(\ske.\msk,\lock)$.
\end{description}

From the OT-IND-KLA security of $\SKECRSKL$, we can show that
$\hybi{1}^\coin \approx \hybi{2}^\coin$. Suppose that $\Hyb_1^\coin
\not \approx \Hyb_2^\coin$ and $\qD$ is a corresponding
distinguisher. We consider the following reduction
$\qR$:

\begin{description}
\item Execution of $\qR^{\qD}$ in
$\expc{\SKECRSKL,\qR}{ot}{ind}{kla}(1^\secp, b)$:

\begin{enumerate}
\item $\qCh$ runs $\ske.\msk \gets \SKECRSKL.\Setup(1^\secp)$ and initializes $\qR$
with the security parameter $1^\secp$.
\item $\qR$ generates $(\abe.\pk, \abe.\msk) \gets
\ABE.\Setup(1^\secp)$ and sends $\ek \seteq \abe.\pk$ to $\qD$.
\item $\qD$ requests $q$ decryption keys for some polynomial $q$.
$\qR$ requests $q$ decryption keys. $\qCh$ generates $(\ske.\qdk_i,
\ske.\vk_i, \ske.\tk_i) \gets \SKECRSKL.\qKG(\msk)$ for every $i \in [q]$
and sends $(\ske.\qdk_i, \ske.\tk_i)_{i\in[q]}$ to $\qR$.
\item $\qR$ computes $\qdk_1, \ldots, \qdk_q$ as in Step $2.$ of
$\Hyb_0^\coin$, except that the received values
$(\ske.\qdk_i)_{i \in [q]}$ are used instead of the original ones.
\item $\qR$ simulates the access to $\Oracle{\qVrfy}(i,
\widetilde{\qdk})$ for $\qD$ as follows:
\begin{description}
\item $\Oracle{\qVrfy}(i, \widetilde{\qdk}):$
\begin{enumerate}
\item Perform Step $3.(\textrm{a})$-$3.(\textrm{f})$ of $\Hyb_0^\coin$ to obtain
$\ske.\qdk'$, but using the received value $\ske.\tk_i$
instead of the original one.
\item Set $d \seteq \top$ if $\top =
\SKECRSKL.\Oracle{\qVrfy}(i, \ske.\qdk')$ and set $d \seteq \bot$
otherwise.
It returns $d$ to $\qD$. Finally, if $V_i = \bot$ and $d = \top$, it
updates $V_i = \top$.
\end{enumerate}
\end{description}
\item $\qR$ sends $(\ske.\msg_0^*, \ske.\msg_1^*) \seteq (0^\secp, \lock)$
to $\qCh$ and receives $\ske.\ct^* \gets \SKECRSKL.\Enc(\ske.\msk,
\ske.\msg^*_b)$.
\item $\qD$ sends $(\msg_0^*, \msg_1^*) \in \cM^2$ to $\qR$. If $V_i =
\bot$ for some $i\in[q]$, $\qR$ outputs $0$. Otherwise, $\qR$
generates $\tlC^* \gets \CCObf(1^\secp, D[\ske.\ct^*], \lock, 0)$
and $\abe.\ct^* \gets \ABE.\Enc(\abe.\pk, \tlC^*, \msg^*_\coin)$ and
sends $\ct^* \seteq \abe.\ct^*$ to $\qD$.
\item $\qD$ outputs a bit $b'$. $\qR$ outputs $b'$ and $\qCh$
outputs $b'$ as the final output of the experiment.
\end{enumerate}
\end{description}

It is easy to see that the view of $\qD$ is the same as that in
$\Hyb^\coin_2$ when $\lock$ is encrypted in $\ske.\ct^*$ and that of
$\Hyb^\coin_1$ when $0^\secp$ is encrypted. Moreover, for $\qD$ to
distinguish between the two hybrids, it must be the case that $V_i =
\top$ for all $i \in [q]$, which directly implies that the $q$
analogous values checked by $\SKECRSKL.\Oracle{\qVrfy}$ must also be
$\top$.  Consequently, $\qR$ breaks the OT-IND-KLA security of
$\SKECRSKL$. Therefore, it must be that $\Hyb_1^\coin \approx
\Hyb_2^\coin$.

\begin{description}
\item[$\hybi{3}^\coin$:]This is the same as $\hybi{2}^\coin$ except
that $\qdk_i$ is generated as follows for every $i\in[q]$. (The
difference is colored in red.)
\begin{itemize}
    \item Generate $(\ske.\qdk_i,\ske.\vk_i,\ske.\tk_i)\gets\SKECRSKL.\qKG(\ske.\msk)$. We denote the register holding $\ske.\qdk_i$ as $\qreg{SKE.DK_i}$.
    \textcolor{red}{
        \item Prepare a register $\qreg{ABE.R_i}$ that is initialized to $\ket{0}_{\qreg{ABE.R_i}}$.
    \item Apply the map $\ket{u}_{\qreg{SKE.DK_i}}\ket{\beta}_{\qreg{ABE.R_i}}\ra\ket{u}_{\qreg{SKE.DK_i}}\ket{\beta\oplus R(\tlC^*,u)}_{\qreg{ABE.R_i}}$
to the registers $\qreg{SKE.DK_i}$ and $\qreg{ABE.R_i}$. (Note that we
can generate $\tlC^*$ at the beginning of the game.)
    \item Measure $\qreg{ABE.R_i}$ in the computational basis and set $\qdk_i\seteq\bot$ if the result is $0$. Otherwise, go to the next step.
    }
    \item Prepare a register $\qreg{ABE.SK_i}$ that is initialized to $\ket{0\cdots0}_{\qreg{ABE.SK_i}}$.

    \item Sample explicit randomness $\key_i \gets \bit^\secp$.
    \item Apply the map
        $\ket{u}_{\qreg{SKE.DK_i}}\ket{v}_{\qreg{ABE.SK_i}}\ra\ket{u}_{\qreg{SKE.DK_i}}\ket{v\oplus\ABE.\KG(\abe.\msk,u,\key_i)}_{\qreg{ABE.SK_i}}$ to the registers $\qreg{SKE.DK_i}$ and $\qreg{ABE.SK_i}$, and obtain $\qdk_i$ over the registers $\qreg{SKE.DK_i}$ and $\qreg{ABE.SK_i}$.
\end{itemize}
\end{description}

From the decryption correctness of $\SKECRSKL$, the added procedure that checks $R(\tlC^*,u)$ in superposition does not affect the final state $\qdk_i$ with overwhelming probability since $R(\tlC^*,u)=1$ in this hybrid for any $u$ that appears in $\ske.\qdk_i$ when describing it in the computational basis.
Therefore, we have $\Hyb_2^\coin \approx \Hyb_3^\coin$.

\begin{description}
\item[$\hybi{4}^\coin$:]This is the same as $\hybi{3}^\coin$ except
that the oracle $\Oracle{\qVrfy}$ behaves as follows. (The
difference is colored in red.)
\begin{description}
\item[ $\Oracle{\qVrfy}(i,\widetilde{\qdk})$:] It computes $d$ as follows.  
\begin{enumerate}[(a)]
    \item Let the register holding $\widetilde{\qdk}$ be $\qreg{SKE.DK_i}\tensor\qreg{ABE.SK_i}$.
    \item Prepare a register $\qreg{SKE.KT_i}$ that is initialized to $\ket{0}_{\qreg{SKE.KT_i}}$.
    \item Apply the map $\ket{u}_{\qreg{SKE.DK_i}}\ket{\beta}_{\qreg{SKE.KT_i}}\ra\ket{u}_{\qreg{SKE.DK_i}}\ket{\beta\oplus\SKECRSKL.\KeyTest(\ske.\tk_i,u)}_{\qreg{SKE.KT_i}}$ to the registers $\qreg{SKE.DK_i}$ and $\qreg{SKE.KT_i}$.
    \item Measure $\qreg{SKE.KT_i}$ in the computational basis and set $d\seteq\bot$ if the result is $0$. Otherwise, go to the next step.
    \textcolor{red}{
    \item Prepare a register $\qreg{ABE.R_i}$ that is initialized to $\ket{0}_{\qreg{ABE.R_i}}$.
    \item Apply the map $\ket{u}_{\qreg{SKE.DK_i}}\ket{\beta}_{\qreg{ABE.R_i}}\ra\ket{u}_{\qreg{SKE.DK_i}}\ket{\beta\oplus R(\tlC^*,u)}_{\qreg{ABE.R_i}}$ to the registers $\qreg{SKE.DK_i}$ and $\qreg{ABE.R_i}$.
    \item Measure $\qreg{ABE.R_i}$ in the computational basis and set $d\seteq\bot$ if the result is $0$. Otherwise, go to the next step.
    }
    \item Apply the map
        $\ket{u}_{\qreg{SKE.DK_i}}\ket{v}_{\qreg{ABE.SK_i}}\ra\ket{u}_{\qreg{SKE.DK_i}}\ket{v\oplus\ABE.\KG(\abe.\msk,u,\key_i)}_{\qreg{ABE.SK_i}}$ to the registers $\qreg{SKE.DK_i}$ and $\qreg{ABE.SK_i}$.
    \item Trace out the register $\qreg{ABE.SK_i}$ and obtain $\ske.\qdk^\prime$ over $\qreg{SKE.DK_i}$.
    \item Set $d\seteq\top$ if $\top=\SKECRSKL.\qVrfy(\ske.\vk_i,
        \ske.\qdk')$ and set $d\seteq\bot$ otherwise.
    Return $d$ to $\qA$. Finally, if $V_i=\bot$ and $d=\top$, update $V_i\seteq \top$.
\end{enumerate}
            \end{description}
\end{description}


Suppose there exists a QPT distinguisher $\qD$ that has
non-negligible advantage in distinguishing $\Hyb_3^\coin$ and
$\Hyb_4^\coin$. Let $\qD$ make $q = \poly(\lambda)$ many queries to
the oracle $\Oracle{\qVrfy}(\cdot, \cdot)$. We will now consider the
following QPT algorithm $\qA_\oh$ with access to an oracle
$\Oracle{\qKG}$ and an oracle $\cO$ that runs $\qD$ as follows:

\begin{description}
\item $\underline{\qA_\oh^{\Oracle{\qKG}, \cO}(\abe.\pk, \ske.\msk)}:$
\begin{enumerate}
\item $\qA_\oh$ initializes $\qD$ with input $\ek = \abe.\pk$.
\item When $\qD$ requests $q$ decryption keys, $\qA_\oh$ queries
$\Oracle{\qKG}$ on input $q$ and receives $(\{\qdk_i\}_{i \in [q]}, \{\vk_i\}_{i \in [q]})$. It forwards $\{\qdk_i\}_{i\in [q]}$ to $\qD$.

\item $\qA_\oh$ simulates the access of $\Oracle{\qVrfy}$ for $\qD$ as
follows:
\begin{description}
\item $\underline{\Oracle{\qVrfy}(y, \widetilde{\qdk})}:$
\begin{enumerate}
\item Execute Steps (a)-(e) of $\Oracle{\qVrfy}$ as in
    $\Hyb_4^\coin$.
    \item Apply the map
        $\ket{u}_{\qreg{SKE.DK_y}}\ket{\beta}_{\qreg{ABE.R_y}}\ra\ket{u}_{\qreg{SKE.DK_y}}\ket{\beta\oplus
        \calO(u)}_{\qreg{ABE.R_y}}$ to the registers
        $\qreg{SKE.DK_y}$ and $\qreg{ABE.R_y}$. 
\item Execute Steps (g)-(j) of $\Oracle{\qVrfy}$ as in
$\Hyb_4^\coin$.
\end{enumerate}
\end{description}

\item $\qD$ sends $(\msg_0^*, \msg_1^*) \in \cM^2$ to $\qA_\oh$. If $V_i =
\bot$ for some $i \in [q]$, $\qA_\oh$ outputs $0$. Otherwise,
$\qA_\oh$ generates $\tlC^* \gets \CCObf(1^\secp, D[\ske.\ct^*], \lock,
0)$, where $\ske.\ct^* \gets \SKECRSKL.\Enc(\ske.\msk, \lock)$.
It then generates $\abe.\ct^* \gets\ABE.\Enc(\abe.\pk, \tlC^*, \msg^*_\coin)$ and sends $\ct^* \seteq\abe.\ct^*$ to $\qD$.
\item $\qD$ outputs a guess $b'$. $\qA_\oh$ outputs $b'$.
\end{enumerate}
\end{description}

Let $H$ be an oracle that for every input $u$, outputs $1$.
Consider now the extractor $\qB_\oh^{\Oracle{\qKG}, H}$
as specified by the O2H Lemma (Lemma \ref{lem:O2H}). We will now
construct a reduction $\qR$ that runs $\qB_\oh$ by simulating the oracles
$\Oracle{\qKG}$ and $H$ for $\qB_\oh$, and breaks the
key-testability of the $\SKECRSKL$ scheme.

\begin{description}
\item Execution of $\qR$ in
$\expb{\SKECRSKL,\qR}{key}{test}(1^\secp)$:

\begin{enumerate}
\item The challenger $\qCh$ runs $\ske.\msk \leftarrow
\SKECRSKL.\Setup(1^\secp)$ and
initializes $\qR$ with input $\ske.\msk$.
\item $\qR$ samples $(\abe.\pk, \abe.\msk) \gets
\ABE.\Setup(1^\secp)$ and initializes $\qB_\oh$ with the input
$(\abe.\pk, \ske.\msk)$.
\item When $\qB_\oh$ queries input $q$ to $\Oracle{\qKG}$, $\qR$
generates the decryption-keys $\qdk_1, \ldots, \qdk_q$ in the same way
as in $\Hyb_3^\coin$ and sends them to $\qB_\oh$.

\item When $\qB_\oh$ queries an input $u$ to $H$, $\qR$ responds with $1$.

\item $\qB_\oh$ outputs measured index $y$ and measurement outcome $\dk$. $\qR$ sends $(y, \dk, \lock)$
to $\qCh$.

\end{enumerate}
\end{description}

We will now claim that with non-negligible probability, $\qR$ obtains
values $\dk$ and $y$ such that $R(\tlC^*, \dk) = 0$. By the definition of
$R$ and $\tlC^*$ and the decryption correctness of $\SKECRSKL$, this
will imply that $\SKECRSKL.\CDec(\dk, \allowbreak\ske.\ct^\star) \neq
\lock$. Moreover, $\KeyTest(\ske.\tk_y, \dk)$ also holds.
Consequently, this will imply $\qR$ breaks the key-testability of
$\SKECRSKL$. To prove this, we will rely on the One-Way to Hiding
(O2H) Lemma (Lemma \ref{lem:O2H}).  Consider an oracle $G$ which takes
as input $u$ and outputs $R(\tlC^*, u)$ and an oracle $H$ which takes
as input $u$ and outputs $1$. Notice that if the oracle $\calO = G$,
then the view of $\qD$ as run by $\qA_\oh$ is the same as in
$\Hyb_4^\coin$, while if $\calO = H$, the view of $\qD$ is the same as
in $\Hyb_3^\coin$. By the O2H Lemma, we have the following, where $z =
(\abe.\pk, \ske.\msk)$.

\begin{align}
\abs{\Pr[\qA_\oh^{\Oracle{\qKG},
H}(z)=1]-\Pr[\qA_\oh^{\Oracle{\qKG}, G}(z)=1]} \leq
2q\cdot\sqrt{\Pr[\qB_\oh^{\Oracle{\qKG}, H}(z)\in S]}
\enspace.
\end{align}
where $S$ is a set where
the oracles $H$ and $G$ differ, which happens only for inputs $u$ s.t.
$R(\tlC^*, u) = 0$. Since $\qR$ obtains $\dk$ and $y$ as the
output of $\qB_\oh^{\Oracle{\qKG},H}(z)$, the argument
follows that $\Hyb_3^\coin \approx \Hyb_4^\coin$.

\begin{description}
\item[$\hybi{5}^\coin$:]This is the same as $\hybi{4}^\coin$ except that $\ct^*\seteq\abe.\ct^*$ is generated as $\abe.\ct^*\gets\ABE.\Enc(\abe.\pk,\tlC^*,0^{\msglen})$.
\end{description}

The view of $\qA$ in $\hybi{4}^\coin$ and $\hybi{5}^\coin$ can be
simulated with $\abe.\pk$ and the access to the quantum key
generation oracle $\Oracle{qkg}$. This is because before $\ABE.\KG$
is required to be applied in both the generation of $\{\qdk_i\}_{i
\in [q]}$ and to compute the responses of $\Oracle{\qVrfy}$, the
relation check $R(\tlC^*, u)$ is already applied in
superposition. Thus, we have $\Hyb_4^\coin \approx \Hyb_5^\coin$.

Lastly, $\hybi{5}^{0}$ and $\hybi{5}^{1}$ are exactly the same experiment and thus we have $\abs{\Pr[\hybi{5}^0=1]-\Pr[\hybi{5}^1=1]}=\negl(\secp)$.
Then, from the above arguments, we obtain
\begin{align}
   & \abs{\Pr[\expb{\PKECRSKL,\qA}{ind}{kla}(1^\secp,0)=1]-\Pr[\expb{\PKECRSKL,\qA}{ind}{kla}(1^\secp,1)=1]}\\
=&\abs{\Pr[\hybi{0}^0=1]-\Pr[\hybi{0}^1=1]}\le\negl(\secp).
\end{align}
This completes the proof. \qed

Given the fact that SKE-CR-SKL with Key-Testability (implied by
BB84-based SKE-CD and OWFs), Compute-and-Compare Obfuscation, and
Ciphertext-Policy ABE for General Circuits are all implied by the
LWE assumption, we state the following theorem:

\begin{theorem}
There exists a PKE-CR-SKL scheme satisfying IND-KLA security,
assuming the polynomial hardness of the LWE assumption.
\end{theorem}

\nikhil{At this point, we can probably just state the ABE-CR-SKL and ABE-CR\textsuperscript{2}-SKL theorems and move them to the Appendix.}

\ifnum\submission=0

\newcommand{\decryptable}{\mathsf{decryptable}}
\newcommand{\undecryptable}{\mathsf{undecryptable}}

\section{ABE-CR-SKL from LWE}\label{sec:ABE-SKL}
In this section, we show how to achieve ABE-CR-SKL from the LWE assumption.
To this end, we also introduce SKFE-CR-SKL with classical decryption and key-testability.
First, we recall the standard SKFE.

\begin{definition}[Secret-Key Functional Encryption]\label{def:SKFE}
An SKFE scheme $\SKFE$ for the functionality $F: \cX \times
\cY \ra \cZ$ is a tuple of four PPT algorithms $(\Setup, \KG, \Enc, \Dec)$. 
\begin{description}
\item[$\Setup(1^\secp)\ra\msk$:] The setup algorithm takes a security parameter $1^\lambda$, and outputs a master secret key $\msk$.
\item[$\KG(\msk,y)\ra\sk_y$:] The key generation algorithm takes a master secret key $\msk$ and a function $y \in \cY$, and outputs a functional decryption key $\sk_y$.

\item[$\Enc(\msk,x)\ra\ct$:] The encryption algorithm takes a master secret key $\msk$ and a plaintext $x \in \cX$, and outputs a ciphertext $\ct$.

\item[$\Dec(\sk_y,\ct)\ra z$:] The decryption algorithm takes a
functional decryption key $\sk_y$ and a ciphertext $\ct$, and
outputs  $z \in \{ \bot \} \cup \cZ$.

\item[Correctness:] We require that for every $x \in \cX$ and $y\in\cY$, we have that

\[
\Pr\left[
\Dec(\sk_y, \ct) = F(x,y)
 \ :
\begin{array}{rl}
 &\msk \la \Setup(1^\secp)\\
 & \sk_y \gets \KG(\msk,y) \\
 &\ct \gets \Enc(\msk,x)
\end{array}
\right] \ge 1 - \negl(\secp).
\]
\end{description}
\end{definition}

\begin{definition}[Selective Single-Ciphertext Security]\label{def:sel-1ct-SKFE}
We formalize the experiment
$\expb{\SKFE,\qA}{adp}{ind}(1^\secp,\coin)$ between an adversary
$\qA$ and a challenger for an SKFE scheme for the functionality $F:\cX\times\cY\ra\cZ$ as follows:
        \begin{enumerate}
            \item Initialized with $1^\secp$, $\qA$ outputs $(x_0^*,x_1^*)$. The challenger $\Ch$ runs $\msk\gets\Setup(1^\secp)$ and sends $\ct^*\gets\Enc(\msk,x_\coin^*)$ to $\qA$.
            \item $\qA$ can get access to the following oracle.
            \begin{description}
            \item[$\Oracle{\qKG}(y)$:] Given $y$, if $F(x_0^*,y)\ne F(x_1^*,y)$, returns $\bot$. Otherwise, it returns $\sk_y\gets\KG(\msk,y)$.
            \end{description}
            \item $\qA$ outputs a guess $\coin^\prime$ for $\coin$. 
            $\qCh$ outputs $\coin'$ as the final output of the experiment.
        \end{enumerate}
        
We say that $\SKFE$ satisfies selective single-ciphertext security if, for any QPT $\qA$, it holds that
\begin{align}
\advb{\SKFE,\qA}{sel}{1ct}(1^\secp) \seteq \abs{\Pr[\expb{\SKFE,\qA}{sel}{1ct} (1^\secp,0) \ra 1] - \Pr[\expb{\SKFE,\qA}{sel}{1ct} (1^\secp,1) \ra 1] }\le \negl(\secp).
\end{align}
\end{definition}

\begin{theorem}[\cite{C:GorVaiWee12}]\label{thm:1ct_adaptive_function_private_SKFE}
    Assuming the existence of OWFs, there exists selective single-ciphertext secure SKFE.
\end{theorem}

\subsection{Definitions of SKFE-CR-SKL}

We consider the classical decryption property and key-testability for
SKFE-CR-SKL as in the case of SKE-CR-SKL.
\begin{definition}[SKFE-CR-SKL]
An SKFE-CR-SKL scheme $\SKFESKL$ for the functionality $F: \cX \times
\cY \ra \cZ$ is a tuple
of five algorithms $(\Setup,\qKG, \Enc, \allowbreak\qDec,\qVrfy)$.
\begin{description}
\item[$\Setup(1^\secp)\ra\msk$:] The setup algorithm takes a security parameter $1^\lambda$ and a master secret key $\msk$.

\item[$\qKG(\msk,y)\ra(\qsk_y,\vk,\tk)$:] The key generation algorithm
takes a master secret key $\msk$ and a string $y \in \cY$, and
outputs a functional secret key $\qsk_y$, a certificate verification
key $\vk$, and a testing key $\tk$.

\item[$\Enc(\msk,x)\ra\ct$:] The encryption algorithm takes a master
secret key $\msk$ and a string $x \in \cX$, and outputs a ciphertext $\ct$.

\item[$\qDec(\qsk_y,\ct)\ra z$:] The decryption algorithm takes a
functional secret key $\qsk_y$ and a ciphertext $\ct$, and outputs a
value $z \in \cZ \cup \{\bot\}$.


\item[$\qVrfy(\vk,\widetilde{\qsk}_y)\ra\top/\bot$:] The verification algorithm takes a verification key $\vk$ and a (possibly malformed) functional secret key $\widetilde{\qsk}_y$, and outputs $\top$ or $\bot$.


\item[Decryption correctness:] For all $x\in\cX$ and
$y\in\cY$, we have
\begin{align}
\Pr\left[
\qDec(\qsk_y, \ct) \allowbreak = F(x,y)
\ :
\begin{array}{ll}
\msk\gets\Setup(1^\secp)\\
(\qsk_y,\vk,\tk)\gets\qKG(\msk,y)\\
\ct\gets\Enc(\msk,x)
\end{array}
\right] 
\ge 1-\negl(\secp).
\end{align}


\item[Verification correctness:] We have 
\begin{align}
\Pr\left[
\qVrfy(\vk,\qsk_y)=\top
\ :
\begin{array}{ll}
\msk\gets\Setup(1^\secp)\\
(\qsk_y,\vk,\tk)\gets\qKG(\msk,y)
\end{array}
\right] 
\ge1-\negl(\secp).
\end{align}

\end{description}
\end{definition}

\begin{definition}[Classical Decryption Property]\label{def:classical-Dec-SKFE}
We say that $\SKFESKL=(\Setup,\qKG,\Enc,\qDec,\qVrfy)$ has
the classical decryption property if there exists a deterministic
polynomial time algorithm $\CDec$ such that given $\qsk_y$ in the
register $\qreg{SK}$ and ciphertext $\ct$, $\qDec$ applies the map
$\ket{u}_{\qreg{SK}}\ket{v}_{\qreg{OUT}}\ra\ket{u}_{\qreg{SK}}\ket{v\oplus\CDec(u,\ct)}_{\qreg{OUT}}$
and outputs the measurement result of the register $\qreg{OUT}$ in
the computational basis, where $\qreg{OUT}$ is initialized to
$\ket{0\cdots0}_{\qreg{OUT}}$.
\end{definition}
\begin{definition}[Key-Testability]\label{def:key-testability-SKFE}
We say that an SKFE-CR-SKL scheme $\SKFESKL$ with the classical
decryption property satisfies key testability, if there exists a
classical deterministic algorithm $\KeyTest$ that satisfies the
following conditions:

\begin{itemize}
\item \textbf{Syntax:} $\KeyTest$ takes as input a testing key $\tk$
and a classical string $\sk$ as input. It outputs $0$ or $1$.

\item \textbf{Correctness:} Let $\msk\gets\Setup(1^\secp)$ and
$(\qsk_y,\vk,\tk)\gets\qKG(\msk,y)$ where $y$ is a string. We denote the register holding
$\qsk_y$ as $\qreg{SK}$. Let $\qreg{KT}$ be a register that is
initialized to $\ket{0}_{\qreg{KT}}$. If we apply the map
$\ket{u}_{\qreg{SK}}\ket{\beta}_{\qreg{KT}}\ra\ket{u}_{\qreg{SK}}\ket{\beta\oplus\KeyTest(\tk,u)}_{\qreg{KT}}$
to the registers $\qreg{SK}$ and $\qreg{KT}$ and then measure
$\qreg{KT}$ in the computational basis, we obtain $1$ with
overwhelming probability.

\item \textbf{Security:} Consider the following experiment
$\expb{\SKFESKL,\qA}{key}{test}(1^\secp)$.

\begin{enumerate}
\item The challenger $\qCh$ runs $\msk\gets\Setup(1^\secp)$ and
initializes $\qA$ with input $\msk$. 

\item $\qA$ can get access to the following oracle, where the list $\List{\qKG}$ used by the oracles is initialized
to an empty list.
\begin{description}
\item[$\Oracle{\qKG}(y)$:] Given $y$, it finds an entry of the form
$(y,\tk)$ from $\List{\qKG}$. If there is such an entry, it
returns $\bot$.

Otherwise, it generates $(\qsk_y,\vk,\tk)\la\qKG(\msk,y)$, sends
$(\qsk_y,\vk,\tk)$ to $\qA$, and adds $(y,\tk)$ to $\List{\qKG}$.
\end{description}

\item $\qA$ sends a tuple of classical strings $(y, \sk, x)$ to $\qCh$.
$\qCh$ outputs $\bot$ if there is no entry of the form $(y,\tk)$ in $\List{\qKG}$ for some $\tk$.
Otherwise, $\qCh$ generates $\ct\gets\Enc(\msk,x)$, and outputs $\top$ if
$\KeyTest(\tk,\sk)=1$ and $\CDec(\sk,\ct)\ne F(x,y)$, and outputs $\bot$ otherwise.
\end{enumerate}

For all QPT $\qA$, the following must hold:

\begin{align}
\advb{\SKFESKL,\qA}{key}{test}(1^\secp) \seteq
\Pr[\expb{\SKFESKL,\qA}{key}{test}(1^\secp) \ra \top] \le
\negl(\secp).
\end{align} 
\end{itemize}
\end{definition}

\begin{definition}[Selective Single-Ciphertext KLA Security]\label{def:sel-1ct-security_SKFE-SKL}
We say that an SKFE-CR-SKL scheme $\SKFESKL$ is selective single-ciphertext secure, if
it satisfies the following requirement, formalized from the
experiment $\expc{\SKFESKL,\qA}{sel}{1ct}{kla}(1^\secp,\coin)$
between an adversary $\qA$ and a challenger:

\begin{enumerate}
\item Initialized with $1^\secp$, $\qA$ outputs $(x_0^*,x_1^*)$. The challenger $\qCh$ runs $\msk\gets\Setup(1^\secp)$. 

\item $\qA$ can get access to the following (stateful) oracles,
where the list $\List{\qKG}$ used by the oracles is initialized
to an empty list:

\begin{description}
\item[$\Oracle{\qKG}(y)$:] Given $y$, it finds an entry of the form
$(y,\vk,V)$ from $\List{\qKG}$. If there is such an entry, it
returns $\bot$.

Otherwise, it generates $(\qsk,\vk,\tk)\la\qKG(\msk,y)$, sends
$\qsk$ and $\tk$ to $\qA$, and adds $(y,\vk,\bot)$ to $\List{\qKG}$.

\item[$\Oracle{\qVrfy}(y,\widetilde{\qsk})$:] Given
$(y,\widetilde{\qsk})$, it finds an entry $(y,\vk,V)$ from
$\List{\qKG}$. (If there is no such entry, it returns $\bot$.) It
then runs $\decision \gets \qVrfy(\vk,\widetilde{\qsk})$ and returns
$\decision$ to $\qA$. If $V=\bot$, it updates the entry into
$(y,\vk,\decision)$. 
\end{description}

\item $\qA$ requests the challenge ciphertext. If there exists
an entry $(y,\vk,V)$ in $\List{\qKG}$ such that $F(x_0^*,y)\ne
F(x_1^*,y)$ and $V=\bot$, $\qCh$ outputs $0$ as the
final output of this experiment. Otherwise, $\qCh$ generates
$\ct^*\la\Enc(\msk,x_\coin^*)$ and sends $\ct^*$ to $\qA$.

\item $\qA$ continues to make queries to $\Oracle{\qKG}$. However, $\qA$ is not allowed to send $y$ such that $F(x_0^*,y)\ne F(x_1^*,y)$ to $\Oracle{\qKG}$.

\item $\qA$ outputs a guess $\coin^\prime$ for $\coin$. $\qCh$
outputs $\coin'$ as the final output of the experiment.
\end{enumerate}

For any QPT $\qA$, it holds that

\begin{align}
\advc{\SKFESKL,\qA}{sel}{1ct}{kla}(1^\secp) \seteq \abs{\Pr[\expc{\SKFESKL,\qA}{sel}{1ct}{kla} (1^\secp,0) \ra 1] - \Pr[\expc{\SKFESKL,\qA}{sel}{1ct}{kla} (1^\secp,1) \ra 1] }\leq \negl(\secp).
\end{align} 
\end{definition}

In Appendix \ref{sec:SKFESKL-KT}, we prove the following theorem:

\begin{theorem}
Assuming the existence of a BB84-based SKE-CD scheme and the existence
of OWFs, there exists a selective single-ciphertext KLA secure
SKFE-CR-SKL scheme satisfying the key-testability property.
\end{theorem}

\subsection{Definitions of ABE-CR-SKL}\label{sec:ABE_SKL_def}

In this section, we recall definitions of ABE-CR-SKL by Agrawal et al.~\cite{EC:AKNYY23}.

\begin{definition}[ABE-CR-SKL]
An ABE-CR-SKL scheme $\ABESKL$ is a tuple of five algorithms $(\Setup,
\qKG, \Enc, \qDec,\allowbreak\qVrfy)$.
Below, let $\cX = \{ \cX_\secp \}_\secp$, $\cY= \{ \cY_\secp
\}_\secp$, and $R= \{ R_\secp: \cX_\secp \times \cY_\secp \to \bit
\}_\secp$ be the ciphertext attribute space, the key attribute
space, and the associated relation of $\ABESKL$, respectively. Let
$\cM$ denote the message space.

\begin{description}
\item[$\Setup(1^\secp)\ra(\pk,\msk)$:] The setup algorithm takes a security parameter $1^\lambda$, and outputs a public key $\pk$ and master secret key $\msk$.
\item[$\qKG(\msk,y)\ra(\qsk_y,\vk)$:] The key generation algorithm
    takes a master secret key $\msk$ and a key attribute $y \in
    \calY$, and outputs a user secret key $\qsk_y$ and a verification key $\vk$.


\item[$\Enc(\pk,x,\msg)\ra\ct$:] The encryption algorithm takes a public key $\pk$, a ciphertext attribute $x\in \cX$, and a plaintext $\msg\in\cM$, and outputs a ciphertext $\ct$.

\item[$\qDec(\qsk_y,\ct)\ra \msg^\prime$:] The decryption algorithm
takes a user secret key $\qsk_y$ and a ciphertext $\ct$. It outputs
a value $\msg^\prime\in \{\bot\}\cup \cM$.


\item[$\qVrfy(\vk,\qsk^\prime)\ra\top/\bot$:] The verification algorithm takes a verification key $\vk$ and a quantum state $\qsk^\prime$, and outputs $\top$ or $\bot$.

\item[Decryption correctness:]For every $x \in \cX$ and $y \in \cY$
satisfying $R(x,y)=0$ and $\msg\in\cM$, we have
\begin{align}
\Pr\left[
\qDec(\qsk_y, \ct) \allowbreak = \msg
\ \middle |
\begin{array}{ll}
(\pk,\msk) \la \Setup(1^\secp)\\
(\qsk_y,\vk)\gets\qKG(\msk,y)\\
\ct\gets\Enc(\pk,x,\msg)
\end{array}
\right] 
\ge1-\negl(\secp).
\end{align}

\item[Verification correctness:] For every $y \in \cY$, we have 
\begin{align}
\Pr\left[
\qVrfy(\vk,\qsk_y)=\top
\ \middle |
\begin{array}{ll}
(\pk,\msk) \la \Setup(1^\secp)\\
(\qsk_y,\vk)\gets\qKG(\msk,y)\\
\end{array}
\right] 
\ge1-\negl(\secp).
\end{align}

\end{description}
\end{definition}

\begin{definition}[Adaptive IND-KLA Security]\label{def:ada_lessor_ABESKL}
We say that an ABE-CR-SKL scheme $\ABESKL$ for relation $R:\cX\times
\cY \to \bin$ is adaptively IND-KLA secure, if it satisfies the
following requirement, formalized from the experiment
$\expc{\ABESKL, \qA}{ada}{ind}{kla}(1^\secp,\coin)$ between an
adversary $\qA$ and a challenger $\qCh$:
        \begin{enumerate}
            \item At the beginning, $\qCh$ runs $(\pk,\msk)\gets\Setup(1^\secp)$
            and initializes the list $\List{\qKG}$ to be an empty set. 
            Throughout the experiment, $\qA$ can access the following oracles.
            \begin{description}
            \item[$\Oracle{\qKG}(y)$:] Given $y$, it finds an entry of the form $(y,\vk,V)$ from $\List{\qKG}$. If there is such an entry, it returns $\bot$.
            Otherwise, it generates $(\qsk_y,\vk)\la\qKG(\msk,y)$,
            sends $\qsk_y$ to $\qA$, and adds $(y,\vk,\bot)$ to $\List{\qKG}$.
            
            \item[$\Oracle{\qVrfy}(y,\widetilde{\qsk})$:] Given $(y,\widetilde{\qsk})$, it finds an entry $(y,\vk,V)$ from $\List{\qKG}$. (If there is no such entry, it returns $\bot$.) 
            It then runs $\decision \seteq \qVrfy(\vk,
            \widetilde{\qsk})$ and returns $\decision$ to $\qA$.
            If $V=\bot$, it updates the entry into $(y,\vk,\decision)$.
            \end{description}
            \item \label{ada_lessor_abe_challenge}
            When $\qA$ sends $(x^*,\msg_0,\msg_1)$ to $\qCh$, 
            it checks that for every entry $(y,\vk,V)$ in 
            $\List{\qKG}$ such that $R(x^*,y)=0$, it holds
            that $V=\top$. If so,
            $\qCh$ generates $\ct^*\la\Enc(\pk,x^*,\msg_\coin)$ and
            sends $\ct^*$ to $\qA$. Otherwise, it outputs $0$.
            \item 
            $\qA$ continues to make queries to $\Oracle{\qKG}(\cdot)$ and  $\Oracle{\qVrfy}(\cdot,\cdot )$.
            However, $\qA$ is not allowed to send a key attribute
            $y$ such that $R(x^*,y)=0$ to $\Oracle{\qKG}$.
            \item $\qA$ outputs a guess $\coin^\prime$ for $\coin$. 
            $\qCh$ outputs $\coin'$ as the final output of the experiment.
        \end{enumerate}
        For any QPT $\qA$, it holds that
\ifnum\llncs=0        
\begin{align}
\advc{\ABESKL,\qA}{ada}{ind}{kla}(1^\secp) \seteq \abs{\Pr[\expc{\ABESKL,\qA}{ada}{ind}{kla} (1^\secp,0) \ra 1] - \Pr[\expc{\ABESKL,\qA}{ada}{ind}{kla} (1^\secp,1) \ra 1] }\leq \negl(\secp).
\end{align}
\else
\begin{align}
\advb{\PKFESKL,\qA}{ada}{lessor}(\secp) 
&\seteq \abs{\Pr[\expb{\ABESKL,\qA}{ada}{lessor} (1^\secp,0) \ra 1] - \Pr[\expb{\ABESKL,\qA}{ada}{lessor} (1^\secp,1) \ra 1] }\\
&\leq \negl(\secp).
\end{align}
\fi
\end{definition}

\begin{remark}\label{remark:same_query_remark}
Although we can handle the situation where multiple keys for the same attribute $y$ are generated using an index management such as $(y,1,vk_1,V_1)$, $(y,2,vk_2,V_2)$, we use the simplified definition as Agrawal et al.~\cite{EC:AKNYY23} did.\ryo{I avoided copy-and-paste from AKNYY23.}
\end{remark}

We also consider relaxed versions of the above security notion. 

\begin{definition}[Selective IND-KLA
Security]\label{def:sel_ind_ABE_SKL} We consider the same security
game as that for adaptive IND-KLA security except that the adversary
$\qA$ should declare its target $x^*$ at the beginning of the game
(even before it is given $\pk$).

We then define the advantage
$\advc{\ABESKL,\qA}{sel}{ind}{kla}(1^\secp)$ for the selective case
similarly. We say $\ABESKL$ is selectively IND-KLA secure if for any
QPT adversary $\qA$, $\advc{\ABESKL,\qA}{sel}{ind}{kla}(1^\secp)$ is
negligible.    
\end{definition}

\subsection{Construction}
We construct an ABE-CR-SKL scheme
$\ABESKL=\ABESKL.(\Setup,\qKG,\allowbreak\Enc,\qDec,\qVrfy)$ for the relation
$R$ with the message space $\cM \seteq \bit^{\msglen}$ using the following building blocks.
\begin{itemize}
\item ABE scheme $\ABE=\ABE.(\Setup,\KG,\Enc,\Dec)$ for the following relation $R^\prime$.
\begin{description}
\item[$R^\prime(x^\prime,y^\prime)$:]Interpret $x^\prime\seteq x\|C$ and
    $y^\prime\seteq y\|z$, where $C$ is a circuit. Then, output $0$ if
    $R(x,y)=0$ and $C(z)=\bot$, and otherwise output $1$.
\end{description}

\item Compute-and-Compare Obfuscation $\CCObf$ with the simulator $\CCSim$.
\item SKFE-CR-SKL scheme with key-testability
$\SKFESKL=\SKFESKL.\allowbreak(\Setup,\qKG,\Enc,\qDec,\qVrfy,\allowbreak\KeyTest)$ for
the following functionality $F$. It also has the classical
decryption algorithm $\SKFESKL.\CDec$.

\begin{description}
\item[$F(\widetilde{x},\widetilde{y})$:]Interpret
    $\widetilde{x}\seteq x\|z$ and $\widetilde{y}\seteq y$. Then, output $z$
    if $R(x,y)=0$, and otherwise output $\bot$.
\end{description}

\end{itemize} 

The construction is as follows.

\begin{description}
\item[$\ABESKL.\Setup(1^\secp)$:] $ $
\begin{itemize}
    \item Generate $(\abe.\pk,\abe.\msk)\gets\ABE.\Setup(1^\secp)$.
    \item Generate $\skfe.\msk\gets\SKFESKL.\Setup(1^\secp)$.
    \item Output $\pk\seteq\abe.\pk$ and $\msk\seteq(\abe.\msk,\skfe.\msk)$.
\end{itemize}

\item[$\ABESKL.\qKG(\msk,y)$:] $ $
\begin{itemize}
    \item Parse $\msk=(\abe.\msk,\skfe.\msk)$.
    \item Generate $(\skfe.\qsk,\skfe.\vk,\skfe.\tk)\gets\SKFESKL.\qKG(\skfe.\msk,y)$. We denote the register holding $\skfe.\qsk$ as $\qreg{SKFE.SK}$.
    \item Sample explicit randomness $\key \gets \bit^\secp$.
    \item Prepare a register $\qreg{ABE.SK}$ that is initialized to $\ket{0\cdots0}_{\qreg{ABE.SK}}$.
    \item Apply the map
        $\ket{u}_{\qreg{SKFE.SK}}\ket{v}_{\qreg{ABE.SK}}\ra\ket{u}_{\qreg{SKFE.SK}}\ket{v\oplus\ABE.\KG(\abe.\msk,y\|u,\key)}_{\qreg{ABE.SK}}$ to the registers $\qreg{SKFE.SK}$ and $\qreg{ABE.SK}$, and obtain $\qsk$ over the registers $\qreg{SKFE.SK}$ and $\qreg{ABE.SK}$.
    \item Output $\qsk$ and
        $\vk\seteq(y,\abe.\msk,\skfe.\vk,\skfe.\tk, \key)$.
\end{itemize}

\item[$\ABESKL.\Enc(\pk,x,\msg)$:] $ $
\begin{itemize}
    \item Parse $\pk=\abe.\pk$.
\item Generate $\tlC\gets\CCSim(1^\secp,\pp_D,1)$, where $\pp_D$ consists of circuit parameters of $D$ defined in the security proof.
    \item Generate $\abe.\ct\gets\ABE.\Enc(\abe.\pk,x\|\tlC,\msg)$.
    \item Output $\ct\seteq\abe.\ct$.
\end{itemize}
 
\item[$\ABESKL.\qDec(\qsk,\ct)$:] $ $
\begin{itemize}
   \item Parse $\ct=\abe.\ct$. We denote the register holding $\qsk$ as $\qreg{SKFE.SK}\tensor\qreg{ABE.SK}$.
   \item Prepare a register $\qreg{MSG}$ that is initialized to $\ket{0\cdots0}_{\qreg{MSG}}$
   \item Apply the map $\ket{v}_{\qreg{ABE.SK}}\ket{w}_{\qreg{MSG}}\ra\ket{v}_{\qreg{ABE.SK}}\ket{w\oplus\ABE.\Dec(v,\abe.\ct)}_{\qreg{MSG}}$ to the registers $\qreg{ABE.SK}$ and $\qreg{MSG}$.
   \item Measure the register $\qreg{MSG}$ in the computational basis and output the result $\msg^\prime$.
\end{itemize}

\item[$\ABESKL.\qVrfy(\vk,\qsk^\prime)$:] $ $
\begin{itemize}
    \item Parse $\vk=(y,\abe.\msk,\skfe.\vk,\skfe.\tk, \key)$. We denote the register holding $\qsk^\prime$ as $\qreg{SKFE.SK}\tensor\qreg{ABE.SK}$.
    \item Prepare a register $\qreg{SKFE.KT}$ that is initialized to $\ket{0}_{\qreg{SKFE.KT}}$.
    \item Apply the map $\ket{u}_{\qreg{SKFE.SK}}\ket{\beta}_{\qreg{SKFE.KT}}\ra\ket{u}_{\qreg{SKFE.SK}}\ket{\beta\oplus\SKFESKL.\KeyTest(\skfe.\tk,u)}_{\qreg{SKFE.KT}}$ to the registers $\qreg{SKFE.SK}$ and $\qreg{SKFE.KT}$.
    \item Measure $\qreg{SKFE.KT}$ in the computational basis and
        output $\bot$ if the result is $0$. Otherwise, go to the next step.
    \item Apply the map
        $\ket{u}_{\qreg{SKFE.SK}}\ket{v}_{\qreg{ABE.SK}}\ra\ket{u}_{\qreg{SKFE.SK}}\ket{v\oplus\ABE.\KG(\abe.\msk,y\|u,\key)}_{\qreg{ABE.SK}}$ to the registers $\qreg{SKFE.SK}$ and $\qreg{ABE.SK}$.
    \item Trace out the register $\qreg{ABE.SK}$ and obtain
        $\skfe.\qsk^\prime$ over register $\qreg{SKFE.SK}$.
    \item Output $\top$ if $\top=\SKFESKL.\qVrfy(\skfe.\vk,\skfe.\qsk^\prime)$ and $\bot$ otherwise.

\end{itemize}
\end{description}

\paragraph{Decryption correctness:}
\fuyuki{We require that simulated obfuscated program generated by $\CCSim$ outputs $\bot$ for every input.}
Let $x$ and $y$ be a ciphertext-attribute and a key-attribute,
respectively, such that $R(x,y)=0$. The secret-key
$\qsk$ output by $\ABESKL.\qKG$ is of the form $\sum_u \alpha_u
\ket{u}_{\qreg{SKFE.SK}}\ket{\abe.\sk_{y\|u}}_{\qreg{ABE.SK}}$,
where $\abe.\sk_{y\|u}\gets\ABE.\KG(\abe.\msk,y\|u,\key)$.
Let $\ct\gets\ABESKL.\Enc(\pk,x,\msg)$, where
$\ct=\abe.\ct\gets\ABE.\Enc(\abe.\pk,x\|\tlC,\msg)$ and
$\tlC\gets\CCSim(1^\secp, \pp_D, 1)$.
If we apply the map
\begin{align}
\ket{v}_{\qreg{ABE.SK}}\ket{w}_{\qreg{MSG}}\ra\ket{v}_{\qreg{ABE.SK}}\ket{w\oplus\ABE.\Dec(v,\abe.\ct)}_{\qreg{MSG}}
\end{align}
to $\sum_u \alpha_u \ket{u}_{\qreg{SKFE.SK}}\ket{\abe.\sk_{y\|u}}_{\qreg{ABE.SK}} \tensor \ket{0\cdots0}_{\qreg{MSG}}$, the result is
\begin{align}
\sum_u \alpha_u \ket{u}_{\qreg{SKFE.SK}}\ket{\abe.\sk_{y\|u}}_{\qreg{ABE.SK}} \tensor \ket{\msg}_{\qreg{MSG}}
\end{align}
since $\tlC(u)=\bot$ and thus $R^\prime(x\|\tlC,y\|u)=0$ for every string $u$.
Therefore, $\ABESKL$ satisfies decryption correctness.

\paragraph{Verification correctness.}
Let $y$ be a key attribute and $(\qsk,\vk)\gets\ABESKL.\qKG(\msk,y)$.
It is clear that the state $\skfe.\qsk^\prime$ obtained when computing $\ABESKL.\qVrfy(\vk,\qsk)$ is the same as $\skfe.\qsk$ generated when generating $\qsk$.
Therefore, the verification correctness of $\ABESKL$ follows from that of $\SKFESKL$.

\subsection{Proof of Selective IND-KLA Security}

Let $\qA$ be an adversary for the selective IND-KLA security of $\ABESKL$.
We consider the following sequence of experiments.
\begin{description}
\item[$\hybi{0}^\coin$:]This is $\expc{\ABESKL,\qA}{sel}{ind}{kla}(1^\secp,\coin)$.
\begin{enumerate}
\item $\qA$ declares the challenge ciphertext attribute $x^*$. The
    challenger $\qCh$ generates $(\abe.\pk,\abe.\msk)\gets\ABE.\Setup(1^\secp)$ and $\skfe.\msk\gets\SKFESKL.\Setup(1^\secp)$, and sends $\pk\seteq\abe.\pk$ to $\qA$.

\item $\qA$ can get access to the following (stateful) oracles, where the list $\List{\qKG}$ used by the oracles is initialized to an empty list:

\begin{description}
\item[$\Oracle{\qKG}(y)$:]Given $y$, it finds an entry of the form
$(y,\vk_y,V)$ from $\List{\qKG}$. If there is such an entry, it
returns $\bot$. Otherwise, it generates $\qsk_y, \vk_y$ as follows.

\begin{itemize}
\item Generate $(\skfe.\qsk_y,\skfe.\vk_y,\skfe.\tk_y)\gets\SKFESKL.\qKG(\skfe.\msk,y)$. We denote the register holding $\skfe.\qsk_y$ as $\qreg{SKFE.SK_y}$.

\item Prepare a register $\qreg{ABE.SK_y}$ that is initialized to $\ket{0\cdots0}_{\qreg{ABE.SK_y}}$.

\item Sample explicit randomness $\key_y \gets \bit^\secp$.

\item Apply the map
    $\ket{u}_{\qreg{SKFE.SK_y}}\ket{v}_{\qreg{ABE.SK_y}}\ra\ket{u}_{\qreg{SKFE.SK_y}}\ket{v\oplus\ABE.\KG(\abe.\msk,y\|u,\key_y)}_{\qreg{ABE.SK_y}}$ to the registers $\qreg{SKFE.SK_y}$ and $\qreg{ABE.SK_y}$, and obtain $\qsk_y$ over the registers $\qreg{SKFE.SK_y}$ and $\qreg{ABE.SK_y}$.

\item Set $\vk_y\seteq(y,\abe.\msk,\skfe.\vk_y,\skfe.\tk_y, \key_y)$.
\end{itemize}
It returns $\qsk_y$ to $\qA$ and adds the entry $(y,\vk_y,\bot)$ to $\List{\qKG}$.
\item[
$\Oracle{\qVrfy}(y,\widetilde{\qsk})$:] It finds an entry
$(y,\vk_y,V)$ from $\List{\qKG}$. (If there is no such entry, it
returns $\bot$.) It parses
$\vk_y=(y,\abe.\msk,\skfe.\vk_y,\skfe.\tk_y,\key_y)$ and computes
$d$ as follows.  

\begin{itemize}
    \item Let the register holding $\widetilde{\qsk}$ be $\qreg{SKFE.SK_y}\tensor\qreg{ABE.SK_y}$.
    \item Prepare a register $\qreg{SKE.KT_y}$ that is initialized to $\ket{0}_{\qreg{SKE.KT_y}}$.
    \item Apply $\ket{u}_{\qreg{SKFE.SK_y}}\ket{\beta}_{\qreg{SKE.KT_y}}\ra\ket{u}_{\qreg{SKFE.SK_y}}\ket{\beta\oplus\SKFESKL.\KeyTest(\skfe.\tk_y,u)}_{\qreg{SKE.KT_y}}$ to the registers $\qreg{SKFE.SK_y}$ and $\qreg{SKE.KT_y}$.
    \item Measure $\qreg{SKE.KT_y}$ in the computational basis and
        set $d\seteq\bot$ if the result is $0$. Otherwise, go to the next step.
    \item Apply the map
        $\ket{u}_{\qreg{SKFE.SK_y}}\ket{v}_{\qreg{ABE.SK_y}}\ra\ket{u}_{\qreg{SKFE.SK_y}}\ket{v\oplus\ABE.\KG(\abe.\msk,y\|u,\key_y)}_{\qreg{ABE.SK_y}}$ to the registers $\qreg{SKFE.SK_y}$ and $\qreg{ABE.SK_y}$.
    \item Trace out the register $\qreg{ABE.SK_y}$ and obtain $\skfe.\qsk^\prime$ over $\qreg{SKFE.SK_y}$.
    \item Set $d\seteq\top$ if $\top=\SKFESKL.\qVrfy(\skfe.\vk_y,\skfe.\qsk^\prime)$ and set $d\seteq\bot$ otherwise.
It returns $d$ to $\qA$. Finally, if $V=\bot$, it updates the entry
$(y,\vk_y,V)$ to $(y,\vk_y,d)$. 
\end{itemize}
\end{description}
\item $\qA$ sends $(\msg_0^*,\msg_1^*)\in \cM^2$ to
$\qCh$. $\qCh$ checks if for every entry $(y, \vk_y, V)$ in
$L_{\qKG}$ such that $R(x^*,y) = 0$, it holds that $V = \top$. If
so, it generates $\tlC^*\gets\CCSim(1^\secp, \pp_D, 1)$ and
$\abe.\ct^*\gets\ABE.\Enc(\abe.\pk,x^*\|\tlC^*,\allowbreak\msg_\coin^*)$, and
sends $\ct^*\seteq\abe.\ct^*$ to $\qA$. Otherwise, it outputs $0$.

\item $\qA$ outputs $\coin^\prime$. The challenger outputs $\coin'$
as the final output of the experiment.
\end{enumerate}
        
\item[$\hybi{1}^\coin$:]This is the same as $\hybi{0}^\coin$ except
that $\tlC^*$ is generated as
$\tlC^*\gets\CCObf(1^\secp,D[\skfe.\ct^*], \lock, 0)$, where $\skfe.\ct^*\gets\SKFESKL.\Enc(\skfe.\msk,\allowbreak x^*\|0^\secp)$, $\lock\gets\bit^\secp$, and $D[\skfe.\ct^*](x)$ is a circuit that has $\skfe.\ct^*$ hardwired and outputs $\SKFESKL.\CDec(x,\skfe.\ct^*)$.
\end{description}

We pick $\lock$ as a uniformly random string that is completely independent of other variables such as $\skfe.\ct^*$.
Thus, from the security of $\CCObf$, we have $\abs{\Pr[\hybi{0}^\coin=1]-\Pr[\hybi{1}^\coin=1]}=\negl(\secp)$.

\begin{description}
\item[$\hybi{2}^\coin$:]This is the same as $\hybi{1}^\coin$ except
that $\skfe.\ct^*$ hardwired into the circuit $D$ is generated
as $\skfe.\ct^*\gets\SKFESKL.\Enc(\skfe.\msk,x^*\|\lock)$.
\end{description}

\fuyuki{$\qA$ returns every key for a key attribute $y$ such that $R(x^*,y)=\decryptable$. On the other hand, for any key attirbute $y$ such that $R(x^*,y)=\undecryptable$, $F(x^*\|\lock,y)=F(x^*\|0^\secp,y)=\bot$. Then, we can invoke the OT-IND-KLA security of $\SKFESKL$.}

From the selective single-ciphertext security of $\SKFESKL$, we can show that
$\hybi{1}^\coin \approx \hybi{2}^\coin$. Suppose that $\Hyb_1^\coin
\not \approx \Hyb_2^\coin$ and $\qD$ is a corresponding
distinguisher. We consider the following reduction
$\qR$:

\begin{description}
\item Execution of $\qR^{\qD}$ in
$\expc{\SKFESKL,\qR}{sel}{1ct}{kla}(1^\secp, b)$:

\begin{enumerate}
\item $\qCh$ runs $\skfe.\msk \gets \SKFESKL.\Setup(1^\secp)$ and initializes $\qR$
with the security parameter $1^\secp$.

\item $\qD$ declares the challenge ciphertext-attribute $x^*$. $\qR$
generates $(\abe.\pk, \abe.\msk) \gets \ABE.\Setup(1^\secp)$ and
sends $\ek \seteq \abe.\pk$ to $\qD$.

\item $\qR$ chooses $\lock \chosen \zo{\secp}$ and sends $(\skfe.\msg_0^*,
\skfe.\msg_1^*) \seteq (x^* \| 0^\secp, x^* \| \lock)$ to $\qCh$.

\item $\qR$ simulates the access to $\Oracle{\qKG}(y)$ for $\qD$ as
    follows, where $L_{\qKG}$ is a list initialized to be empty:
\begin{description}
\item $\Oracle{\qKG}(y):$ Given $y$, it finds an entry of the
form $(y, \vk_y, V)$ from $L_{\qKG}$. If there is such an entry,
it returns $\bot$. Otherwise, it generates $(\qsk_y, \vk_y)$ 
similar to $\Hyb_0^\coin$ except that the values
$(\skfe.\qsk_y, \skfe.\vk_y, \skfe.\tk_y)$ are generated as
$(\skfe.\qsk_y,
\skfe.\vk_y, \skfe.\tk_y) \gets \SKFESKL.\Oracle{\qKG}(y)$ instead.
It adds $(y, \vk_y, \bot)$ to $L_{\qKG}$ and returns $\qsk_y$.
\end{description}

\item $\qR$ simulates the access to $\Oracle{\qVrfy}(y,
\widetilde{\qsk})$ as follows:
\begin{description}
\item $\Oracle{\qVrfy}(y, \widetilde{\qsk}):$ Given $(y,
\widetilde{\qsk})$, it finds an entry $(y, \vk_y, V)$ from
$L_{\qKG}$ (If there is no such entry, it returns $\perp$). It then
executes a procedure similar to that in $\Hyb_0^\coin$, except that $\SKFESKL.\Oracle{\qVrfy}(y,
\skfe.\qsk')$ is executed instead of $\SKFESKL.\qVrfy(\skfe.\vk_y,
\skfe.\qsk')$. The corresponding output $d$ is returned as the
output of the oracle.
\end{description}
\item $\qR$ requests the challenge ciphertext from $\qCh$ and receives
$\skfe.\ct^* \gets \SKFESKL.\Enc(\skfe.\msk, \skfe.\msg^*_b)$.
\item $\qD$ sends $(\msg_0^*, \msg_1^*) \in \cM^2$ to $\qR$. $\qR$
checks that for every entry $(y,\vk_y,V)$ such that $R(x^*,y)=0$, it
holds that $V = \top$. If so, it
generates $\tlC^* \gets \CCObf(1^\secp, D[\skfe.\ct^*], \lock, 0)$
and $\abe.\ct^* \gets \ABE.\Enc(\abe.\pk, x^* \| \tlC^*, \msg^*_\coin)$ and
sends $\ct^* \seteq \abe.\ct^*$ to $\qD$. Else, it outputs $0$.
\item $\qD$ outputs a bit $b'$. $\qR$ outputs $b'$ and $\qCh$
outputs $b'$ as the final output of the experiment.
\end{enumerate}
\end{description}

It is easy to see that the view of $\qD$ is the same as that in
$\Hyb^\coin_2$ when $\lock$ is encrypted in $\skfe.\ct^*$ and that
of $\Hyb^\coin_1$ when $0^\secp$ is encrypted. Moreover, for $\qD$
to distinguish between the two hybrids, it must be the case that
$V = \top$ for all entries $(y, \vk_y, V)$ such that $R(x^*, y) =
0$, which directly implies that the analogous values checked by
$\SKFESKL.\Oracle{\qVrfy}$ must also be $\top$. If $R(x^*,y)=1$, $F(x^*\concat 0^\secp,y)=F(x^*\concat \lock,y)=\bot$.  Consequently,
$\qR$ breaks the selective single-ciphertext security of $\SKFESKL$. Therefore, it
must be that $\Hyb_1^\coin \approx \Hyb_2^\coin$.

\begin{description}
\item[$\hybi{3}^\coin$:]This is the same as $\hybi{2}^\coin$ except that $\Oracle{\qKG}$ behaves as follows. (The difference is red colored.)

\begin{description}
\item[$\Oracle{\qKG}(y)$:]Given $y$, it finds an entry of the form
$(y,\vk_y,V)$ from $\List{\qKG}$. If there is such an entry, it
returns $\bot$. Otherwise, it generates $\qsk_y, \vk_y$ as follows.

\begin{itemize}
\item Generate $(\skfe.\qsk_y,\skfe.\vk_y,\skfe.\tk_y)\gets\SKFESKL.\qKG(\skfe.\msk,y)$. We denote the register holding $\skfe.\qsk_y$ as $\qreg{SKFE.SK_y}$.
    \textcolor{red}{
        \item Prepare a register $\qreg{ABE.R^\prime_y}$ that is initialized to $\ket{0}_{\qreg{ABE.R^\prime_y}}$.
    \item Apply the map $\ket{u}_{\qreg{SKFE.SK_y}}\ket{\beta}_{\qreg{ABE.R^\prime_y}}\ra\ket{u}_{\qreg{SKFE.SK_y}}\ket{\beta\oplus R^\prime(x^*\|\tlC^*,y\|u)}_{\qreg{ABE.R^\prime_y}}$ to the registers $\qreg{SKFE.SK_y}$ and $\qreg{ABE.R^\prime_y}$.
    \item Measure $\qreg{ABE.R^\prime_y}$ in the computational basis
        and output $\bot$ if the result is $0$. Otherwise, go to the next step.
    }

\item Prepare a register $\qreg{ABE.SK_y}$ that is initialized to $\ket{0\cdots0}_{\qreg{ABE.SK_y}}$.

\item Sample explicit randomness $\key_y \gets \bit^\secp$.

\item Apply the map
    $\ket{u}_{\qreg{SKFE.SK_y}}\ket{v}_{\qreg{ABE.SK_y}}\ra\ket{u}_{\qreg{SKFE.SK_y}}\ket{v\oplus\ABE.\KG(\abe.\msk,y\|u,\key_y)}_{\qreg{ABE.SK_y}}$ to the registers $\qreg{SKFE.SK_y}$ and $\qreg{ABE.SK_y}$, and obtain $\qsk_y$ over the registers $\qreg{SKFE.SK_y}$ and $\qreg{ABE.SK_y}$.

\item Set $\vk_y\seteq(y,\abe.\msk,\skfe.\vk_y,\skfe.\tk_y, \key_y)$.
\end{itemize}
It returns $\qsk_y$ to $\qA$ and adds the entry $(y,\vk_y,\bot)$ to $\List{\qKG}$.
\end{description}
\end{description}

With overwhelming probability, we observe that the added procedure that checks
$R^\prime(x^* \| \tlC^*,y \| u)$ in superposition affects the final state $\qsk_y$ at most negligibly.
We consider the following two cases.
\begin{itemize}
\item The first case is where $R(x^*,y)=1$. In this case, we have
    $R^\prime(x^*\|\tlC^*,y\|u)=1$ for any $u$.
Hence, we see that the added procedure does not affect the state $\skfe.\qsk_y$.
\item The second case is where $R(x^*,y)=0$. Let
$\skfe.\qsk_y=\sum_u \alpha_u \ket{u}_{\qreg{SKFE.SK_y}}$. From
the classical decryption property of $\SKFESKL$, we have that for
any $u$, it holds that
$\SKFESKL.\CDec(u,\skfe.\ct^*)=\lock$. Thus, $\tlC^*(u)\ne\bot$ and
$R^\prime(x^*\|\tlC^*,y\|u)=1$ with overwhelming probability from
the correctness of $\CCObf$. Hence, we see that the added procedure affects the state $\skfe.\sk_y$ at most negligibly.
\end{itemize}
Therefore, we have $\abs{\Pr[\hybi{2}^\coin=1]-\Pr[\hybi{3}^\coin=1]}=\negl(\secp)$.

\begin{description}
\item[$\hybi{4}^\coin$:]This is the same as $\hybi{3}^\coin$ except that the oracle $\Oracle{\qVrfy}$ behaves as follows. (The difference is red colored.)
\begin{description}
\item[ $\Oracle{\qVrfy}(y,\widetilde{\qsk})$:] It finds an entry
    $(y,\vk_y,V)$ from $\List{\qKG}$. (If there is no such entry, it
    returns $\bot$.) It parses
    $\vk_y=(y,\abe.\msk,\skfe.\vk_y,\skfe.\tk_y,\key_y)$. If
    $R(x^*, y) = 1$, then it behaves exactly as in $\Hyb^\coin_3$
    and otherwise, it computes $d$ as follows:
    \begin{enumerate}[(a)]
    \item Let the register holding $\widetilde{\qsk}$ be $\qreg{SKFE.SK_y}\tensor\qreg{ABE.SK_y}$.
    \item Prepare a register $\qreg{SKFE.KT_y}$ that is initialized to $\ket{0}_{\qreg{SKFE.KT_y}}$.
    \item Apply $\ket{u}_{\qreg{SKFE.SK_y}}\ket{\beta}_{\qreg{SKFE.KT_y}}\ra\ket{u}_{\qreg{SKFE.SK_y}}\ket{\beta\oplus\SKFESKL.\KeyTest(\skfe.\tk_y,u)}_{\qreg{SKFE.KT_y}}$ to the registers $\qreg{SKFE.SK_y}$ and $\qreg{SKFE.KT_y}$.
    \item Measure $\qreg{SKFE.KT_y}$ in the computational basis and
        set $d\seteq\bot$ if the result is $0$. Otherwise, go to the next step.
    \textcolor{red}{
    \item Prepare a register $\qreg{ABE.R^\prime_y}$ that is initialized to $\ket{0}_{\qreg{ABE.R^\prime_y}}$.
    \item Apply the map $\ket{u}_{\qreg{SKFE.SK_y}}\ket{\beta}_{\qreg{ABE.R^\prime_y}}\ra\ket{u}_{\qreg{SKFE.SK_y}}\ket{\beta\oplus R^\prime(x^*\|\tlC^*,y\|u)}_{\qreg{ABE.R^\prime_y}}$ to the registers $\qreg{SKFE.SK_y}$ and $\qreg{ABE.R^\prime_y}$.
    \item Measure $\qreg{ABE.R^\prime_y}$ in the computational basis and set $d\seteq\bot$ if the result is $0$. Otherwise, go to the next step.
    }
    \item Apply the map
        $\ket{u}_{\qreg{SKFE.SK_y}}\ket{v}_{\qreg{ABE.SK_y}}\ra\ket{u}_{\qreg{SKFE.SK_y}}\ket{v\oplus\ABE.\KG(\abe.\msk,y\|u,\key)}_{\qreg{ABE.SK_y}}$ to the registers $\qreg{SKFE.SK_y}$ and $\qreg{ABE.SK_y}$.
    \item Trace out the register $\qreg{ABE.SK_y}$ and obtain $\skfe.\qsk^\prime$ over $\qreg{SKFE.SK_y}$.
    \item Set $d\seteq\top$ if $\top=\SKFESKL.\qVrfy(\skfe.\vk_y,\skfe.\qsk^\prime)$ and set $d\seteq\bot$ otherwise.
It returns $d$ to $\qA$. Finally, if $V=\bot$, it updates the entry $(y,\vk_y,V)$ into $(y,\vk_y,d)$. 
\end{enumerate}
            \end{description}
\end{description}

%

Suppose there exists a QPT distinguisher $\qD$ that has
non-negligible advantage in distinguishing $\Hyb_3^\coin$ and
$\Hyb_4^\coin$. Let $\qD$ make $q = \poly(\lambda)$ many queries to
the oracle $\Oracle{\qVrfy}(\cdot, \cdot)$. We will now consider the
following QPT algorithm $\qA_\oh$ with access to an oracle
$\widetilde{\Oracle{\qKG}}$ and an oracle $\cO$ that runs $\qD$ as follows:

\begin{description}
\item $\underline{\qA_\oh^{\widetilde{\Oracle{\qKG}}, \cO}(\abe.\pk,
\skfe.\msk)}:$
\begin{enumerate}
\item $\qA_\oh$ runs $\qD$ who sends the challenge ciphertext-attribute
$x^*$ to $\qA_\oh$.
\item $\qA_\oh$ sends $\ek = \abe.\pk$ to $\qD$ and initializes $L_{\qKG}$
to be an empty list.
\item When $\qD$ queries $\Oracle{\qKG}$ on input $y$, $\qA_\oh$ queries
the oracle $\widetilde{\Oracle{\qKG}}(y)$ in order to obtain its
response $\ket{\phi}$ and a list $\widetilde{L_{\qKG}}$. It updates
$L_{\qKG} = \widetilde{L_{\qKG}}$. It sends $\ket{\phi}$ to $\qD$.
\item $\qA_\oh$ simulates the access of $\Oracle{\qVrfy}$ for $\qD$ as
follows:
\begin{description}
\item $\underline{\Oracle{\qVrfy}(y, \widetilde{\qsk})}:$
\begin{enumerate}
\item Execute Steps (a)-(e) of $\Oracle{\qVrfy}$ as in
    $\Hyb_4^\coin$.
    \item Apply the map
        $\ket{u}_{\qreg{SKE.SK_y}}\ket{\beta}_{\qreg{ABE.R_y}}\ra\ket{u}_{\qreg{SKE.SK_y}}\ket{\beta\oplus
        \calO(u)}_{\qreg{ABE.R_y}}$ to the registers
        $\qreg{SKE.SK_y}$ and $\qreg{ABE.R_y}$. 
\item Execute Steps (g)-(j) of $\Oracle{\qVrfy}$ as in
$\Hyb_4^\coin$.
\end{enumerate}
\end{description}

\item $\qD$ sends $(\msg_0^*, \msg_1^*) \in \cM^2$ to $\qA_\oh$. $\qA_\oh$ checks if
for every entry $(y, \vk_y, V)$ in $L_{\qKG}$ such that $R(x^*, y) =0$, it holds that $V = \top$. If so, it generates
$\tlC^*\gets\CCObf(1^\secp,D[\skfe.\ct^*], \lock, 0)$
 and $\abe.\ct^* \gets \ABE.\Enc(\abe.\pk, x^* \|
\tlC^*, \msg^*_\coin)$, where $\skfe.\ct^* = \SKFESKL.\Enc(\skfe.\msk,
x^* \| 0^\secp) $and sends $\ct^* \seteq \abe.\ct^*$ to $\qD$.
Otherwise, it outputs $0$.
\item $\qD$ outputs a guess $b'$. $\qA_\oh$ outputs $b'$.
\end{enumerate}
\end{description}

Let $H$ be an oracle that for every input $u$, outputs $1$.
Consider now the extractor $\qB_\oh^{\widetilde{\Oracle{\qKG}}, H}$
as specified by the O2H Lemma (Lemma \ref{lem:O2H}). We will now
construct a reduction $\qR$ that runs $\qB_\oh$ by simulating the oracles
$\widetilde{\Oracle{\qKG}}$ and $H$ for $\qB_\oh$, and breaks
key-testability of the $\SKFESKL$ scheme.

\begin{description}
\item Execution of $\qR$ in
$\expb{\SKFESKL,\qR}{key}{test}(1^\secp)$:

\begin{enumerate}
\item The challenger $\qCh$ runs $\skfe.\msk \leftarrow
\SKFESKL.\Setup(1^\secp)$ and
initializes $\qR$ with input $\skfe.\msk$.
\item $\qR$ samples $(\abe.\pk, \abe.\msk) \gets
\ABE.\Setup(1^\secp)$ and initializes $\qB_\oh$ with the input
$(\abe.\pk, \skfe.\msk)$.
\item $\qR$ simulates the access to $\widetilde{\Oracle{\qKG}}(y)$ for $\qB$ as
follows, where $L_{\qKG}$ is a list initialized to be empty:
\begin{description}
\item $\widetilde{\Oracle{\qKG}}(y):$ Given $y$, it finds an entry of the
form $(y, \vk_y, V)$ from $L_{\qKG}$ (If there is such an entry,
it returns $(\bot, L_{\qKG})$). Otherwise, it generates $(\qsk_y, \vk_y)$ 
similar to $\Hyb_3^\coin$.
It adds $(y, \vk_y, \bot)$ to $L_{\qKG}$ and returns $(\qsk_y,
L_{\qKG})$.
\end{description}

\item When $\qB_\oh$ queries an input $u$ to $H$, $\qR$ responds with $1$.

\item $\qB_\oh$ outputs values $y$ and $\sk$. $\qR$ sends $(\sk, y,
\lock)$ to $\qCh$.

\end{enumerate}
\end{description}

We will now claim that with non-negligible probability, $\qR$ obtains
values $\sk$ and $y$ such that $R'(x^* \| \tlC^*,y \| \sk) = 0$.  Recall that
the hybrids differ only when $R(x^*, y) = 0$.  By the definition of
$R'$ and $\tlC^*$ and the decryption correctness of $\SKFESKL$, this
will imply that $\SKFESKL.\CDec(\sk, \skfe.\ct^\star) \neq
F(x^*\|\lock,y)$ since $F(x^* \| \lock, y) = \lock$, and $R'(x^* \| \tlC^*,y \| \sk) = 0$ requires $\tlC^*(\sk)=\bot$. Moreover,
$\KeyTest(\skfe.\tk_y, \sk)$ also holds. Consequently, this will imply
$\qR$ breaks the key-testability of $\SKFESKL$. To prove this, we will
rely on the One-Way to Hiding (O2H) Lemma (Lemma \ref{lem:O2H}).
Consider an oracle $G$ which takes as input $u$ and outputs $R'(x^* \|
\tlC^\star, y\|u)$ and an oracle $H$ which takes as input $u$ and
outputs $1$. Notice that if the oracle $\calO = G$, then the view of
$\qD$ as run by $\qA_\oh$ is the same as in $\Hyb_4^\coin$, while if $\calO = H$, the view
of $\qD$ is the same as in $\Hyb_3^\coin$. By the O2H Lemma, we
have the following, where $z = (\abe.\pk, \skfe.\msk)$.

\begin{align}
\abs{\Pr[\qA_\oh^{\widetilde{\Oracle{\qKG}},
H}(z)=1]-\Pr[\qA_\oh^{\widetilde{\Oracle{\qKG}}, G}(z)=1]} \leq
2q\cdot\sqrt{\Pr[\qB_\oh^{\widetilde{\Oracle{\qKG}}, H}(z)\in S]}
\enspace.
\end{align}
where $S$ is a set where
the oracles $H$ and $G$ differ, which happens only for inputs $u$ s.t.
$R'(x^*\|\tlC^\star, y\|u) = 0$. Since $\qR$ obtains $\sk$ and $y$ as the
output of $\qB_\oh^{\widetilde{\Oracle{\qKG}},H}(z)$, the argument
follows that $\Hyb_3^\coin \approx \Hyb_4^\coin$.

\begin{description}
\item[$\hybi{5}^\coin$:]This is the same as $\hybi{4}^\coin$ except that $\ct^*\seteq\abe.\ct^*$ is generated as $\abe.\ct^*\gets\ABE.\Enc(\abe.\pk,x^*\|\tlC^*,0^{\msglen})$.
\end{description}

The view of $\qA$ in $\hybi{4}^\coin$ and $\hybi{5}^\coin$ can be simulated with $\abe.\pk$ and the access to the quantum key generation oracle $\Oracle{qkg}$.
This is because before $\ABE.\KG$ is required to be applied in the
simulation of oracles $\Oracle{\qKG}$ and $\Oracle{\qVrfy}$, the
relation check $R'(x^* \| \tlC^\star, y\|u)$ is already applied in
superposition.
Thus, we have $\abs{\Pr[\hybi{4}^\coin=1]-\Pr[\hybi{5}^\coin=1]}=\negl(\secp)$.

Lastly, $\hybi{5}^{0}$ and $\hybi{5}^{1}$ are exactly the same experiment and thus we have $\abs{\Pr[\hybi{5}^0=1]-\Pr[\hybi{5}^1=1]}=\negl(\secp)$.
Then, from the above arguments, we obtain
\begin{align}
&\abs{\Pr[\expc{\ABESKL,\qA}{sel}{ind}{kla}(1^\secp,0)=1]-\expc{\ABESKL,\qA}{sel}{ind}{kla}(1^\secp,1)=1]}\\
=&\abs{\Pr[\hybi{0}^0=1]-\Pr[\hybi{0}^1=1]}\le\negl(\secp).
\end{align}
This completes the proof. \qed

Recall now that SKFE-CR-SKL with Key-Testability is implied by
BB84-based SKE-CD, OWFs and adaptively single-ciphertext
function-private SKFE. Since this notion of SKFE is implied by OWFs
(Theorem \ref{thm:1ct_adaptive_function_private_SKFE}), we have that
SKFE-CR-SKL with Key-Testability is implied by LWE. Since,
Compute-and-Compare Obfuscation and quantum-secure Ciphertext-Policy
ABE for Circuits are both implied by LWE, this gives us the 
following theorem:

\begin{theorem}
There exists an ABE-CR-SKL scheme satisfying Selective IND-KLA
security, assuming the polynomial hardness of the LWE assumption.
\end{theorem}


\section{ABE-CR\textsuperscript{2}-SKL from Multi-Input ABE}\label{sec:ABR-CR2-SKL}
In this section, we show how to achieve ABE-CR-SKL with classical certificates (ABE-CR\textsuperscript{2}-SKL) from MI-ABE.
\subsection{Definitions of Multi-Input ABE}\label{sec:def-mi-abe}
First, we recall the definitions of MI-ABE.
The syntax of MI-ABE is as follows:

\begin{definition}\label{def:miabe}
A Multi-Input (Ciphertext-Policy) Attribute-Based Encryption scheme
$\MIABE$ is a tuple of four PPT algorithms $(\Setup, \KeyGen,
\allowbreak \Enc,
\Dec)$. Let $k=\poly(\secp)$ and let $\{\cX_\secp\}_\secp$,
$\{(\cY_\secp)^k\}_\secp$ and $R = \{R_\secp: \cX_\secp \times
(\cY_\secp)^k \ra \bit\}$ be the ciphertext attribute space, key
attribute space and the relation associated with $\MIABE$
respectively. Let $\bit^\ell$ be the message space. Let $s(\secp)$
denote the maximum bit string length required to describe PPT
circuits with input size $k\cdot n(\secp)$ and depth $d(\secp)$ for
polynomials $s, n$ and $d$. Let $\cX_\secp = \bit^{s(\secp)}$ and
$\cY = \bit^{n(\secp)}$. Let $R$ be such that $R(x, y_1, \ldots,
y_k) = 0 \iff x(y_1, \ldots, y_k) = \bot$ where $x \in \cX$ is
parsed as a $k$-input circuit and $(y_1, \ldots, y_k)\in \cY^k$. 

\begin{description}
\item[$\Setup(1^\secp)\ra(\pk,\msk)$:] The setup algorithm takes a security parameter $1^\secp$ and outputs a public key $\pk$ and master secret key $\msk$.

\item[$\KG(\msk, i, y_i)\ra \sk_i$:] The key generation algorithm
takes the master secret key $\msk$, an index $i \in [k]$, and a
key-attribute $y_i \in \cY$ as input. It outputs a secret-key
$\sk_i$.

\item[$\Enc(\pk, x, \msg)\ra \ct$:] The encryption algorithm takes
the public key $\pk$, a ciphertext attribute $x \in \cX$, and a
message $\msg \in \bit^\ell$ as input. It outputs a ciphertext
$\ct$.

\item[$\Dec(\ct, \sk_1, \ldots \sk_k)\ra \widetilde{\msg}$:]
The decryption algorithm takes as input a ciphertext $\ct$ and $k$
secret-keys $\sk_1, \ldots, \sk_k$. It outputs a value
$\widetilde{\msg} \in \bit^\ell \cup \bot$.

\item[Correctness:] We require that
\[
\Pr\left[
\Dec(\ct, \sk_1, \ldots, \sk_k) = \msg
 \ \Bigg\lvert
\begin{array}{rl}
 &(\pk,\msk) \la \Setup(1^\secp),\\
 & \forall i \in [k]: \sk_i \gets \KG(\msk,i,y_i), \\
 &\ct \gets \Enc(\pk,x,\msg)
\end{array}
\right] \ge 1 -\negl(\secp).
\]
holds for all $x\in \cX$ and $(y_1, \ldots, y_k) \in \cY^k$ such
that $R(x,y_1,\ldots,y_k)=0$ and $\msg\in \bin^\ell$.
\end{description}

\begin{remark}
We say that an MI-ABE scheme has \emph{polynomial-arity}, if it
allows $k$ to be an arbitrary polynomial in $\secp$.
\end{remark}

%
\end{definition}

\begin{definition}[Post-Quantum Selective Security for MI-ABE:]\label{def:pq_sel_MIABE}
We say that $\MIABE$ is a selective-secure MI-ABE scheme for relation
$R:\cX \times \cY^k \ra \bit$, if it satisfies the following
requirement, formalized from the experiment
$\expb{\MIABE, \qA}{sel}{ind}(1^\secp,\coin)$ between an adversary
$\qA$ and a challenger $\Ch$:

\begin{enumerate}
\item $\qA$ declares the challenge ciphertext attribute
$x$. $\Ch$ runs $(\pk,\msk)\gets\Setup(1^\secp)$ and sends $\pk$ to $\qA$.

\item $\qA$ can get access to the following key generation oracle.

\begin{description}
\item[$\Oracle{kg}(i, y_i)$:] It outputs $\sk_i \gets \KG(\msk, i,
y_i)$ to $\qA$.

\end{description}

\item At some point, $\qA$ sends $(\msg_0,\msg_1)$ to $\Ch$. Then, $\Ch$
generates $\ct^*\gets\Enc(\pk,x,\msg_\coin)$ and sends $\ct^*$ to
$\qA$.

\item Again, $\qA$ can get access to the oracle $\Oracle{kg}$.
\item $\qA$ outputs a guess $\coin^\prime$ for $\coin$ and the
experiment outputs $\coin'$.
\end{enumerate}

For an adversary to be admissible, we require that for all tuples
$(y_1, \ldots, y_k) \in \cY^k$ received by $\qA$, it must hold that
$R(x, y_1, \ldots, y_k) = 1$ (non-decrytable). Given this
constraint, we say $\MIABE$ satisfies selective security if for all
QPT $\qA$, the following holds:

\begin{align}
\advb{\MIABE,\qA}{sel}{ind}(1^\secp) \seteq
\abs{\Pr[\expb{\MIABE,\qA}{sel}{ind} (1^\secp,0) \ra 1] -
\Pr[\expb{\MIABE,\qA}{sel}{ind} (1^\secp,1) \ra 1] }\leq \negl(\secp).
\end{align}

%

\end{definition}
\begin{remark}
    We do not require quantum security in~\cref{def:pq_sel_MIABE} unlike~\cref{def:qsel_ind_ABE}.
\end{remark}

\paragraph{Comparison with the Definition of Agrawal et.
al \cite{C:ARYY23}.}

Their definition considers a key-policy variant of MI-ABE. It
consists of algorithms $(\Setup, \Enc, \KG_1, \cdots,\allowbreak
\KG_{k-1}, \KG_k, \Dec)$ where $\Enc$ works in a similar way
as our definition, i.e., $\Enc(\pk,\allowbreak x_0, \msg) \ra \ct$ for message
$\msg$ and attribute $x_0$. For each $i \in [k-1]$, $\KG_i$ works
as $\KG_i(\msk, x_i) \ra \sk_i$ for the $i$-th key-attribute
$x_i$. On the other hand, $\KG_k$ works as $\KG_k(\msk, f)
\ra \sk_f$ where $f$ is an arbitrary $k$-input function. The
guarantee is that decryption is feasible if and only if
$f(x_0, \ldots, x_{k-1}) = \bot$. The algorithm $\Dec$ has the
syntax $\Dec(\pk, \ct, \sk_1, \ldots, \sk_{k-1}, \sk_f) \ra \msg'$.

It is easy to see that this definition directly implies a
$(k-1)$-input ciphertext-policy MI-ABE with algorithms $(\Setup',
\Enc', \KG', \Dec')$. Specifically, $\Setup'(1^\secp) \seteq
\Setup(1^\secp)$, $\Enc'(\pk, x_0, \msg) \seteq \big(\pk, \Enc(\pk,
x_0, \msg)\big)$ and $\KG'$ can be defined as $\KG'(\msk, i, x_i)
\seteq \KG_i(\msk, x_i)$ for all $i \in [k-1]$. Finally, $\Dec'$ is
defined as $\Dec'(\ct = (\pk, \widetilde{\ct}), \sk_1, \ldots,
\sk_{k-1}) \seteq \Dec(\pk, \widetilde{\ct}, \sk_1, \ldots,
\sk_{k-1}, \sk_f)$ for the function $f(x_0, \ldots, x_{k-1}) =
C_{x_0}(x_1, \ldots, x_{k-1})$ where $C_{x_0}$ is the circuit
described by ciphertext-attribute $x_0$.

Unfortunately, \cite{C:ARYY23} only allows for a constant $k$, and
hence we do not currently know of a construction where
$k=\poly(\secp)$. In the recent work of \cite{myEPRINT:AgrKumYam24a}, a
variant of MI-ABE (for $k=\poly(\secp)$) was achieved from the LWE and evasive LWE assumptions, which is non-standard.
However, the scheme it implies would require the encryption algorithm $\Enc'$ to utilize the master secret-key
$\msk$.

\paragraph{Disucssion on MI-ABE and IO.}
We focus on MI-ABE satisfying~\cref{def:miabe} in this paragraph.
MI-ABE is not stronger than IO since IO is equivalent to multi-input functional encryption (MIFE)~\cite{EC:GGGJKL14} and MIFE trivially implies MI-ABE.
Although we do not know whether MIFE is separated from MI-ABE, it is likely since PKFE is separated from ABE~\cite{C:GarMahMoh17} and constructing FE schemes is significantly more challenging than ABE schemes.
Moreover, MI-ABE is not currently known from weaker assumptions than IO.
They are qualitatively different primitives, and MI-ABE could be constructed from weaker assumptions in the future.
We also know that MI-ABE implies witness encryption~\cite{SCN:BJKPW18},\ryo{but this does not imply anything about relationship between MI-ABE and PKFE.} and witness encryption is achieved from the evasive LWE assumption~\cite{C:Tsabary22,AC:VaiWeeWic22}.
Hence, MI-ABE is somewhere between witness encryption and IO.
Achieving MI-ABE from lattice assumptions is an interesting open problem.

\subsection{Definitions of ABE-CR\textsuperscript{2}-SKL}
The syntax of ABE-CR\textsuperscript{2}-SKL is defined as follows. 

\begin{definition}[ABE-CR\textsuperscript{2}-SKL]\label{def:pke-cr2}
An ABE-CR\textsuperscript{2}-SKL scheme $\ABECRCRSKL$ is a tuple of
six algorithms $(\Setup,\qKG, \Enc, \allowbreak\qDec,\qDel, \Vrfy)$. 
Below, let $\cM$  be the message space of $\ABECRCRSKL$, $\cX$ be the
ciphertext-attribute space, and $\cY$ the key-attribute space.
\begin{description}
\item[$\Setup(1^\secp)\ra(\ek,\msk)$:] The setup algorithm takes a
    security parameter $1^\lambda$, and outputs an encryption key
    $\ek$ and a master secret-key $\msk$.

\item[$\qKG(\msk, y)\ra(\qsk,\vk)$:] The key generation algorithm takes
the master secret-key $\msk$ and a key-attribute $y\in \cY$ as inputs, and outputs a decryption key
$\qsk$ and a verification key $\vk$.

\item[$\Enc(\ek,\widetilde{x},\msg)\ra\ct$:] The encryption algorithm takes an
encryption key $\ek$, a ciphertext-attribute $\widetilde{x} \in \cX$, and a message $\msg
\in \cM$ as inputs, and outputs a ciphertext $\ct$.

\item[$\qDec(\qsk,\ct)\ra\widetilde{\msg}$:] The decryption
    algorithm takes a decryption key $\qsk$ and a ciphertext $\ct$,
    and outputs a value $\widetilde{\msg}$ or $\bot$.


\item[$\qDel(\qsk)\ra\cert$:] The deletion
algorithm takes a decryption key $\qsk$ and outputs a deletion
certificate $\cert$.

\item[$\Vrfy(\vk, \cert')\ra\top/\bot$:] The verification algorithm
takes a verification key $\vk$ and a certificate $\cert'$,
and outputs $\top$ or $\bot$.

\item[Decryption correctness:]For every $\msg \in \cM$, $\widetilde{x} \in \cX$
    and $y \in \cY$ such that $\widetilde{x}(y) = 0$ (decryptable), we have
\begin{align}
\Pr\left[
\qDec(\qsk, \ct) \allowbreak = \msg
\ :
\begin{array}{ll}
(\ek,\msk)\gets\Setup(1^\secp)\\
(\qsk,\vk)\gets\qKG(\msk, y)\\
\ct\gets\Enc(\ek,\widetilde{x},\msg)
\end{array}
\right] 
\ge 1-\negl(\secp).
\end{align}

\item[Verification correctness:] For every $y \in \cY$, we have
\begin{align}
\Pr\left[
\Vrfy(\vk,\cert)=\top
\ :
\begin{array}{ll}
(\ek,\msk)\gets\Setup(1^\secp)\\
(\qsk,\vk)\gets\qKG(\msk, y)\\
\cert\gets\qDel(\qsk)
\end{array}
\right] 
\ge 1-\negl(\secp).
\end{align}
\end{description}

\begin{remark}
We use the same experiment identifier
$\mathsf{sel}\textrm{-}\mathsf{ind}\textrm{-}\mathsf{kla}$ to refer to
the selective IND-KLA
security experiments of ABE-CR-SKL and
ABE-CR\textsuperscript{2}-SKL. This is for the sake of simplicity,
as the exact experiment will be clear from the context.
\end{remark}

\begin{description}
\item[Selective IND-KLA Security:] This security notion for an
ABE-CR\textsuperscript{2}-SKL scheme is defined in the same way as
for ABE-CR-SKL, except that instead of access to the oracle
$\Oracle{\qVrfy}$ in the experiment, $\qA$ receives access to the
following oracle:

\begin{description}

\item[$\Oracle{\Vrfy}(y, \cert)$:] Given $(y,\cert)$, it finds an
entry $(y,\vk,V)$ from $\List{\qKG}$. (If there is no such entry, it returns $\bot$.) 
It then runs $\decision \seteq \Vrfy(\vk, \cert)$ and returns
$\decision$ to $\qA$. If $V=\bot$, it updates the entry into
$(y,\vk,\decision)$.

\end{description}

\end{description}
\end{definition}

\subsection{Construction of ABE-CR\textsuperscript{2}-SKL from MI-ABE}\label{sec-ABE-CR-SKL-classical-certificate}
\paragraph{Construction overview.}

We first provide the main idea
behind our ABE-CR\textsuperscript{2}-SKL based on MI-ABE, which follows along the lines of our
ABE-CR-SKL construction with some key-differences.
In this case, we directly rely on a BB84-based SKE-CD scheme as a
building block, instead of needing something similar to SKFE-CR-SKL.
Consider now an SKE-CD ciphertext $\skecd.\qct$ of the plaintext $0^\secp$.
Let its corresponding verification-key $\skecd.\vk$ be of the form
$\skecd.\vk = (x,
\theta)$, where $x$ and $\theta$ are $k$-bit strings. Then, $\skecd.\qct$ is
of the form $\ket{\psi_1} \otimes \cdots \otimes \ket{\psi_k}$, where
for $i \in [k]$, $\ket{\psi_i}$ is of the following form:

\begin{align}
\ket{\psi_i}=
\begin{cases}
    \ket{x[i]} & if~~ \theta[i]=0\\
    \ket{0}+(-1)^{x[i]}\ket{1} & if~~ \theta[i]=1
\end{cases}
\end{align}

Now, for each $i \in \{2, \dots, k\}$ and $u \in \bit$, consider the attribute-key
$\abe.\sk_{i, u}$ generated by the MI-ABE key-generation
algorithm for slot $i$ and attribute $t \| u$. Here, $t$ is a value chosen at
random and is common for all slots corresponding to a given decryption-key
$\qsk$. Additionally, for $u \in \bit$, consider the keys
$\abe.\sk_{1,u}$ corresponding to the attributes $t \| u \| y$ where
$y$ is the actual ABE key-attribute.

Consider now the following state $\rho_i$ for each $i \in [k]$:

\begin{align}
    \rho_i=
    \begin{cases}
        \ket{x[i]}\ket{\abe.\sk_{i,x[i]}} & if~~ \theta[i]=0\\
        \ket{0}\ket{\abe.\sk_{i,0}}+(-1)^{x[i]}\ket{1}\ket{\abe.\sk_{i,1}} & if~~ \theta[i]=1,
    \end{cases}
\end{align}

The quantum decryption-key of our scheme will be the tuple $\qsk =
(\rho_i)_{i \in [k]}$. The encryption algorithm is similar to the
ABE-CR-SKL scheme, and outputs an MI-ABE encryption of message $\msg$
under a policy $C_{\mathsf{ABE}} \| \tlC$ where $\tlC \gets
\CCSim(1^\secp)$ is a simulator of a compute and compare obfuscator and $C_{\mathsf{ABE}}$ is the actual ABE policy.
However, the MI-ABE relation is a little different. The relation is
such that for a set of $k$ attributes $y_1 = t_1 \| u_1 \| y$ and $y_2
= t_2 \| u_2, \ldots, y_k = t_k \| u_k$, it outputs $0$ (decryptable)
whenever $C_{\mathsf{ABE}}(y) = 0$ (the ABE relation is satisfied) AND $t_1 = \ldots = t_k$ AND
$\tlC(u_1 \| \ldots \| u_k) = \bot$.  Otherwise, it outputs $1$. As in
ABE-CR-SKL, the idea is that $\tlC$ is indistinguishable from $\tlC^*$
that is an obfuscation of the compute-and-compare circuit
$\cnc{D[\lock \xor r, \skecd.\sk]}{\lock, 0}$.  Here, $\lock$ and $r$ are
random values, and $\skecd.\sk$ is the secret-key of the SKE-CD scheme.  The
circuit $D$ is such that $D(u)$ outputs $\lock \xor r \xor
\SKECD.\CDec(\skecd.\sk, u)$, where $\SKECD.\CDec$ is the classical
decryption algorithm of the SKE-CD scheme.

Consequently, when every $\qsk = (\rho_i)_{i\in[k]}$ the adversary
receives is generated using an SKE-CD encryption of $r$ (instead of
$0^\secp$ as in the scheme), then any tuple of $k$ MI-ABE keys of the
adversary having attributes $\big(t_1 \| u_1 \| y, t_2 \| u_2 \ldots, t_k \| u_k\big)$
satisfies one of the three conditions:

\begin{itemize}
\item The values $t_1 \ldots t_k$ are not all the same.
\item $C_{\mathsf{ABE}}(y) \neq 0$.
\item For $u = u_1 \| \ldots \| u_k$, $D[\lock \xor r, \skecd.\sk](u)$
returns $\lock$.
\end{itemize}

Notice that the former condition ensures that the adversary cannot
interleave keys corresponding to different decryption-keys.
The last condition holds because the adversary never receives
$\abe.\sk_{i, 1 - x[i]}$ for any $i$ such that $\theta[i] = 0$.
Consequently, $u$ and $x$ are the same at all positions where
$\theta[i]=0$. This means that $\SKECD.\CDec(\skecd.\sk, u)$ outputs $r$ by the
classical decryption property of $\SKECD$ (Definition \ref{def:bb84}).
It is important to note that the positions $i$ where $\theta[i] = 1$
(the Hadamard positions) have no effect on the value output by
$\SKECD.\CDec$
as their purpose is just in the verification of deletion.
As a result, the security of
MI-ABE allows to simulate the adversary's view in this hybrid, as
no ``decryptable'' set of $k$ keys is given out.

Importantly, the switch from SKE-CD encryptions of $0^\secp$ to $r$ is
indistinguishable, given that the adversary deletes all the information
in the SKE-CD ciphertexts.  To enforce this, the deletion algorithm
requires the adversary to measure both the SKE-CD and MI-ABE registers
for each slot to obtain values $(c_i, d_i)_{i \in [k]}$. Then, given
the values $\{\abe.\sk_{i, 0} \xor \abe.\sk_{i, 1}\}_{i \in [k]}$ as
part of the verification key, the verification checks whether $x[i] =
c_i \xor d_i \cdot (\abe.\sk_{i, 0} \xor \abe.\sk_{i, 1})$ holds for
every $i \in [k]$ such that $\theta[i] = 1$. As a result, we are able
use a standard hybrid argument to turn any distinguisher (of the
$0^\secp$ and $r$ hybrids) into an attack on the certified deletion
security of the SKE-CD scheme.

\paragraph{Construction.}
We will construct an ABE-CR\textsuperscript{2}-SKL
scheme $\ABECRCRSKL = \ABECRCRSKL.(\Setup, \qKG, \Enc,$ $\qDec,
\qDel, \Vrfy)$ using the following building blocks:

\begin{itemize}
\item Multi-Input (Ciphertext-Policy) ABE Scheme $\MIABE =
\MIABE.(\Setup, \KG,\allowbreak \Enc, \Dec)$ for the following relation:

\begin{description}
    \item[$R(\widetilde{x} \| z,y_1, \cdots, y_k)$:] Let $\widetilde{x}$ be
interpreted as a circuit and $z$ as a $k$-input circuit.
\begin{itemize}
\item Parse $y_1 = t_1 \| u_1 \| y$.
\item Parse $y_i=t_i\|u_i$ for every $i\in 2, \ldots, k$.
\item If $t_i\ne t_j$ for some $i,j\in[k]$, output $1$. Otherwise, go to the next step.
\item If $z(u_1\|\cdots\|u_k) = \bot$ AND $\widetilde{x}(y) = 0$, output $0$
(decryptable). Else, output $1$.
\end{itemize}
\end{description}

\item Compute-and-Compare Obfuscation $\CCObf$ with the simulator
$\CCSim$.

\item BB84-based SKE-CD scheme $\SKECD =
\SKECD.(\KG,\qEnc,\qDec,\qDel,\Vrfy)$ with the classical
decryption algorithm $\SKECD.\CDec$.
\end{itemize}

The description of each algorithm of $\ABECRCRSKL$ is as follows.

\begin{description}
\item[$\ABECRCRSKL.\Setup(1^\secp)$:] $ $
\begin{itemize}
    \item Generate $(\abe.\pk, \abe.\msk)\gets\MIABE.\Setup(1^\secp)$.
        \item Generate $\skecd.\sk\gets\SKECD.\KG(1^\secp)$.
    \item Output $\ek\seteq\abe.\pk$ and
        $\msk\seteq(\abe.\msk,\skecd.\sk)$.
\end{itemize}

\item[$\ABECRCRSKL.\qKG(\msk, y)$:] $ $
\begin{itemize}
\item Parse $\msk=(\abe.\msk,\ske.\msk)$.
\item Sample a random value $t\gets\bit^\secp$.
\item Generate
$(\skecd.\qct,\skecd.\vk)\gets\SKECD.\qEnc(\skecd.\sk,0^\secp)$.
$\skecd.\vk$ is of the form
$(x,\theta)\in\bit^{k}\times\bit^{k}$, and $\skecd.\qct$ is of
the form
$\ket{\psi_1}_{\qreg{SKECD.CT_1}}\tensor\cdots\tensor\ket{\psi_{k}}_{\qreg{SKECD.CT_{k}}}$.

\item Generate $\abe.\sk_{1,b}\gets\MIABE.\KG(\abe.\msk, 1, t\|b\|y)$
  for each $b \in \bit$.

\item For every $i \in 2, \ldots, k$, do the following:
\begin{itemize}
\item Generate $\abe.\sk_{i,b}\gets\MIABE.\KG(\abe.\msk, i, t\|b)$
for each $b\in\bit$.
\end{itemize}

\item For every $i \in [k]$, do the following:
\begin{itemize}
\item Prepare a register $\qreg{ABE.SK_i}$ that is initialized
to $\ket{0\cdots0}_{\qreg{ABE.SK_i}}$.

\item Apply the map
$\ket{u}_{\qreg{SKECD.CT_i}}\ket{v}_{\qreg{ABE.SK_i}}\ra\ket{u}_{\qreg{SKECD.CT_i}}\ket{v\oplus\abe.\sk_{i,u}}_{\qreg{ABE.SK_i}}$
to the registers $\qreg{SKECD.CT_i}$ and $\qreg{ABE.SK_i}$, and
obtain the resulting state $\rho_i$. 
\end{itemize}

\item Output $\qsk:=(\rho_i)_{i \in [k]}$ and
$\vk\seteq\big(\skecd.\vk,
\{\abe.\sk_{i,0} \xor
\abe.\sk_{i,1}\}_{i\in[k]:\theta[i]=1}\big)$.

\end{itemize}

\item[$\ABECRCRSKL.\Enc(\ek,\widetilde{x},\msg)$:] $ $
\begin{itemize}
    \item Parse $\ek=\abe.\pk$.
\item Generate $\tlC\gets\CCSim(1^\secp,\pp_D,1)$, where $\pp_D$ consists of circuit parameters of $D$ defined in the security proof.
    \item Generate $\abe.\ct\gets\MIABE.\Enc(\abe.\pk,\widetilde{x} \| \tlC,\msg)$.
    \item Output $\ct\seteq\abe.\ct$.
\end{itemize}
 
\item[$\ABECRCRSKL.\qDec(\qsk,\ct)$:] $ $
\begin{itemize}
\item Parse $\qsk=(\rho_i)_{i\in[k]}$ and $\ct=\abe.\ct$. We
denote the register holding $\rho_i$ as $\qreg{SKECD.CT_i}\tensor\qreg{ABE.SK_i}$.

\item Prepare a register $\qreg{MSG}$ that is initialized to $\ket{0\cdots0}_{\qreg{MSG}}$

\item To the registers $\qreg{\bigotimes_{i\in[k]}ABE.SK_i}$ and
$\qreg{MSG}$, apply 
$\bigotimes_{i\in[k]}\ket{v_i}_{\qreg{ABE.SK_i}}\otimes\ket{w}_{\qreg{MSG}}\ra\bigotimes_{i\in[k]}\ket{v_i}_{\qreg{ABE.SK_i}}\otimes\ket{w\oplus\MIABE.\Dec(\abe.\ct,
v_1,\cdots,v_k)}_{\qreg{MSG}}$.

   \item Measure the register $\qreg{MSG}$ in the computational basis and output the result $\msg^\prime$.
\end{itemize}

\item[$\ABECRCRSKL.\qDel(\qsk)$:] $ $
\begin{itemize}
\item Parse $\qsk = (\rho_i)_{i \in [k]}$. Let the register 
holding $\rho_i$ be denoted as $\qreg{SKECD.CT_i}\tensor\qreg{ABE.SK_i}$.
\item For each $i \in [k]$, measure the registers $\qreg{SKECD_i}$ and $\qreg{ABE.SK_i}$ in the Hadamard basis to
obtain outcomes $c_i$ and $d_i$.
\item Output $\cert = (c_i,d_i)_{i \in [k]}$.
\end{itemize}

\item[$\ABECRCRSKL.\Vrfy(\vk,\cert)$:] $ $
\begin{itemize}
\item Parse $\vk = \big(\skecd.\vk=(x,\theta), \{\abe.\sk_{i,0}
\xor \abe.\sk_{i,1}\}_{i\in[k]:\theta[i]=1}\big)$ and $\cert =
(c_i,d_i)_{i \in [k]}$.

\item Output $\top$ if $x[i]=c_i\oplus
    d_i\cdot(\abe.\sk_{i,0}\oplus\abe.\sk_{i,1})$ holds for every
    $i\in[k]$ such that $\theta[i]=1$ and $\bot$ otherwise.

\end{itemize}
\end{description}

Let $\qsk \gets\ABECRCRSKL.\qKG(\msk, y)$. $\qsk$ is of the form
$(\rho_i)_{i \in [k]}$, where $\rho_i$ is of the following form, where
$(x, \theta)$ is the verification key of a BB84-based SKECD scheme:
\begin{align}
\rho_i=
\begin{cases}
    \ket{x[i]}\ket{\abe.\sk_{i,x[i]}} & if~~ \theta[i]=0\\
    \ket{0}\ket{\abe.\sk_{i,0}}+(-1)^{x[i]}\ket{1}\ket{\abe.\sk_{i,1}} & if~~ \theta[i]=1,
\end{cases}
\end{align}

Recall that $\abe.\sk_{i, b} = \MIABE.\KG(\abe.\msk, i, t\|b)$ for every
$i \in \{2, \ldots, k\}$ and $b \in \bit$. Moreover, 
$\abe.\sk_{1, b} = \MIABE.\KG(\abe.\msk, 1, t\|b\|y)$ for each $b \in
\bit$.

\paragraph{Decryption correctness.}
The MI-ABE relation defined in the construction is as follows:

\begin{description}
\item[$R(\widetilde{x} \| z,y_1, \cdots, y_k)$:] Let $\widetilde{x}$ be
interpreted as a circuit and $z$ as a $k$-input circuit.
\begin{itemize}
\item Parse $y_1 = t_1 \| u_1 \| y$.
\item Parse $y_i=t_i\|u_i$ for every $i\in 2, \ldots, k$.
\item If $t_i\ne t_j$ for some $i,j\in[k]$, output $1$. Otherwise, go to the next step.
\item If $z(u_1\|\cdots\|u_k) = \bot$ AND $\widetilde{x}(y) = 0$, output $0$
(decryptable). Else, output $1$.
\end{itemize}
\end{description}

Clearly, $z \seteq \tlC \gets \CCSim(1^\secp, \pp_D, 1)$ always outputs $\bot$
and $t_1 = \ldots = t_k$ holds by construction. Hence, the guarantee
follows from the decryption correctness of $\MIABE$, as long as
$\widetilde{x}(y) = 0$ holds, as desired.

\paragraph{Verification correctness.}

Recall that the verification only checks the Hadamard basis positions,
i.e, positions $i \in [k]$ such that $\theta[i] = 1$. Consider the
outcome $(c_i, d_i)$ obtained by measuring the state
$\ket{0}\ket{\abe.\sk_{i,0}}+(-1)^{x[i]}\ket{1}\ket{\abe.\sk_{i,1}}$
in the Hadamard basis, where $c_i$ denotes the first bit of the
outcome. It is easy to see that $x[i] = c_i \xor d_i \cdot
(\abe.\sk_{i, 0} \xor \abe.\sk_{i, 1})$ is satisfied. Hence, the
verification correctness follows.

\subsection{Proof of Selective IND-KLA Security}
Let $\qA$ be an adversary for the selective IND-KLA security of $\ABECRCRSKL$.
We consider the following sequence of experiments.
\begin{description}
\item[$\hybi{0}^\coin$:] This is
$\expc{\ABECRCRSKL,\qA}{sel}{ind}{kla}(1^\secp,\coin)$.

\begin{enumerate}
\item $\qA$ declares the challenge ciphertext attribute $x^* \in \cX$.
\item The challenger $\qCh$ generates
    $(\abe.\pk,\abe.\msk)\gets\MIABE.\Setup(1^\secp)$ and
    $\skecd.\sk \gets \SKECD.\KG(1^\secp)$, and sends $\ek\seteq\abe.\pk$ to $\qA$.

\item $\qA$ can get access to the following (stateful) oracles, where
the list $L_{\qKG}$ used by the oracles is initialized to an empty
list:


\begin{description}
\item[$\Oracle{\qKG}(y)$:] Given $y$, it finds an entry of the form
$(y,\vk_y,V_y)$ from $\List{\qKG}$. If there is such an entry, it returns $\bot$.
Otherwise, it generates $(\qsk_y,\vk_y)\la\qKG(\msk,y)$,
sends $\qsk_y$ to $\qA$, and adds $(y,\vk_y,\bot)$ to $\List{\qKG}$.
    
\item[$\Oracle{\Vrfy}(y, \cert)$:] Given $(y, \cert)$, it finds an
entry $(y,\vk_y,V_y)$ from $\List{\qKG}$. (If there is no such entry, it returns $\bot$.) 
It then runs $\decision \seteq \Vrfy(\vk_y, \cert)$ and returns
$\decision$ to $\qA$. If $V_y=\bot$, it updates the entry into
$(y,\vk_y,\decision)$.
\end{description}

\item $\qA$ sends $(\msg_0^*,\msg_1^*)\in \cM^2$ to the challenger. If
$V_j=\bot$ for some $j\in[q]$, $\qCh$ outputs $0$ as the final
output of this experiment. Otherwise, $\qCh$ generates
$\tlC^*\gets\CCSim(1^\secp, \pp_D, 1)$ and
$\abe.\ct^*\gets\MIABE.\Enc(\abe.\pk,x^* \| \tlC^*,\msg_\coin^*)$, and sends
$\ct^*\seteq\abe.\ct^*$ to $\qA$.

\item $\qA$ outputs $\coin^\prime$. $\qCh$ outputs $\coin'$ as the final output of the experiment.
\end{enumerate}

\item[$\hybi{1}^\coin$:]This is the same as $\hybi{0}^\coin$ except
that $\tlC^*$ is generated as $\tlC^*\gets\CCObf(1^\secp,D[\lock\oplus
r, \skecd.\sk], \lock,0)$, where $\lock\gets\bit^\secp$, $r\gets\bit^\secp$ and
$D[\lock\oplus r, \skecd.\sk](x)$ is a circuit that has $\lock\oplus
r$ and $\skecd.\sk$ hardwired and outputs $\lock\oplus
r\oplus\SKECD.\CDec(\skecd.\sk,x)$.

Since $\lock$ is chosen at random independently of all other
variables, from the security of compute-and-compare obfuscation, we
have that $\Hyb_0^\coin \approx \Hyb_1^\coin$.

\item[$\hybi{2}^\coin$:]This is the same as $\hybi{1}^\coin$ except
that $\skecd.\qct$ generated as part of $\ABECRCRSKL.\qKG(\msk, y)$
such that $x^*(y) = 0$, is generated as
$(\skecd.\qct,\allowbreak\skecd.\vk)\gets\SKECD.\qEnc(\skecd.\sk,r)$.

In the previous step, we changed the distribution of $\tlC^*$ used
to generate the challenge ciphertext so that $\tlC^*$ has
$\skecd.\sk$ hardwired. However, the ciphertext is given to $\qA$
after $\qA$ deletes all the leased secret keys satisfying the ABE
relation, and thus the
corresponding ciphertetexts of $\SKECD$.  Thus, we can use the security of
$\SKECD$ to argue that $\hybi{1}^\coin \approx \hybi{2}^\coin$. Let
$q$ be the number of key-queries made to $\Oracle{\qKG}$ that satisfy
the relation wrt $x^*$. Consider now $q+1$ hybrids $\Hyb_{1,0}^\coin, \cdots, 
\Hyb_{1,q}^\coin$ where for every $l \in [q]$, $\Hyb_{1,l}^\coin$
is such that the first $l$ of the $q$ keys are generated using
$\SKECD$ encryptions of $r$ (instead of $0^\secp$). Notice that
$\Hyb_{1,0}^\coin \equiv \Hyb_1^\coin$ and
$\Hyb_{1,q}^\coin \equiv \Hyb_2^\coin$. Suppose that
$\Hyb_{1,l}^\coin \not\approx \Hyb_{1,l+1}^\coin$ for some $l \in
[q]$. Let $\qD$ be a distinguisher with non-negligible advantage in
distinguishing the hybrids.
We will construct the following reduction to the IND-CVA-CD
security of $\SKECD$:

\begin{description}
\item Execution of $\qR^\qD$ in
$\expc{\SKECD,\qR}{ind}{cva}{cd}(1^\secp,b)$:
\begin{enumerate}
\item The challenger $\qCh$ computes $\sk \gets \KG(1^\secp)$.
\item $\qD$ declares a challenge ciphertext attribute $x^*$ to $\qR$. $\qR$ generates $(\abe.\pk, \abe.\msk) \gets
\MIABE.\Setup(\allowbreak1^\secp)$ and sends $\ek \seteq \abe.\pk$ to $\qD$.

\item $\qR$ sends $(r_0,r_1) \seteq (0^\secp, r)$ to $\qCh$.
\item $\qCh$ computes $(\qct^\star, \vk^\star) \gets \qEnc(\sk,
r_b)$ and sends $\qct^\star$ to $\qR$.

\item $\qR$ simulates the oracle $\Oracle{\qKG}(y)$ for $\qD$ as
follows. Given $y$, it finds an entry $(y, \vk_y, V_y)$ from
$L_{\qKG}$. If there is such an entry, it returns $\perp$. Otherwise,
it proceeds as follows:

\begin{itemize}
\item For the $j$-th query (of the $q$ satisfying key-queries), if $j \neq l+1$, $\qR$ computes 
$\qsk_y$ and $\vk_y$ as in $\Hyb_0^\coin$, except
that the values $(\skecd.\qct, \skecd.\vk)$ are computed using
$\Oracle{\Enc}(r)$ for $j \in [l]$ and using
$\Oracle{\Enc}(0^\secp)$ for $j \in \{l+2, \ldots, q\}$.

\item For the $j$-th query (of the $q$ satisfying key-queries), if $j = l+1$, $\qR$ computes 
$\qsk_{y}$ as in $\Hyb_0^\coin$, except that the value $\qct^\star$ is used instead
of $\skecd.\qct$. Consider the values $\{\abe.\sk_{i,0},
\abe.\sk_{i,1}\}_{i\in[k] : \theta[i]=1}$ computed during the
computation of $\qsk_{y}$. $\qR$ computes the value
$\widetilde{\vk}_{y} = \{\abe.\sk_{i,0}\xor
\abe.\sk_{i,1}\}_{i\in[k]}$.

\item If $x^*(y) \neq 0$ (unsatisfying key-queries), $\qR$ computes
$\qsk_y$ and $\vk_y$ as in $\Hyb_0^\coin$, except that $(\skecd.\qct,
\skecd.\vk)$ are computed using $\Oracle{\Enc}(0^\secp)$.
\end{itemize}

\item $\qR$ simulates the view of oracle $\Oracle{\Vrfy}(y, \cert)$ for
$\qD$ as follows. Given $(y, \cert)$, it finds an entry $(y, \vk_y,
V_y)$ from $L_{\qKG}$. If there is no such entry, it returns $\perp$.
Otherwise, it proceeds as follows:
\begin{itemize}
\item If $x^*(y) \neq 0$, $\qR$ simulates access to $\Oracle{\Vrfy}$ as in $\Hyb_0^\coin$.

\item If $y$ corresponds to the $j$-th satisfying key-query to
$\Oracle{\qKG}$ and $j \neq l+1$,
$\qR$ simulates access to $\Oracle{\Vrfy}$ as in $\Hyb_0^\coin$.

\item If $y$ corresponds to the $j$-th satisfying key-query to
$\Oracle{\qKG}$ and $j = l+1$, $\qR$ simulates access to
$\Oracle{\Vrfy}(y,\cert)$ as follows:

\begin{enumerate}
    \item Parse $\widetilde{\vk}_{y} = \{\abe.\sk_{i,0}\xor
\abe.\sk_{i,1}\}_{i\in[k]}$ and $\cert=(c_i,d_i)_{i \in [k]}$.
\item Compute $x[i] = c_i \xor d_i\cdot(\abe.\sk_{i,0} \xor
\abe.\sk_{i,1})$ for every $i \in [k]$.
\item If $\Oracle{\Vrfy}(x) = \sk$, set $d \seteq \top$. Else,
      set $d \seteq \bot$.
\end{enumerate}

It returns $d$ to $\qD$. If $V_{y} = \bot$, it updates the entry in
$L_{\qKG}$ into $(y, \vk_y, d)$.
\end{itemize}

\item $\qD$ sends $(m_0^*, m_1^*) \in \cM^2$ to $\qR$. $\qR$ 
computes $\ct^*$ as in Step $4.$ of $\Hyb_1^\coin$ using
$\skecd.\sk \seteq \sk$ obtained from $\qCh$. It sends $\ct^*$ to
$\qD$ if $V_y = \top$ for every entry of the form $(y, \vk_y, V_y)$ in
$L_{\qKG}$ such that $x^*(y) = 0$. Else, it outputs $\bot$.
\item $\qD$ guesses a bit $b'$. $\qR$ sends $b'$ to $\qCh$.
\end{enumerate}
\end{description}
Notice that the view of $\qD$ in the reduction is that of
$\Hyb_{1,l}^\coin$ if $b=0$ and $\Hyb_{1,l+1}^\coin$ if $b=1$.
Moreover, $\qD$ can only distinguish when $V_y = \top$ for all entries
$(y, \vk_y, V_y) \in L_{\qKG}$ such that $x^*(y) = 0$.
This means that $V_{y}$ corresponding to the $(l+1)$-th satisfying
query to $\Oracle{\qKG}$ also
satisfies $V_y = \top$. Consequently, the
corresponding verification check of $\qCh$ is also $\top$. Hence, $\qR$ succeeds
in breaking IND-CVA-CD security of $\SKECD$. Therefore, we have
$\Hyb_1^\coin \approx \Hyb_2^\coin$.
\end{description}

We will now bound the distinguishing gap between $\hybi{2}^0$ and
$\hybi{2}^1$ using the security of $\MIABE$. Consider a
decryption key $\qsk_y = (\rho_i)_{i\in[k]}$ such that $x^*(y) = 0$. Recall that a ciphertext of BB84-based $\SKECD$ is of the form
$\ket{\psi_1}\tensor\cdots \tensor\ket{\psi_k}$, where
    \begin{align}
    \ket{\psi_i}=
    \begin{cases}
        \ket{x[i]} & \textrm{if}~~ \theta[i]=0\\
        \ket{0}+(-1)^{x[i]}\ket{1} & \textrm{if}~~ \theta[i]=1
    \end{cases}
    \end{align}
As a result, we have
    \begin{align}
    \rho_i=
    \begin{cases}
        \ket{x[i]}\ket{\abe.\sk_{i,x[i]}} & \textrm{if}~~ \theta[i]=0\\
        \ket{0}\ket{\abe.\sk_{i,0}}+(-1)^{x[i]}\ket{1}\ket{\abe.\sk_{i,1}} & \textrm{if}~~ \theta[i]=1,
    \end{cases}
    \end{align}
where $\abe.\sk_{1,b}\gets\MIABE.\KG(\abe.\msk,1,t_y\|b\|y)$ for each
$b \in \bit$ and
$\abe.\sk_{i,b}\gets\MIABE.\KG(\abe.\msk,i,\allowbreak t_y\|b)$
for each $i \in \{2, \ldots, k\}$ and $b \in \bit$. Here, $t_y$ is a
random value that is common for each of the $k$ ``slots'' of the
decryption-key $\qsk_y$. Notice first that no
two values $t_y, t_w$ are equal for $t_w$ corresponding to some
$\qsk_w$ except with negligible probability.
Consequently, due to the defined relation $R$ and the selective
security of $\MIABE$, any set of $k$ secret keys is ``decryptable'' only if
they correspond to the same decryption key $\qsk_y$. Now, we notice
that $\abe.\sk_{i,1-x[i]}$ is not given to $\qA$ for any $i\in[k]$
such that $\theta[i]=0$. This means that for any set of $k$ $\MIABE$
keys corresponding to the decryption-key $\qsk_y$, the attributes
$(x'[1], \cdots, x'[k])$ satisfy $x'[i] = x[i]$ for all $i:\theta[i] =
0$. Consequently, the classical decryption property (See Definition
\ref{def:bb84}) of the BB84-based SKE-CD scheme guarantees that $\lock
\xor r \xor \SKECD.\CDec(\skecd.\sk, x') = \lock$ in these hybrids. This means that
$\tlC^*(x') = 0$ (instead of $\bot$).  As a result, any subset of
$\MIABE$ keys of size $k$ given to $\qA$ is a ``non-decryptable'' set
in the hybrids $\Hyb_2^0$ and $\Hyb_2^1$.
It is easy to see now that we can reduce to the selective security of
$\MIABE$. Specifically, the reduction specifies the target attribute
$x^* \| \tlC^*$ where
$\tlC^* \gets \CCObf(1^\secp, D[\lock\xor r, \skecd.\sk], \lock, 0)$
and simulates the view for $\qA$ by querying the key-generation oracle
of $\MIABE$ accordingly. Note that the reduction can handle
verification queries since it can obtain both $\abe.\sk_{i,0}$ and
$\abe.\sk_{i,1}$ for every $i\in[k]$ such that $\theta[i]=1$ that are
sufficient to perform the verification.  (Especially, it does not need
$\abe.\sk_{i,1-x[i]}$ for $i\in[k]$ such that $\theta[i]=0$ for
verification.) We now state the following theorem:

\begin{theorem}
There exists an ABE-CR\textsuperscript{2}-SKL scheme satisfying
selective IND-KLA Security, assuming the existence of a
selectively-secure Multi-Input ABE scheme for polynomial-arity,
Compute-and-Compare Obfuscation, and a BB84-based SKE-CD scheme.
\end{theorem}

\else
\fi

\ifnum\anonymous=1
\else

\fi

	\ifnum\llncs=1
\bibliographystyle{splncs04}
\bibliography{bib/abbrev3,bib/crypto,bib/siamcomp_jacm,bib/other}
	\else
\bibliographystyle{alpha} 
\bibliography{bib/abbrev3,bib/crypto,bib/siamcomp_jacm,bib/other}

\newcommand{\etalchar}[1]{$^{#1}$}
\begin{thebibliography}{HKM{\etalchar{+}}24}

\bibitem[ABB10]{EC:AgrBonBoy10}
Shweta Agrawal, Dan Boneh, and Xavier Boyen.
\newblock Efficient lattice {(H)IBE} in the standard model.
\newblock In Henri Gilbert, editor, {\em EUROCRYPT~2010}, volume 6110 of {\em {LNCS}}, pages 553--572. Springer, Berlin, Heidelberg, May~/~June 2010.

\bibitem[AGKZ20]{STOC:AGKZ20}
Ryan Amos, Marios Georgiou, Aggelos Kiayias, and Mark Zhandry.
\newblock One-shot signatures and applications to hybrid quantum/classical authentication.
\newblock In Konstantin Makarychev, Yury Makarychev, Madhur Tulsiani, Gautam Kamath, and Julia Chuzhoy, editors, {\em 52nd ACM STOC}, pages 255--268. {ACM} Press, June 2020.

\bibitem[AHH24]{TCC:AnaHuHua24}
Prabhanjan Ananth, Zihan Hu, and Zikuan Huang.
\newblock Quantum key-revocable dual-regev encryption, revisited.
\newblock In Elette Boyle and Mohammad Mahmoody, editors, {\em TCC~2024, Part~III}, volume 15366 of {\em {LNCS}}, pages 257--288. Springer, Cham, December 2024.

\bibitem[AHU19]{C:AmbHamUnr19}
Andris Ambainis, Mike Hamburg, and Dominique Unruh.
\newblock Quantum security proofs using semi-classical oracles.
\newblock In Alexandra Boldyreva and Daniele Micciancio, editors, {\em CRYPTO~2019, Part~II}, volume 11693 of {\em {LNCS}}, pages 269--295. Springer, Cham, August 2019.

\bibitem[AJ15]{C:AnaJai15}
Prabhanjan Ananth and Abhishek Jain.
\newblock Indistinguishability obfuscation from compact functional encryption.
\newblock In Rosario Gennaro and Matthew J.~B. Robshaw, editors, {\em CRYPTO~2015, Part~I}, volume 9215 of {\em {LNCS}}, pages 308--326. Springer, Berlin, Heidelberg, August 2015.

\bibitem[AJS15]{EPRINT:AnaJaiSah15a}
Prabhanjan Ananth, Abhishek Jain, and Amit Sahai.
\newblock Indistinguishability obfuscation from functional encryption for simple functions.
\newblock Cryptology ePrint Archive, Report 2015/730, 2015.

\bibitem[AKN{\etalchar{+}}23]{EC:AKNYY23}
Shweta Agrawal, Fuyuki Kitagawa, Ryo Nishimaki, Shota Yamada, and Takashi Yamakawa.
\newblock Public key encryption with secure key leasing.
\newblock In Carmit Hazay and Martijn Stam, editors, {\em EUROCRYPT~2023, Part~I}, volume 14004 of {\em {LNCS}}, pages 581--610. Springer, Cham, April 2023.

\bibitem[AKY24]{myEPRINT:AgrKumYam24a}
Shweta Agrawal, Simran Kumari, and Shota Yamada.
\newblock Pseudorandom multi-input functional encryption and applications.
\newblock {\em {IACR} Cryptol. ePrint Arch.}, page 1720, 2024.

\bibitem[AL21]{EC:AnaLaP21}
Prabhanjan Ananth and Rolando~L. {La Placa}.
\newblock Secure software leasing.
\newblock In Anne Canteaut and Fran\c{c}ois-Xavier Standaert, editors, {\em EUROCRYPT~2021, Part~II}, volume 12697 of {\em {LNCS}}, pages 501--530. Springer, Cham, October 2021.

\bibitem[AMP24]{myEPRINT:AnaMutPor24}
Prabhanjan Ananth, Saachi Mutreja, and Alexander Poremba.
\newblock Revocable encryption, programs, and more: The case of multi-copy security.
\newblock {\em {IACR} Cryptol. ePrint Arch.}, page 1687, 2024.

\bibitem[APV23]{TCC:AnaPorVai23}
Prabhanjan Ananth, Alexander Poremba, and Vinod Vaikuntanathan.
\newblock Revocable cryptography from learning with errors.
\newblock In Guy~N. Rothblum and Hoeteck Wee, editors, {\em TCC~2023, Part~IV}, volume 14372 of {\em {LNCS}}, pages 93--122. Springer, Cham, November~/~December 2023.

\bibitem[ARYY23]{C:ARYY23}
Shweta Agrawal, M{\'e}lissa Rossi, Anshu Yadav, and Shota Yamada.
\newblock Constant input attribute based (and predicate) encryption from evasive and tensor {LWE}.
\newblock In Helena Handschuh and Anna Lysyanskaya, editors, {\em CRYPTO~2023, Part~IV}, volume 14084 of {\em {LNCS}}, pages 532--564. Springer, Cham, August 2023.

\bibitem[ATY23]{C:AgrTomYad23}
Shweta Agrawal, Junichi Tomida, and Anshu Yadav.
\newblock Attribute-based multi-input {FE} (and more) for attribute-weighted sums.
\newblock In Helena Handschuh and Anna Lysyanskaya, editors, {\em CRYPTO~2023, Part~IV}, volume 14084 of {\em {LNCS}}, pages 464--497. Springer, Cham, August 2023.

\bibitem[AWY20]{TCC:AgrWicYam20}
Shweta Agrawal, Daniel Wichs, and Shota Yamada.
\newblock Optimal broadcast encryption from {LWE} and pairings in the standard model.
\newblock In Rafael Pass and Krzysztof Pietrzak, editors, {\em TCC~2020, Part~I}, volume 12550 of {\em {LNCS}}, pages 149--178. Springer, Cham, November 2020.

\bibitem[AY20]{EC:AgrYam20}
Shweta Agrawal and Shota Yamada.
\newblock Optimal broadcast encryption from pairings and {LWE}.
\newblock In Anne Canteaut and Yuval Ishai, editors, {\em EUROCRYPT~2020, Part~I}, volume 12105 of {\em {LNCS}}, pages 13--43. Springer, Cham, May 2020.

\bibitem[AYY22]{C:AgrYadYam22}
Shweta Agrawal, Anshu Yadav, and Shota Yamada.
\newblock Multi-input attribute based encryption and predicate encryption.
\newblock In Yevgeniy Dodis and Thomas Shrimpton, editors, {\em CRYPTO~2022, Part~I}, volume 13507 of {\em {LNCS}}, pages 590--621. Springer, Cham, August 2022.

\bibitem[BB20]{BB84}
Charles~H Bennett and Gilles Brassard.
\newblock Quantum cryptography: Public key distribution and coin tossing.
\newblock {\em arXiv preprint arXiv:2003.06557}, 2020.

\bibitem[BGG{\etalchar{+}}14]{EC:BGGHNS14}
Dan Boneh, Craig Gentry, Sergey Gorbunov, Shai Halevi, Valeria Nikolaenko, Gil Segev, Vinod Vaikuntanathan, and Dhinakaran Vinayagamurthy.
\newblock Fully key-homomorphic encryption, arithmetic circuit {ABE} and compact garbled circuits.
\newblock In Phong~Q. Nguyen and Elisabeth Oswald, editors, {\em EUROCRYPT~2014}, volume 8441 of {\em {LNCS}}, pages 533--556. Springer, Berlin, Heidelberg, May 2014.

\bibitem[BGI{\etalchar{+}}12]{JACM:BGIRSVY12}
Boaz Barak, Oded Goldreich, Russell Impagliazzo, Steven Rudich, Amit Sahai, Salil~P. Vadhan, and Ke~Yang.
\newblock On the (im)possibility of obfuscating programs.
\newblock {\em Journal of the {ACM}}, 59(2):6:1--6:48, 2012.

\bibitem[BGK{\etalchar{+}}24]{EC:BGKMRR24}
James Bartusek, Vipul Goyal, Dakshita Khurana, Giulio Malavolta, Justin Raizes, and Bhaskar Roberts.
\newblock Software with certified deletion.
\newblock In Marc Joye and Gregor Leander, editors, {\em EUROCRYPT~2024, Part~IV}, volume 14654 of {\em {LNCS}}, pages 85--111. Springer, Cham, May 2024.

\bibitem[BI20]{TCC:BroIsl20}
Anne Broadbent and Rabib Islam.
\newblock Quantum encryption with certified deletion.
\newblock In Rafael Pass and Krzysztof Pietrzak, editors, {\em TCC~2020, Part~III}, volume 12552 of {\em {LNCS}}, pages 92--122. Springer, Cham, November 2020.

\bibitem[BJK{\etalchar{+}}18]{SCN:BJKPW18}
Zvika Brakerski, Aayush Jain, Ilan Komargodski, Alain Passel{\`e}gue, and Daniel Wichs.
\newblock Non-trivial witness encryption and null-{iO} from standard assumptions.
\newblock In Dario Catalano and Roberto {De Prisco}, editors, {\em SCN 18}, volume 11035 of {\em {LNCS}}, pages 425--441. Springer, Cham, September 2018.

\bibitem[BJL{\etalchar{+}}21]{TCC:BJLPS21}
Anne Broadbent, Stacey Jeffery, S{\'e}bastien Lord, Supartha Podder, and Aarthi Sundaram.
\newblock Secure software leasing without assumptions.
\newblock In Kobbi Nissim and Brent Waters, editors, {\em TCC~2021, Part~I}, volume 13042 of {\em {LNCS}}, pages 90--120. Springer, Cham, November 2021.

\bibitem[BK23]{C:BarKhu23}
James Bartusek and Dakshita Khurana.
\newblock Cryptography with certified deletion.
\newblock In Helena Handschuh and Anna Lysyanskaya, editors, {\em CRYPTO~2023, Part~V}, volume 14085 of {\em {LNCS}}, pages 192--223. Springer, Cham, August 2023.

\bibitem[BKM{\etalchar{+}}23]{TCC:BKMPW23}
James Bartusek, Dakshita Khurana, Giulio Malavolta, Alexander Poremba, and Michael Walter.
\newblock Weakening assumptions for publicly-verifiable deletion.
\newblock In Guy~N. Rothblum and Hoeteck Wee, editors, {\em TCC~2023, Part~IV}, volume 14372 of {\em {LNCS}}, pages 183--197. Springer, Cham, November~/~December 2023.

\bibitem[B{\"U}W24]{AC:BrzUnaWoo24}
Chris Brzuska, Akin {\"U}nal, and Ivy K.~Y. Woo.
\newblock Evasive {LWE} assumptions: Definitions, classes, and counterexamples.
\newblock In Kai-Min Chung and Yu~Sasaki, editors, {\em ASIACRYPT~2024, Part~IV}, volume 15487 of {\em {LNCS}}, pages 418--449. Springer, Singapore, December 2024.

\bibitem[BV18]{JACM:BitVai18}
Nir Bitansky and Vinod Vaikuntanathan.
\newblock Indistinguishability obfuscation from functional encryption.
\newblock {\em Journal of the {ACM}}, 65(6):39:1--39:37, 2018.

\bibitem[BV22]{ITCS:BraVai22}
Zvika Brakerski and Vinod Vaikuntanathan.
\newblock Lattice-inspired broadcast encryption and succinct ciphertext-policy {ABE}.
\newblock In Mark Braverman, editor, {\em ITCS 2022}, volume 215, pages 28:1--28:20. {LIPIcs}, January~/~February 2022.

\bibitem[BZ13a]{EC:BonZha13}
Dan Boneh and Mark Zhandry.
\newblock Quantum-secure message authentication codes.
\newblock In Thomas Johansson and Phong~Q. Nguyen, editors, {\em EUROCRYPT~2013}, volume 7881 of {\em {LNCS}}, pages 592--608. Springer, Berlin, Heidelberg, May 2013.

\bibitem[BZ13b]{C:BonZha13}
Dan Boneh and Mark Zhandry.
\newblock Secure signatures and chosen ciphertext security in a quantum computing world.
\newblock In Ran Canetti and Juan~A. Garay, editors, {\em CRYPTO~2013, Part~II}, volume 8043 of {\em {LNCS}}, pages 361--379. Springer, Berlin, Heidelberg, August 2013.

\bibitem[{\c C}G24]{TCC:CakGoy24}
Alper {\c C}akan and Vipul Goyal.
\newblock Unclonable cryptography with unbounded collusions and impossibility of hyperefficient shadow tomography.
\newblock In Elette Boyle and Mohammad Mahmoody, editors, {\em TCC~2024, Part~III}, volume 15366 of {\em {LNCS}}, pages 225--256. Springer, Cham, December 2024.

\bibitem[CGJL23]{EPRINT:CGJL23}
Orestis Chardouvelis, Vipul Goyal, Aayush Jain, and Jiahui Liu.
\newblock Quantum key leasing for {PKE} and {FHE} with a classical lessor.
\newblock Cryptology {ePrint} Archive, Report 2023/1640, 2023.

\bibitem[CLLZ21]{C:CLLZ21}
Andrea Coladangelo, Jiahui Liu, Qipeng Liu, and Mark Zhandry.
\newblock Hidden cosets and applications to unclonable cryptography.
\newblock In Tal Malkin and Chris Peikert, editors, {\em CRYPTO~2021, Part~I}, volume 12825 of {\em {LNCS}}, pages 556--584, Virtual Event, August 2021. Springer, Cham.

\bibitem[CMP20]{ARXIV:ColMajPor20}
Andrea Coladangelo, Christian Majenz, and Alexander Poremba.
\newblock Quantum copy-protection of compute-and-compare programs in the quantum random oracle model.
\newblock {\em arXiv (CoRR)}, abs/2009.13865, 2020.

\bibitem[CV22]{Quantum:CulVid22}
Eric Culf and Thomas Vidick.
\newblock A monogamy-of-entanglement game for subspace coset states.
\newblock {\em Quantum}, 6:791, sep 2022.

\bibitem[FFMV23]{EC:FFMV23}
Danilo Francati, Daniele Friolo, Giulio Malavolta, and Daniele Venturi.
\newblock Multi-key and multi-input predicate encryption from learning with errors.
\newblock In Carmit Hazay and Martijn Stam, editors, {\em EUROCRYPT~2023, Part~III}, volume 14006 of {\em {LNCS}}, pages 573--604. Springer, Cham, April 2023.

\bibitem[FN94]{C:FiaNao93}
Amos Fiat and Moni Naor.
\newblock Broadcast encryption.
\newblock In Douglas~R. Stinson, editor, {\em CRYPTO'93}, volume 773 of {\em {LNCS}}, pages 480--491. Springer, Berlin, Heidelberg, August 1994.

\bibitem[GGG{\etalchar{+}}14]{EC:GGGJKL14}
Shafi Goldwasser, S.~Dov Gordon, Vipul Goyal, Abhishek Jain, Jonathan Katz, Feng-Hao Liu, Amit Sahai, Elaine Shi, and Hong-Sheng Zhou.
\newblock Multi-input functional encryption.
\newblock In Phong~Q. Nguyen and Elisabeth Oswald, editors, {\em EUROCRYPT~2014}, volume 8441 of {\em {LNCS}}, pages 578--602. Springer, Berlin, Heidelberg, May 2014.

\bibitem[GGSW13]{STOC:GGSW13}
Sanjam Garg, Craig Gentry, Amit Sahai, and Brent Waters.
\newblock Witness encryption and its applications.
\newblock In Dan Boneh, Tim Roughgarden, and Joan Feigenbaum, editors, {\em 45th ACM STOC}, pages 467--476. {ACM} Press, June 2013.

\bibitem[GKP{\etalchar{+}}13]{C:GKPVZ13}
Shafi Goldwasser, Yael~Tauman Kalai, Raluca~A. Popa, Vinod Vaikuntanathan, and Nickolai Zeldovich.
\newblock How to run {T}uring machines on encrypted data.
\newblock In Ran Canetti and Juan~A. Garay, editors, {\em CRYPTO~2013, Part~II}, volume 8043 of {\em {LNCS}}, pages 536--553. Springer, Berlin, Heidelberg, August 2013.

\bibitem[GKW17]{FOCS:GoyKopWat17}
Rishab Goyal, Venkata Koppula, and Brent Waters.
\newblock Lockable obfuscation.
\newblock In Chris Umans, editor, {\em 58th FOCS}, pages 612--621. {IEEE} Computer Society Press, October 2017.

\bibitem[GMM17]{C:GarMahMoh17}
Sanjam Garg, Mohammad Mahmoody, and Ameer Mohammed.
\newblock Lower bounds on obfuscation from all-or-nothing encryption primitives.
\newblock In Jonathan Katz and Hovav Shacham, editors, {\em CRYPTO~2017, Part~I}, volume 10401 of {\em {LNCS}}, pages 661--695. Springer, Cham, August 2017.

\bibitem[GPV08]{STOC:GenPeiVai08}
Craig Gentry, Chris Peikert, and Vinod Vaikuntanathan.
\newblock Trapdoors for hard lattices and new cryptographic constructions.
\newblock In Richard~E. Ladner and Cynthia Dwork, editors, {\em 40th ACM STOC}, pages 197--206. {ACM} Press, May 2008.

\bibitem[GVW12]{C:GorVaiWee12}
Sergey Gorbunov, Vinod Vaikuntanathan, and Hoeteck Wee.
\newblock Functional encryption with bounded collusions via multi-party computation.
\newblock In Reihaneh Safavi-Naini and Ran Canetti, editors, {\em CRYPTO~2012}, volume 7417 of {\em {LNCS}}, pages 162--179. Springer, Berlin, Heidelberg, August 2012.

\bibitem[GZ20]{EPRINT:GeoZha20}
Marios Georgiou and Mark Zhandry.
\newblock Unclonable decryption keys.
\newblock Cryptology ePrint Archive, Report 2020/877, 2020.

\bibitem[HKM{\etalchar{+}}24]{EC:HKMNPY24}
Taiga Hiroka, Fuyuki Kitagawa, Tomoyuki Morimae, Ryo Nishimaki, Tapas Pal, and Takashi Yamakawa.
\newblock Certified everlasting secure collusion-resistant functional encryption, and more.
\newblock In Marc Joye and Gregor Leander, editors, {\em EUROCRYPT~2024, Part~III}, volume 14653 of {\em {LNCS}}, pages 434--456. Springer, Cham, May 2024.

\bibitem[HMNY21]{AC:HMNY21}
Taiga Hiroka, Tomoyuki Morimae, Ryo Nishimaki, and Takashi Yamakawa.
\newblock Quantum encryption with certified deletion, revisited: Public key, attribute-based, and classical communication.
\newblock In Mehdi Tibouchi and Huaxiong Wang, editors, {\em ASIACRYPT~2021, Part~I}, volume 13090 of {\em {LNCS}}, pages 606--636. Springer, Cham, December 2021.

\bibitem[KMY24]{myEPRINT:KitMorYam24}
Fuyuki Kitagawa, Tomoyuki Morimae, and Takashi Yamakawa.
\newblock A simple framework for secure key leasing.
\newblock {\em ArXiv (CoRR)}, abs/2410.03413, 2024.

\bibitem[KN22]{AC:KitNis22}
Fuyuki Kitagawa and Ryo Nishimaki.
\newblock Functional encryption with secure key leasing.
\newblock In Shweta Agrawal and Dongdai Lin, editors, {\em ASIACRYPT~2022, Part~IV}, volume 13794 of {\em {LNCS}}, pages 569--598. Springer, Cham, December 2022.

\bibitem[KN23]{TCC:KitNis23}
Fuyuki Kitagawa and Ryo Nishimaki.
\newblock One-out-of-many unclonable cryptography: Definitions, constructions, and more.
\newblock In Guy~N. Rothblum and Hoeteck Wee, editors, {\em TCC~2023, Part~IV}, volume 14372 of {\em {LNCS}}, pages 246--275. Springer, Cham, November~/~December 2023.

\bibitem[KNY21]{TCC:KitNisYam21}
Fuyuki Kitagawa, Ryo Nishimaki, and Takashi Yamakawa.
\newblock Secure software leasing from standard assumptions.
\newblock In Kobbi Nissim and Brent Waters, editors, {\em TCC~2021, Part~I}, volume 13042 of {\em {LNCS}}, pages 31--61. Springer, Cham, November 2021.

\bibitem[KNY23]{TCC:KitNisYam23}
Fuyuki Kitagawa, Ryo Nishimaki, and Takashi Yamakawa.
\newblock Publicly verifiable deletion from minimal assumptions.
\newblock In Guy~N. Rothblum and Hoeteck Wee, editors, {\em TCC~2023, Part~IV}, volume 14372 of {\em {LNCS}}, pages 228--245. Springer, Cham, November~/~December 2023.

\bibitem[Lam79]{Lamport79}
Leslie Lamport.
\newblock Constructing digital signatures from a one-way function.
\newblock Technical Report SRI-CSL-98, SRI International Computer Science Laboratory, October 1979.

\bibitem[LLQZ22]{TCC:LLQZ22}
Jiahui Liu, Qipeng Liu, Luowen Qian, and Mark Zhandry.
\newblock Collusion resistant copy-protection for watermarkable functionalities.
\newblock In Eike Kiltz and Vinod Vaikuntanathan, editors, {\em TCC~2022, Part~I}, volume 13747 of {\em {LNCS}}, pages 294--323. Springer, Cham, November 2022.

\bibitem[MPY23]{EPRINT:MorPorYam23}
Tomoyuki Morimae, Alexander Poremba, and Takashi Yamakawa.
\newblock Revocable quantum digital signatures.
\newblock Cryptology {ePrint} Archive, Report 2023/1937, 2023.

\bibitem[Por23]{ITCS:Poremba23}
Alexander Poremba.
\newblock Quantum proofs of deletion for learning with errors.
\newblock In Yael~Tauman Kalai, editor, {\em ITCS 2023}, volume 251, pages 90:1--90:14. {LIPIcs}, January 2023.

\bibitem[Reg09]{JACM:Regev09}
Oded Regev.
\newblock On lattices, learning with errors, random linear codes, and cryptography.
\newblock {\em Journal of the {ACM}}, 56(6):34:1--34:40, 2009.

\bibitem[SW05]{EC:SahWat05}
Amit Sahai and Brent~R. Waters.
\newblock Fuzzy identity-based encryption.
\newblock In Ronald Cramer, editor, {\em EUROCRYPT~2005}, volume 3494 of {\em {LNCS}}, pages 457--473. Springer, Berlin, Heidelberg, May 2005.

\bibitem[Tsa22]{C:Tsabary22}
Rotem Tsabary.
\newblock Candidate witness encryption from lattice techniques.
\newblock In Yevgeniy Dodis and Thomas Shrimpton, editors, {\em CRYPTO~2022, Part~I}, volume 13507 of {\em {LNCS}}, pages 535--559. Springer, Cham, August 2022.

\bibitem[VWW22]{AC:VaiWeeWic22}
Vinod Vaikuntanathan, Hoeteck Wee, and Daniel Wichs.
\newblock Witness encryption and null-{IO} from evasive {LWE}.
\newblock In Shweta Agrawal and Dongdai Lin, editors, {\em ASIACRYPT~2022, Part~I}, volume 13791 of {\em {LNCS}}, pages 195--221. Springer, Cham, December 2022.

\bibitem[Wee22]{EC:Wee22}
Hoeteck Wee.
\newblock Optimal broadcast encryption and {CP}-{ABE} from evasive lattice assumptions.
\newblock In Orr Dunkelman and Stefan Dziembowski, editors, {\em EUROCRYPT~2022, Part~II}, volume 13276 of {\em {LNCS}}, pages 217--241. Springer, Cham, May~/~June 2022.

\bibitem[Wie83]{wiesner1983conjugate}
Stephen Wiesner.
\newblock Conjugate coding.
\newblock {\em ACM Sigact News}, 15(1):78--88, 1983.

\bibitem[Win99]{Winter99}
Andreas~J. Winter.
\newblock Coding theorem and strong converse for quantum channels.
\newblock {\em {IEEE} Trans. Inf. Theory}, 45(7):2481--2485, 1999.

\bibitem[WZ17]{FOCS:WicZir17}
Daniel Wichs and Giorgos Zirdelis.
\newblock Obfuscating compute-and-compare programs under {LWE}.
\newblock In Chris Umans, editor, {\em 58th FOCS}, pages 600--611. {IEEE} Computer Society Press, October 2017.

\bibitem[Zha12]{FOCS:Zhandry12}
Mark Zhandry.
\newblock How to construct quantum random functions.
\newblock In {\em 53rd FOCS}, pages 679--687. {IEEE} Computer Society Press, October 2012.

\end{thebibliography}
	\fi

\ifnum\cameraready=0
	\ifnum\llncs=0
	\appendix
	

\section{Quantum Secure ABE for All Relations from LWE}\label{sec-quantum-secure-ABE}

We show there exists quantum selective-secure ABE for all relations computable in polynomial time based on the LWE assumption.

\subsection{Preparation}

\begin{definition}[Quantum-Accessible Pseudo-Random Function]\label{def:prf}
Let $\{\PRF_{k}: \bin^{\ell_1} \ra \allowbreak \bin^{\ell_2} \mid k \in \bin^\secp\}$ be a family of polynomially computable functions, where $\ell_1$ and $\ell_2$ are some polynomials of $\secp$.
We say that $\PRF$ is a quantum-accessible pseudo-random function (QPRF) family if for any QPT adversary $\qA$, it holds that
\begin{align}
\advt{\qA}{prf}(\secp)
= \abs{\Pr[\qA^{\ket{\PRF_{k}(\cdot)}}(1^\secp) \gets 1 \mid k \gets \bit^{\secp}]
-\Pr[\qA^{\ket{\sfR(\cdot)}}(1^\secp) \gets 1 \mid \sfR \gets \cU]
}\leq\negl(\secp),
\end{align}
where $\cU$ is the set of all functions from $\bit^{\ell_1}$ to $\bit^{\ell_2}$. 
\end{definition}

\begin{theorem}[\cite{FOCS:Zhandry12}]\label{thm:qprf}
If there exists a OWF, there exists a QPRF.
\end{theorem}

\subsection{Proofs}

We first define Key-Policy ABE for polynomial size circuits and briefly see that it can be used to instantiate ABE for any relations computable in polynomial time, even under quantum selective-security.  
Then, we prove that the Key-Policy ABE scheme for polynomial size circuits by Boneh et al.~\cite{EC:BGGHNS14} with a light modification using QPRF satisfies quantum selective-security under the LWE assumption.  

\textbf{Key-Policy ABE for Circuits:}
Let $\cX_\secp=\bin^{n(\secp)}$ and $\cY_\secp$ be the set of all circuits with input space $\bin^{n(\secp)}$ and size at most $s(\secp)$, where $n$ and $s$ are some polynomials.
Let $R_\secp$ be the following relation:
$$R_\secp(x,y)=0 \iff y(x) = 0$$
An ABE scheme for such $\{\cX_\secp\}_\secp, \{\cY_\secp\}_\secp$,
and $\{R_\secp\}_\secp$ is referred to as a Key-Policy ABE scheme
for circuits. 

\begin{lemma}\label{lma:kp_to_cp}
If there exists a quantum selective-secure Key-Policy ABE scheme
for circuits, then there exists a quantum selective-secure ABE scheme for all relations computable in polynomial time.
\end{lemma}
\begin{proof}
Let $\cX_\secp\subseteq\bit^n$, $\cY_\secp\subseteq \bit^\ell$, and $R_\secp:\cX_\secp\times\cY_\secp\ra\bit$, where $n$ and $\ell$ are polynomials and $R_\secp$ is computable in polynomial time.
We construct ABE scheme $\ABE = (\Setup,\KG,\Enc,\Dec)$ with attribute spaces $\{\cX_\secp\}$ and $\{\cY_\secp\}$ and relation $\{R_\secp\}$, using Key-Policy ABE scheme for circuits $\ABE' =
(\Setup,\KG',\Enc,\Dec)$ with the following attribute spaces $\{\cX^\prime_\secp\}$ and $\{\cY^\prime_\secp\}$.
\begin{itemize}
\item $\cX^\prime_\secp=\bit^n$.
\item Let $C_{\secp,y}$ be the circuit such that $C_{\secp,y}(x)=R_\secp(x,y)$ for every $x\in\bit^n$ and $y\in\bit^\ell$, and let $s$ be the maximum size of $C_y$. Note that $s$ is a polynomial in $\secp$. Then, $\cY^\prime_\secp$ is the set of all circuits with input length $n$ and size at most $s^\prime$.
\end{itemize}
The scheme is as follows:
$\KG'(\msk, y, r) = \KG(\msk, C_{\secp,y}, r)$.
It is easy to see that the correctness of $\ABE$ follows from that of $\ABE^\prime$. For the
quantum selective security of $\ABE'$, consider a reduction $\qR$
to the quantum selective security of $\ABE$.

\begin{description}
\item Execution of $\qR^{\qA}$ in Experiment $\expc{\ABE,
\qA}{q}{sel}{ind}(1^\secp,\coin)$:
\begin{enumerate}
\item $\qR$ receives challenge attribute $x^* \in \cX_\secp$ from $\qA$ and
forwards it to the challenger $\Ch$.
\item For each oracle query $(\qreg{Y}, \qreg{Z})$ made by $\qA$,
$\qR$ performs the following map on
register $\qreg{Z}$ initialized to $\ket{0^{s(\secp)}}$:
$$\ket{y}_{\qreg{Y}}\ket{z}_{\qreg{Z}} \mapsto
\ket{y}_{\qreg{Y}}\ket{z \xor C_{\secp,y}}_{\qreg{Z}}$$
\item Then, $\qR$ queries the registers $\qreg{Y}, \qreg{Z}$ to
$\Oracle{qkg}$ followed by returning the registers $\qreg{Y},
\qreg{Z}$ to $\qA$.
\item When $\qA$ sends $(\msg_0, \msg_1)$ to $\qR$, $\qR$ forwards it to
$\Ch$.
\item On receiving $\ct^\star \leftarrow \Enc(\pk, x^*, \msg_\coin)$ from
$\Ch$, $\qR$ forwards it to $\qA$.
\item Finally, when $\qA$ outputs a guess $\coin'$, $\qR$ sends
    $\coin'$ to $\Ch$.
\end{enumerate}
\end{description}

Since the view of $\qA$ is identical to that in the
quantum selective security experiment for scheme $\ABE'$, $\qR$ ends
up breaking the quantum selective security of $\ABE$.
\end{proof}

\begin{theorem}
Assuming the polynomial hardness of the
LWE problem, there exists a quantum selective-secure
Key-Policy ABE scheme for circuits.
\end{theorem}
\begin{proof}
We claim that the Key-Policy ABE scheme for circuits by Boneh et al.
\cite{EC:BGGHNS14} based on LWE is quantum selective-secure.
Actually, we will alter the key-generation algorithm of their scheme
as follows: $\KG(\msk, y, k) = \KG'(\msk, y, \PRF_k(y))$ where
$\KG'$ is the key-generation algorithm of their construction with
explicit random coins $\PRF_k(y)$, and $\{\PRF_k\}_k$ is QPRF.
This is a common technique utilized in quantum security proofs (See
for Eg. \cite{EC:BonZha13,C:BonZha13}) that allows one to use a
common random value for every term of a superposition. In the
following discussion, whenever we discuss a hybrid titled
"$\mathsf{Game}\;i$" for some value $i$, it refers to the
corresponding hybrid in Theorem 4.2 of \cite{EC:BGGHNS14}. Also, any
indistinguishability claims between $\mathsf{Game}$ hybrids that are
mentioned to be previously established, will refer to their work.
Consider the following sequence of hybrids for the aforementioned
ABE scheme $\ABE$ and a QPT adversary $\qA$:

\begin{description}
\item[$\hybi{0}^\coin$:] This is the same as the experiment
    $\expc{\ABE,\qA}{q}{sel}{ind}(1^\secp, \coin)$.



\item[$\hybi{1}^\coin$:] This is similar to $\Hyb_0^\coin$, except
that the public-key $\pk$ is generated based on the challenge
attribute $x^*$, as in the hybrid $\mathsf{Game}\;1$. It was shown
that the hybrids $\mathsf{Game}\;0$ (the original experiment) and
$\mathsf{Game}\;1$ are statistically indistinguishable. By the same
argument, it follows that $\Hyb_0^\coin \approx_s \Hyb_1^\coin$.

\item[$\hybi{2}^\coin$:] This is similar to $\Hyb_1^\coin$, except
that the public-key is further modified as in $\mathsf{Game}\;2$.
Unlike the change made in $\Hyb_1^\coin$ though, this change
requires the key-queries to be answered differently as in
$\mathsf{Game}\;2$. More specifically, the hybrid now has a
punctured master secret-key that still allows it to simulate
every key-query $y$ such that $C_y(x^*) \neq \bot$. It was shown
previously that $\mathsf{Game}\;1 \approx_s \mathsf{Game}\;2$.
Consider now the intermediate hybrids $\widetilde{\Hyb}_1^\coin,
\widetilde{\Hyb}_2^\coin$ that are similar to $\Hyb_1^\coin,
\Hyb_2^\coin$ respectively, except that all the superposition
key-queries in these hybrids are responded using independent (true)
randomness in every term of the superposition. We will now restate
the following lemma by \cite{EC:BonZha13}, which we will use to
argue that $\widetilde{\Hyb}_1^\coin \approx_s
\widetilde{\Hyb}_2^\coin$.

\begin{lemma}\cite{EC:BonZha13}
Let $\cY$ and $\cZ$ be sets and for each $y \in \cY$, let $D_y$ and
$D'_y$ be distributions on $\cZ$ such that $\SD(D_y, D'_y) \le
\epsilon$.  Let $O:\cY \ra \cZ$ and $O':\cY \ra \cZ$ be functions
such that $O(y)$ outputs $z \gets D_y$ and $O'(y)$ outputs $z' \gets
D'_y$. Then, $O(y)$ and $O'(y)$ are $\epsilon'$-statistically
indistinguishable by quantum algorithms making $q$ superposition
oracle queries, such that $\epsilon' = \sqrt{8C_0q^3\epsilon}$ where
$C_0$ is a constant.
\end{lemma}

Recall that $\cY$ denotes the set of key-attributes. Let us fix a
challenge ciphertext attribute $x^*$ for the following discussion.
For each $y \in \cY$ let the distributions $D^1_y[\pk], D^2_y[\pk]$
correspond to how a key for attribute $y$ is sampled in the hybrids
$\widetilde{\Hyb}^\coin_1, \widetilde{\Hyb}^\coin_2$ respectively,
conditioned on the public key being $\pk$. Note that for each $i \in
[2]$, we consider $D_y^i[\pk]$ to also output the public-key $\pk$
along with the secret-key for $y$. We know that on average over
$\pk$, $D_y^1[\pk] \approx_s D_y^2[\pk]$ holds for all $y \in \cY$.
It now follows from the above lemma that on average over $\pk$, the
analogously defined oracles $O_1[\pk], O_2[\pk]$ are statistically
indistinguishable by algorithms making polynomially many
superposition queries. Observe that with access to oracle $O_i[\pk]$
for each $i \in [2]$, the view of $\qA$ in
$\widetilde{\Hyb}^\coin_i$ (conditioned on the public-key being
$\pk$) can easily
be recreated as the oracle outputs $\pk$, which can then be used to
compute the challenge ciphertext corresponding to $\coin$.
Consequently, it follows that the hybrids $\widetilde{\Hyb}^\coin_1$
and $\widetilde{\Hyb}^\coin_2$ are statistically indistinguishable.
Since $\widetilde{\Hyb}^\coin_i \approx \Hyb^\coin_i$ holds for each
$i\in[2]$ (by the quantum-security of $\PRF$), we have that
$\Hyb^\coin_1 \approx \Hyb^\coin_2$.

\item[$\hybi{3}^\coin$:] This is similar to $\Hyb_2^\coin$, except
that the challenge ciphertext is chosen uniformly at random, as in
$\mathsf{Game}\;3$.

Observe that $\Hyb_3^0 \equiv \Hyb_3^1$. It was shown previously
that $\mathsf{Game}\;3 \approx \mathsf{Game}\;2$ by a reduction to
LWE. Specifically, the reduction prepares the setup based on the LWE
sample and the challenge ciphertext $x^*$, and plants the LWE
challenge in the challenge ciphertext. It is easy to see that a
similar reduction works in our case, thereby showing that
$\Hyb_2^\coin \approx \Hyb_3^\coin$. Consequently, it follows that
$\Hyb_0^0 \approx \Hyb_0^1$.
\end{description}
This completes the proof.
\end{proof}


\newcommand{\xbar}{\bar{x}}
\newcommand{\thetabar}{\bar{\theta}}
\newcommand{\FV}{\mathtt{FV}}
\newcommand{\cla}{\mathsf{cla}}

\section{IND-CVA-CD Secure BB84 Based SKE-CD}\label{sec:SKECD-BB84}

To prove \cref{thm:SKECD-BB84}, we show how to transform IND-CD secure SKE-CD to IND-CVA-CD secure one.

\begin{definition}[IND-CD Security]\label{def-ind-cd}
We define the security experiment $\expb{\CDSKE,\qA}{ind}{cd}(1^\secp,\coin)$ in the same way as $\expc{\CDSKE,\qA}{ind}{cva}{cd}(1^\secp,\coin)$ except that the adversary $\qA$ is allowed to get access to the verification oracle only once.
We say that $\CDSKE$ is IND-CD secure if for any QPT $\qA$, it holds that
\begin{align}
\advb{\CDSKE,\qA}{ind}{cd}(1^\secp)\seteq \abs{\Pr[
\expb{\CDSKE,\qA}{ind}{cd}(1^\secp, 0)=1] - \Pr[
\expb{\CDSKE,\qA}{ind}{cd}(1^\secp, 1)=1] }\le \negl(\secp).
\end{align}
\end{definition}

Bartusek and Khurana~\cite{C:BarKhu23} showed the following theorem.

\begin{theorem}[\cite{C:BarKhu23}]
There exists an IND-CD secure BB84 based SKE-CD scheme assuming just an IND-CPA secure SKE scheme.
\end{theorem}

We prove \cref{thm:SKECD-BB84} by proving the following theorem.

\begin{theorem}
Let $\CDSKE=(\KG,\qEnc,\qDec,\qDel,\Vrfy)$ be a BB84-based SKE-CD scheme.
If $\CDSKE$ is IND-CD secure, then it is also IND-CVA-CD secure.
\end{theorem}

\begin{proof}
Let $\qA$ be a QPT adversary that attacks the IND-CVA-CD security of $\CDSKE$ with $q$ verification queries.
Let $\FV_i$ be the event that the $i$-th verification query is the first verification query such that the answer for it is not $\bot$.
Then, $\qA$'s advantage can be described as follows.
\begin{align}
\advc{\CDSKE,\qA}{ind}{cva}{cd}(1^\secp)
&\seteq \abs{
\Pr[\expc{\CDSKE,\qA}{ind}{cva}{cd}(1^\secp, 0)=1]
-\Pr[\expc{\CDSKE,\qA}{ind}{cva}{cd}(1^\secp, 1)=1]
}\\
&\le \sum_{i\in[q]} \abs{
\Pr[\expc{\CDSKE,\qA}{ind}{cva}{cd}(1^\secp, 0)=1 \land \FV_i]
-\Pr[\expc{\CDSKE,\qA}{ind}{cva}{cd}(1^\secp, 1)=1 \land \FV_i]\label{eqn-SKECD-BB84}
}\\
\end{align}
To bound each term of \cref{eqn-SKECD-BB84}, we construct the following QPT adversary $\qB_i$ that attacks the IND-CD security of $\CDSKE$ using $\qA$.

\begin{enumerate}
    \item $\qB_i$ initializes $\qA$ with the security parameter $1^\secp$.
    \item $\qA$ makes queries to the encryption oracle throughout the experiment.
    \begin{description}
    \item[$\Oracle{\qEnc}(\msg)$:] When $\qA$ makes an encryption query $\msg$, $\qB_i$ forwards it to its own encryption oracle, and sends back the answer $(\vk,\qct)$ from the encryption oracle to $\qA$.
    \end{description}
    \item When $\qA$ outputs $(\msg_0,\msg_1)\in\cM^2$, $\qB_i$ sends $(\msg_0,\msg_1)$ to its challenger. On receiving the challenge ciphertext $\qct^*$ from the challenger, $\qB_i$ measures its classical part $\cla$. This does not affect $\qct^*$. $\qB_i$ then forwards $\qct^*$ to $\qA$.
    \item Hereafter, $\qA$ can get access to the following oracle.
    \begin{description}
        \item[$\Oracle{\Vrfy}(\cert_j)$:] For the $j$-th query $\cert_j$, if $j<i$, $\qB_i$ returns $\bot$ to $\qA$. if $j=i$, $\qB_i$ queries $\cert_j$ to its verification oracle. If the response is $\bot$, $\qB_i$ aborts. Otherwise if the answer is $\sk$, $\qB_i$ forwards $\sk$ to $\qA$. $\qB_i$ also computes $\theta$ from $\cla$ and $\sk$, and sets $\vk^\prime=(\theta,\cert_i)$. $\qB_i$ checks whether $\Vrfy(\vk^\prime,\cert_j)=\bot$ holds for every $j<i$ (that is, whether the $i$-th query is the first query resulting in the answer other than $\bot$). If not, $\qB_i$ aborts. Otherwise, $\qB_i$ responds to the subsequent verification queries using $\vk^\prime$.
       
    \end{description}
    \item When $\qA$ outputs $\coin^\prime\in \bit$, $\qB_i$ outputs $\coin^\prime$.
\end{enumerate}

Let $\vk^*=(\theta,x)$ be the verification key corresponding to $\qct^*$. For any string $\cert$, if $\Vrfy(\vk^*,\cert)=\top$, $\Vrfy(\vk^*,\cdot)$ and $\Vrfy(\vk^\prime,\cdot)$ for $\vk^\prime=(\theta,\cert)$ are functionally equivalent, and $\vk^\prime$ can be used as an alternative verification key. This is because $\Vrfy(\vk^*,\cert^\prime)$ and $\Vrfy(\vk^\prime,\cert^\prime)$ respectively checks whether $\cert^\prime[i]=x[i]$ and $\cert^\prime[i]=\cert[i]$ holds or not for every $i\in[n]$ such that $\theta[i]=1$, and for such $i$, we have $x[i]=\cert[i]$ from the fact that $\Vrfy(\vk^*,\cert)=\top$.

From the above, after the $i$-th verification query from $\qA$ is responded, $\qB_i$ can check whether its simulation of $\qA$ so far has been successful or not. Moreover, if the simulation has failed, $\qB_i$ aborts, and otherwise, $\qB_i$ can successfully simulate the remaining steps for $\qA$ using the alternative verification key $\vk^\prime=(\theta,\cert_i)$.
Then, we have $\Pr[\expb{\CDSKE,\qB_i}{ind}{cd}(1^\secp, \coin)=1]=\Pr[\expc{\CDSKE,\qA}{ind}{cva}{cd}(1^\secp, \coin)=1 \land \FV_i]$.
Since, $\SKECD$ satisfies IND-CD security, it holds that 
\[\abs{\Pr[\expb{\CDSKE,\qB_i}{ind}{cd}(1^\secp, 0)=1]-\Pr[\expb{\CDSKE,\qB_i}{ind}{cd}(1^\secp, 1)=1]}=\negl(\secp)\] for every $i\in[q]$, which shows each term of \cref{eqn-SKECD-BB84} is negligible.
This completes the proof.
\end{proof}

\section{Construction of SKFE-CR-SKL with Key Testability}\label{sec:SKFESKL-KT}

\subsection{Construction}\label{sec:SKFECRSKL-KT-construction}
We construct an SKFE-CR-SKL scheme for the functionality $F:\cX\times\cY\ra\cZ$ with key testability $\SKFESKL=
\SKFESKL.(\Setup,\qKG,\Enc,\qDec,\allowbreak\qVrfy)$ having the
additional algorithms $\SKFESKL.(\CDec, \KeyTest)$, using the
following building blocks.

\begin{itemize}
\item BB84-based SKE-CD scheme (Definition \ref{def:bb84}) $\SKECD =
\SKECD.(\KG,\qEnc,\qDec,\allowbreak \qDel,\Vrfy)$ having the classical
decryption algorithm $\SKECD.\CDec$.

\item Classical SKFE scheme $\SKFE=\SKFE.(\Setup,\KG,\Enc,\Dec)$ for the functionality $F:\cX\times\cY\ra\cZ$.

\item OWF $f:\bit^\secp\ra\bit^{p(\secp)}$ for some polynomial $p$.
\end{itemize}

The construction is as follows:

\begin{description}

\item[$\SKECRSKL.\Setup(1^\secp)$:] $ $
\begin{enumerate}
    \item Generate $\skecd.\sk\gets\SKECD.\KG(1^\secp)$.
    \item Generate $\skfe.\msk\gets\SKFE.\Setup(1^\secp)$.
    \item Output $\msk\seteq(\skecd.\sk, \skfe.\msk)$.
\end{enumerate}

\item[$\SKFESKL.\qKG(\msk,y)$:] $ $
\begin{enumerate}
    \item Parse $\msk=(\skecd.\sk, \skfe.\msk)$.
    \item Generate $\skfe.\sk_y \gets \SKFE.\KG(\skfe.\msk,y)$.
    \item Generate
        $(\skecd.\qct,\skecd.\vk)\gets\SKECD.\qEnc(\skecd.\sk,\skfe.\sk_y)$.
        Recall that $\skecd.\vk$ is of the form
        $(x,\theta)\in\bit^{\ctlen}\times\bit^{\ctlen}$, and
        $\skecd.\qct$ is of the form
        $\ket{\psi_1}_{\qreg{SKECD.CT_1}}\tensor\cdots\tensor\ket{\psi_{\ctlen}}_{\qreg{SKECD.CT_{\ctlen}}}$.

    \item Generate $s_{i,b}\la\bit^\secp$ and compute $t_{i,b}\la f(s_{i,b})$ for every $i\in[\ctlen]$ and $b\in\bit$. 
    Set $T\seteq
    t_{1,0}\|t_{1,1}\|\cdots\|t_{\ctlen,0}\|t_{\ctlen,1}$ and $S =
    \{s_{i,0} \xor s_{i, 1}\}_{i \in [\ctlen] \; : \;\theta[i] =
    1}$.
    \item Prepare a register $\qreg{S_i}$ that is initialized to
    $\ket{0^\secp}_{\qreg{S_i}}$ for every $i\in[\ctlen]$. 
    \item For every $i\in[\ctlen]$, apply the map
    \begin{align}
    \ket{u_i}_{\qreg{SKECD.CT_i}}\tensor\ket{v_i}_{\qreg{S_i}}
    \ra
    \ket{u_i}_{\qreg{SKECD.CT_i}}\tensor\ket{v_i\oplus s_{i,u_i}}_{\qreg{S_i}}
    \end{align}
    to the registers $\qreg{SKECD.CT_i}$ and $\qreg{S_i}$ and obtain the resulting state $\rho_i$.
    \item Output $\qsk_y = (\rho_i)_{i\in{[\ctlen]}}$,
    $\vk=(x,\theta,S)$, and $\tk=T$.
\end{enumerate}

\item[$\SKFESKL.\Enc(\msk, x)$:] $ $
\begin{enumerate}
    \item Parse $\msk = (\skecd.\sk, \skfe.\msk)$.
    \item Generate $\skfe.\ct\gets\SKFE.\Enc(\skfe.\msk,x)$.
    \item Output $\ct\seteq(\skecd.\sk, \skfe.\ct)$.
    
\end{enumerate}

\item[$\SKFESKL.\CDec(\sk, \ct)$:] $ $
\begin{enumerate}
\item Parse $\ct = (\skecd.\sk, \skfe.\ct)$. Parse $\sk$ as a string over
the registers $\qreg{SKECD.CT} = \qreg{SKECD.CT_1} \otimes \cdots
\otimes \qreg{SKECD.CT_{\ctlen}}$ and $\qreg{S} =
\qreg{S_1} \otimes \cdots \otimes \qreg{S_{\ctlen}}$. Let
$\widetilde{\sk}$ be the sub-string of $\sk$ on register
$\qreg{SKECD.CT}$.
\item Compute $\skfe.\sk \gets \SKECD.\CDec(\skecd.\sk, \widetilde{\sk})$.
\item Output $z\gets\SKFE.\Dec(\skfe.\sk,\skfe.\ct)$.
\end{enumerate}

\item[$\SKFESKL.\qDec(\qsk, \ct)$:] $ $
\begin{enumerate}
\item Parse $(\rho_i)_{i \in [\ctlen]}$. We denote the register
holding $\rho_i$ as $\qreg{SKECD.CT_i}\tensor\qreg{S_i}$ for
every $i\in[\ctlen]$.

\item Prepare a register $\qreg{MSG}$ of $\msglen$ qubits that is
initialized to $\ket{0\cdots0}_{\qreg{MSG}}$.

\item Apply the map
\begin{align}
\ket{u}_{\bigotimes_{i\in[\ctlen]}\qreg{SKECD.
CT_i}} \tensor\ket{w}_{\qreg{MSG}} \ra
\ket{u}_{\bigotimes_{i\in[\ctlen]}\qreg{SKECD.CT_i}}\tensor\ket{w\oplus
\SKFESKL.\CDec(u, \ct)}_{\qreg{MSG}}
\end{align}

to the registers
$\bigotimes_{i\in[\ctlen]}\qreg{SKECD.CT_i}$ and
$\qreg{MSG}$.
\item Measure $\qreg{MSG}$ in the computational basis and output the result $\msg^\prime$.
\end{enumerate}

\item[$\SKFESKL.\qVrfy(\vk,\qsk)$:] $ $
\begin{enumerate}
\item Parse $\vk = (x,\theta,S=\{s_{i,0} \xor
    s_{i,1}\}_{i\in[\ctlen]\; : \; \theta[i]=1})$ and
    $\qsk = (\rho_i)_{i\in[\ctlen]}$ where $\rho_i$ is a state on
    the registers $\qreg{SKECD.CT_i}$ and $\qreg{S_i}$.
\item For every $i \in [\ctlen]$, measure $\rho_i$ in the Hadamard
basis to get outcomes $c_i, d_i$ corresponding to 
the registers $\qreg{SKECD.CT_i}$ and $\qreg{S_i}$ respectively.

\item Output $\top$ if $x[i]=c_i \oplus d_i\cdot(s_{i,0}\oplus
        s_{i,1})$ holds for every $i\in[\ctlen]$ such that $\theta[i]=1$.
    Otherwise, output $\bot$.
\end{enumerate}

\item[$\SKFESKL.\KeyTest(\tk,\sk)$:] $ $
\begin{enumerate}
\item Parse $\sk$ as a string over the registers
$\qreg{SKECD.CT} = \qreg{SKECD.CT_1} \otimes \cdots \otimes
\qreg{SKECD.CT_{\ctlen}}$ and $\qreg{S} = \qreg{S_1} \otimes
\cdots \otimes \qreg{S_{\ctlen}}$.
Let $u_i$ denote the value on
$\qreg{SKECD.CT_i}$ and $v_i$ the value on $\qreg{S_i}$. Parse $\tk$
as $T=t_{1,0}\|t_{1,1}\|\cdots \|t_{\ctlen,0}\|t_{\ctlen,1}$.

\item Let $\Check[t_{i,0},t_{i_1}](u_i,v_i)$ be the deterministic
algorithm that outputs $1$ if $f(v_i)=t_{i,u_i}$ holds and $0$
otherwise.

\item Output $\Check[t_{1,0},t_{1,1}](u_1,v_1) \land
\Check[t_{2,0},t_{2,1}](u_2,v_2) \land \cdots \land
\Check[t_{\ctlen,0},t_{\ctlen,1}](u_{\ctlen},v_{\ctlen})$.
\end{enumerate}

\end{description}

\subsection{Proof of Selective
Single-Ciphertext KLA Security}\label{proof:sel-1ct-kla}
Let $\qA$ be an adversary for the selective
single-ciphertext KLA security of the
construction $\SKFESKL$ that makes use of a BB84-based
SKE-CD scheme $\SKECD$. Consider the hybrid $\Hyb_j^\coin$ defined
as  follows:

\begin{description}
\item[$\hybi{j}^\coin$:] $ $
\begin{enumerate}
\item Initialized with $1^\secp$, $\qA$ outputs $(x_0^*, x_1^*)$.
Sample $\msk \gets \SKFESKL.\Setup(1^\secp)$.

\item $\qA$ can get access to the following (stateful) oracles,
where the list $\List{\qKG}$ used by the oracles is initialized
to an empty list:

\begin{description}
\item[$\Oracle{\qKG}(y)$:] Given $y$, it finds an entry of the form
$(y,\vk,V)$ from $\List{\qKG}$. If there is such an entry, it
returns $\bot$. Otherwise it proceeds as follows:
\begin{enumerate}[(i)]
\item 
If this is the $k$-th query for $k \le j$
and $F(x_0^*, y) \neq F(x_1^*, y)$, then compute $(\qsk,
\vk, \tk) \gets \qKGt(\msk, y)$ where $\qKGt$ is defined below.
Otherwise, compute $(\qsk, \vk, \tk) \gets \qKG(\msk, y)$.

\item It sends $\qsk$ and $\tk$ to $\qA$ and adds $(y, \vk, \bot)$
to $L_{\qKG}$.
\end{enumerate}

\item[$\Oracle{\qVrfy}(y,\widetilde{\qsk})$:] Given
$(y,\widetilde{\qsk})$, it finds an entry $(y,\vk,V)$ from
$\List{\qKG}$. (If there is no such entry, it returns $\bot$.) It
then runs $\decision \gets \qVrfy(\vk,\widetilde{\qsk})$ and returns
$\decision$ to $\qA$. If $V=\top$, it updates the entry into
$(y,\vk,\decision)$. 
\end{description}

\item[$\qKGt(\msk)$:] Differences from $\qKG$ are colored in red:
\begin{enumerate}
    \item Parse $\msk=(\skecd.\sk, \skfe.\msk)$.
    \item \textcolor{red}{Generate $r \gets \bit^\secp$.}
    \item Generate
        $(\skecd.\qct,\skecd.\vk)\gets\SKECD.\qEnc(\skecd.\sk,\textcolor{red}{r})$.
        $\skecd.\vk$ is of the form
        $(x,\theta)\in\bit^{\ctlen}\times\bit^{\ctlen}$, and
        $\skecd.\qct$ is of the form
        $\ket{\psi_1}_{\qreg{SKECD.CT_1}}\tensor\cdots\tensor\ket{\psi_{\ctlen}}_{\qreg{SKECD.CT_{\ctlen}}}$.

    \item Generate $s_{i,b}\la\bit^\secp$ and compute $t_{i,b}\la f(s_{i,b})$ for every $i\in[\ctlen]$ and $b\in\bit$. 
    Set $T\seteq
    t_{1,0}\|t_{1,1}\|\cdots\|t_{\ctlen,0}\|t_{\ctlen,1}$ and $S =
    \{s_{i,0} \xor s_{i, 1}\}_{i \in [\ctlen] \; : \;\theta[i] =
    1}$.
    \item Prepare a register $\qreg{S_i}$ that is initialized to
    $\ket{0^\secp}_{\qreg{S_i}}$ for every $i\in[\ctlen]$. 
    \item For every $i\in[\ctlen]$, apply the map
    \begin{align}
    \ket{u_i}_{\qreg{SKECD.CT_i}}\tensor\ket{v_i}_{\qreg{S_i}}
    \ra
    \ket{u_i}_{\qreg{SKECD.CT_i}}\tensor\ket{v_i\oplus s_{i,u_i}}_{\qreg{S_i}}
    \end{align}
    to the registers $\qreg{SKECD.CT_i}$ and $\qreg{S_i}$ and obtain the resulting state $\rho_i$.
    \item Output $\qsk_y = (\rho_i)_{i\in{[\ctlen]}}$,
    $\vk=(x,\theta,S)$, and $\tk=T$.
\end{enumerate}

\item If there exists
an entry $(y,\vk,V)$ in $\List{\qKG}$ such that $F(x_0^*,y)\ne
F(x_1^*,y)$ and $V=\bot$, output $0$. Otherwise, generate
$\ct^*\la\Enc(\msk,x_\coin^*)$ and send $\ct^*$ to $\qA$.

\item $\qA$ continues to make queries to $\Oracle{\qKG}$. However, $\qA$ is not allowed to send $y$ such that $F(x_0^*,y)\ne F(x_1^*,y)$ to $\Oracle{\qKG}$.

\item $\qA$ outputs a guess $\coin^\prime$ for $\coin$. Output
$\coin'$.
\end{enumerate}
\end{description}

Let $\qA$ make $q = \poly(\secp)$ many queries to $\Oracle{\qKG}$
before the challenge phase. We will now prove the following lemma:

\begin{lemma}
$\forall j \in \{0, \ldots, q-1\}$ and $\coin \in \bit: \Hyb_j^\coin
\approx_c \Hyb_{j+1}^\coin$.
\end{lemma}
\begin{proof}
Suppose $\Hyb_j^\coin \not \approx \Hyb_{j+1}^\coin$. Let $\qD$ be
a corresponding distinguisher. We will
construct a reduction $\qR$ that breaks the IND-CVA-CD security of
the BB84-based SKE-CD scheme $\SKECD$. The execution of $\qR^\qD$ in
the experiment $\expc{\SKECD,\qR}{ind}{cva}{cd}(1^\secp,b)$ proceeds
as follows:

\begin{description}
\item Execution of $\qR^\qD$ in
$\expc{\SKECD,\qR}{ind}{cva}{cd}(1^\secp,b)$:

\begin{enumerate}
\item The challenger $\qCh$ computes $\skecd.\sk \leftarrow
\SKECD.\KG(1^\lambda)$.
\item $\qR$ initializes $\qD$ with $1^\secp$ which outputs $(x_0^*,
x_1^*)$.
\item $\qR$ samples $\skfe.\msk \gets \SKFE.\Setup(1^\secp)$. It initializes a list
$L_{\qKG}$ to an empty list.
\item $\qR$ simulates the oracle $\Oracle{\qKG}$ for $\qD$ as follows:
\begin{description}
\item $\Oracle{\qKG}(y):$ Given $y$, it finds an entry of the form
$(y, \vk, V)$ from $L_{\qKG}$. If there is such an entry, it returns
$\bot$. Otherwise, it proceeds as follows for the $k$-th query:
\begin{enumerate}
\item If $k \le j$ and $F(x_0^*, y)
\neq F(x_1^*, y)$, then compute $(\qsk, \vk, \tk)$ as in
$\widetilde{\qKG}$ except that the pair $(\skecd.\qct, \skecd.\vk)$ is
obtained as the output of $\Oracle{\qEnc}(r)$ instead. Send $\qsk$ and
$\tk$ to $\qA$ and add $(y, \vk, \bot)$ to $L_\qKG$.

\item If $k=j+1$ and $F(x_0^*, y) \neq F(x_1^*, y)$,  compute
$\skfe.\sk_y \gets \SKFE.\KG(\skfe.\msk, y)$. Then, send $(r^*,
\skfe.\sk_y)$ to $\qCh$ where $r^*$ is a random value. On receiving
$\skecd.\ct^*$ from $\qCh$, compute $\qsk$ and $\tk$ as in $\qKG$,
except that $\skecd.\ct^*$ is used in place of $\skecd.\ct$.
Send $\qsk$ and $\tk$ to $\qA$ and add $(y, 0^\secp, \bot)$ to
$L_{\qKG}$.

\item If $k > j+1$ or $F(x_0^*, y) = F(x_1^*, y)$, then compute
$(\qsk, \vk, \tk)$ as in $\qKG$ except that the pair
$(\skecd.\qct, \skecd.\vk)$ is replaced with the output of
$\Oracle{\Enc}(\skfe.\sk_y)$ instead. Send $\qsk$ and $\tk$ to $\qA$
and add $(y, \vk, \bot)$ to $L_{\qKG}$.
\end{enumerate}
\end{description}

\item $\qR$ simulates the oracle $\Oracle{\qVrfy}$ for $\qD$ as
follows:

\begin{description}
\item $\Oracle{\qVrfy}(y, \widetilde{\qsk}):$ Given $(y,
\widetilde{\qsk})$, it finds an entry $(y, \vk, V)$ from $L_\qKG$
(If there is no such entry, it returns $\bot$.) It then proceeds as
follows, if $y$ corresponds to the $k$-th query made to
$\Oracle{\qKG}$ and $k \neq j+1$:

\begin{enumerate}
\item Parse $\vk = (x,\theta,S=\{s_{i,0} \xor
s_{i,1}\}_{i\in[\ctlen]\;:\;\theta[i]=1})$ and $\widetilde{\qsk} =
(\rho_i)_{i\in[\ctlen]}$.

\item For every $i \in [\ctlen]$, measure $\rho_i$ in the Hadamard
basis to get outcomes $c_i, d_i$ corresponding to the registers
$\qreg{SKECD.CT_i}$ and $\qreg{S_i}$ respectively.

\item Compute $\cert[i] = c_i \xor d_i \cdot (s_{i,0} \xor
s_{i,1})$ for every $i \in [\ctlen]$.

\item If $x[i] = \cert[i]$ holds
for every $i \in [\ctlen] : \theta[i] = 1$, then update the entry to
$(y, \vk, \top)$ and send $\top$ to $\qD$. Else, send $\bot$.
\end{enumerate}

If $k = j+1$, then it proceeds as follows:
\begin{enumerate}
\item Compute $\cert = \cert[1] \| \ldots \| \cert[\ctlen]$ where
each $\cert[i]$ is computed as in the previous case.
\item Send $\cert$ to $\qCh$. If $\qCh$ returns $\skecd.\sk$, send
$\top$ to $\qD$ and update the corresponding entry to $(y, \vk,
\top)$. Else, output $\bot$.
\end{enumerate}
\end{description}

\item $\qD$ requests the challenge ciphertext. If there exists
an entry $(y,\vk,V)$ in $\List{\qKG}$ such that $F(x_0^*,y)\ne
F(x_1^*,y)$ and $V=\bot$, output $0$. Otherwise, compute and send $\ct^\star =
(\skecd.\sk, \SKFE.\Enc(\skfe.\msk, x^*_\coin))$ to $\qD$.

\item $\qD$ continues to make queries to $\Oracle{\qKG}$. However, $\qA$ is not allowed to send $y$ such that $F(x_0^*,y)\ne F(x_1^*,y)$ to $\Oracle{\qKG}$.
Consequently, $\qR$ simulates these queries as per Step 4. (c) above.

\item $\qD$ outputs a guess $\coin^\prime$ for $\coin$ which $\qR$
forwards to $\qCh$. $\qCh$ outputs $\coin^\prime$ as the output of the
experiment.
\end{enumerate}
\end{description}

We will first argue that when $b=1$, the view of $\qD$ is exactly the
same as its view in the hybrid $\Hyb_j^\coin$. Notice that the for the
first $j$ queries, if $F(x_0^*, y) \neq F(x_1^*, y)$ for a key-query
corresponding to $y$, then the decryption key is computed by querying
the encryption oracle on a random plaintext. If this condition does
not hold, then the key is computed by querying the encryption oracle
on the corresponding SKFE key $\skfe.\sk_y$. The hybrid $\Hyb_j$ on
the other hand, directly computes these values, but there is no
difference in the distribution of the output ciphertexts. A similar
argument holds for the keys $\qdk_{j+2}, \ldots, \qdk_q$, which
contain encryptions of the corresponding keys $\skfe.\sk_y$. Note
that if $F(x_0^*, y^*) = F(x_1^*, y^*)$, where $y^*$ corresponds to
the $j+1$-th query, then the hybrids $\Hyb_j^\coin$ and
$\Hyb_{j+1}^\coin$ are identical. Hence, consider the case when
$F(x_0^*, y^*) \neq F(x_1^*, y^*)$. In this case, if $b=1$,
the value encrypted as part of the key $\qsk_{j+1}$ is
$\skfe.\sk_{y^*}$. This is the same as in $\Hyb_j^\coin$. As for the
verification oracle queries, notice that they are answered similarly
by the reduction and $\Hyb_j^\coin$ for all but the $j+1$-th key. For
the $j+1$-th key, the reduction works differently in that it forwards
the certificate $\cert$ to the verification oracle.  However, the
verification procedure of the BB84-based SKE-CD scheme checks the
validity of the value $\cert$ in the same way as the reduction, so
there is no difference.
Finally, notice that when $b=0$, the encrypted value is an independent
and random value, similar to the hybrid $\Hyb_{j+1}^\coin$.
Consequently, $\qR$ breaks the IND-CVA-CD security of $\SKECD$ with
non-negligible probability, a contradiction.
\end{proof}

Notice now that the hybrid $\Hyb_0^\coin$ is the same as the
experiment $\expc{\SKECRSKL,\qA}{sel}{1ct}{kla}\allowbreak(1^\secp,\coin)$. From
the previous lemma, we have that $\Hyb_0^\coin \approx_c
\Hyb_q^\coin$. However, we have that $\Hyb_q^0 \approx_c \Hyb_q^1$
holds from the selective single-ciphertext security of the underlying
SKFE scheme $\SKFE$. This is because any key-query corresponding to
$y$ after the $q$-th query is such that $F(x_0^*, y) = F(x_1^*, y)$.
For the first $q$ queries, wherever this condition doesn't hold, the
SKFE keys have been replaced with random values. Consequently, we have
that $\Hyb_0^0 \approx_c \Hyb_0^1$, which completes the proof. \qed

\subsection{Proof of Key-Testability}\label{proof:kt_SKFE}
First, we will argue the correctness requirement. Recall that
$\SKFESKL.\qKG$ applies the following map to a BB84 state
$\ket{x}_\theta$, where $(x, \theta) \in
\bit^{\ctlen}\times\bit^{\ctlen}$, for every $i \in [\ctlen]$:

\begin{align}
\ket{u_i}_{\qreg{SKECD.CT_i}}\tensor\ket{v_i}_{\qreg{S_i}}
\ra
\ket{u_i}_{\qreg{SKECD.CT_i}}\tensor\ket{v_i\oplus
s_{i,u_i}}_{\qreg{S_i}}
\end{align}
where $\qreg{SKECD.CT_i}$ denotes the register holding the $i$-th
qubit of $\ket{x}_\theta$ and $\qreg{S_i}$ is a register initialized
to $\ket{0\ldots0}_{\qreg{S_i}}$.

Consider applying the algorithm $\SKFESKL.\KeyTest$ in
superposition to the resulting state, i.e., performing the following
map, where $\qreg{\SKECD.CT} = \qreg{\SKECD.CT_1} \otimes \cdots
\otimes \qreg{\SKECD.CT_{\ctlen}}$ and $\qreg{S} =
\qreg{S_1} \otimes \cdots \otimes \qreg{S_{\ctlen}}$, and
$\qreg{KT}$ is initialized to $\ket{0}$:

\begin{align}
\ket{u}_{\qreg{SKECD.
CT}}\tensor\ket{v}_{\qreg{S}}\tensor\ket{\beta}_{\qreg{KT}} \ra
\ket{u}_{\qreg{SKECD.
CT}}\tensor\ket{v}_{\qreg{S}}\tensor\ket{\beta\oplus\SKFESKL.
\KeyTest(\tk, u\|v)}_{\qreg{KT}}
\end{align}

where $\tk = T = t_{1, 0}\|t_{1, 1} \| \cdots \| t_{\ctlen,
0}\|t_{\ctlen, 1}$. Recall that $\SKFESKL.\KeyTest$
outputs 1 if
and only if $\Check[t_{i,0}, t_{i, 1}](u_i,
v_i) = 1$ for every $i
\in [\ctlen]$, where $u_i, v_i$ denote the states of the registers
$\qreg{SKECD.CT_i}$ and $\qreg{S_i}$ respectively.
Recall that $\Check[t_{i,0},t_{i,1}](u_i, v_i)$ computes $f(v_i)$
and checks if it equals $t_{i, u_i}$. Since the construction chooses
$t_{i, u_i}$ such that $f(s_{i, u_i}) = t_{i, u_i}$, this check
always passes. Consequently, measuring register $\qreg{KT}$ always
produces outcome $1$.

We will now argue that the security requirement holds by showing the
following reduction to the security of the OWF $f$. Let $\qA$ be an
adversary that breaks the key-testability of $\SKFESKL$. Consider
a QPT reduction $\qR$ that works as follows in the OWF security
experiment:

\begin{description}
\item Execution of $\qR^\qA$ in
$\expa{f,\qR}{owf}(1^\secp)$:

\begin{enumerate}
\item The challenger chooses $s \leftarrow \bit^\lambda$ and sends
$y^* = f(s)$ to $\qR$.

\item $\qR$ runs $\SKFESKL.\Setup(1^\secp)$ and initializes $\qA$
with input $\msk$.

\item $\qR$ picks a random $k^* \in [q]$.

\item $\qR$ simulates the access of $\qA$ to the oracle
$\Oracle{\qKG}$ as follows, where the list $\List{\qKG}$ is
initialized to an empty list:

\begin{description}
\item[$\Oracle{\qKG}(y)$:] For the $k$-th query, do the following:
\begin{enumerate}
\item 
Given $y$, it finds an entry of the form
$(y,\tk)$ from $\List{\qKG}$. If there is such an entry, it
returns $\bot$.

\item
Otherwise, if $k \neq k^*$, it generates
$(\qsk_y,\vk,\tk)\la\qKG(\msk,y)$, sends $(\qsk_y,\vk,\tk)$ to $\qA$,
and adds $(y,\tk)$ to $\List{\qKG}$.

\item Otherwise, if $k = k^*$, it generates $(\qsk_y, \vk, \tk) \gets
\widetilde{\qKG}(\msk, y)$ (differences from $\qKG$ are colored in
\textcolor{red}{red}). It then sends $(\qsk_y, \vk, \tk)$ to $\qA$ and
adds $(y, \tk)$ to $L_{\qKG}$.
\end{enumerate}
\end{description}

\begin{description}
\item $\widetilde{\qKG}(\msk, y)$
\begin{enumerate}
    \item Parse $\msk=(\skecd.\sk, \skfe.\msk)$.
    \item Generate $\skfe.\sk_y \gets \SKFE.\KG(\skfe.\msk,y)$.
    \item Generate
        $(\skecd.\qct,\skecd.\vk)\gets\SKECD.\qEnc(\skecd.\sk,\skfe.\sk_y)$.
        Here, $\skecd.\vk$ is of the form
        $(x,\theta)\in\bit^{\ctlen}\times\bit^{\ctlen}$, and
        $\skecd.\qct$ is of the form
        $\ket{\psi_1}_{\qreg{SKECD.CT_1}}\tensor\cdots\tensor\ket{\psi_{\ctlen}}_{\qreg{SKECD.CT_{\ctlen}}}$.

\item \textcolor{red}{Choose an index $i^\star \in [\ctlen]$ such
that $\theta[i^\star] = 0$. For every $i \in [\ctlen]$ such
that $i \neq i^\star$, generate $s_{i,b}\la\bit^\secp$ and compute
$t_{i,b}\la f(s_{i,b})$ for every $b\in\bit$. For $i = i^\star$,
set $t_{i^\star, 1 - x[i^\star]} = y^*$. Then, generate $s_{i^\star,
x[i^\star]} \leftarrow \bit^\lambda$ and compute $t_{i^\star,
x[i^\star]} = f(s_{i^\star, x[i^\star]})$.}
   Set $T\seteq
    t_{1,0}\|t_{1,1}\|\cdots\|t_{\ctlen,0}\|t_{\ctlen,1}$ and $S =
    \{s_{i,0} \xor s_{i, 1}\}_{i \in [\ctlen] \; : \;\theta[i] =
    1}$.
    \item Prepare a register $\qreg{S_i}$ that is initialized to
    $\ket{0^\secp}_{\qreg{S_i}}$ for every $i\in[\ctlen]$. 
    \item For every $i\in[\ctlen]$, apply the map
    \begin{align}
    \ket{u_i}_{\qreg{SKECD.CT_i}}\tensor\ket{v_i}_{\qreg{S_i}}
    \ra
    \ket{u_i}_{\qreg{SKECD.CT_i}}\tensor\ket{v_i\oplus s_{i,u_i}}_{\qreg{S_i}}
    \end{align}
    to the registers $\qreg{SKECD.CT_i}$ and $\qreg{S_i}$ and obtain the resulting state $\rho_i$.
    \item Output $\qsk_y = (\rho_i)_{i\in{[\ctlen]}}$,
    $\vk=(x,\theta,S)$, and $\tk=T$.
\end{enumerate}
\end{description}

\item $\qA$ sends a tuple of classical strings $(y, \sk, x^*)$ to $\qR$.
$\qR$ outputs $\bot$ if there is no entry of the form $(y,\tk)$ in
$\List{\qKG}$ for some $\tk$. Also, if $k \neq k^*$, $\qR$ outputs $\bot$.
Otherwise, $\qR$ parses $\sk$ 
as a string over the registers $\qreg{SKECD.CT} = \qreg{SKECD.CT_1}
\otimes \cdots \otimes
\qreg{SKECD.CT_{\ctlen}}$ and $\qreg{S} = \qreg{S_1} \otimes
\cdots \otimes \qreg{S_{\ctlen}}$ and measures the register
$\qreg{S_{i^\star}}$ to obtain an outcome $s_{i^\star}$. $\qR$
then sends $s_{i^\star}$ to the challenger.
\end{enumerate}
\end{description}

Notice that the view of $\qA$ is the same as its view in the
key-testability experiment, as only the value $t_{i^\star,
1-x[i^\star]}$ is generated differently by forwarding the value $y$,
but this value is distributed identically to the original value.
Note that in both cases, $\qA$ receives no information about a
pre-image of $t_{i^\star, 1-x[i^\star]}$.
Now, $\qR$ guesses the index $k$ that $\qA$ targets with probability
$\frac1q$. By assumption, we have that $\CDec(\sk, \ct)
\neq F(x^*, y)$ where $\ct = \Enc(\msk, x^*)$. The value $\sk$ can be
parsed as a string over the registers $\qreg{SKECD.CT}$ and
$\qreg{S}$. Let $\widetilde{\sk}$ be the sub-string of $\sk$ on the
register $\qreg{SKECD.CT}$. Recall that $\CDec$ invokes the
algorithm $\SKECD.\CDec$ on input $\widetilde{\sk}$. We will
now recall a property of $\SKECD.\CDec$ that was specified in
Definition \ref{def:bb84}:

Let $(\qct, \vk = (x, \theta)) \gets \SKECD.\Enc(\skecd.\sk,
\skfe.\sk_y)$
where $\skecd.\sk \gets \SKECD.\KG(1^\secp)$. Now, let $u$ be any
arbitrary value such that $u[i] = x[i]$ for all $i : \theta[i] = 0$.
Then, the following holds:

$$\Pr\Big[\SKECD.\CDec(\skecd.\sk, u) = \skfe.\sk_y\Big] \ge 1 -
\negl(\secp)$$

Consequently, if $\widetilde{\sk}$ is such that $\widetilde{\sk}[i]
= x[i]$ for all $i: \theta[i] = 0$, where $(x, \theta)$ are
specified by $\vk_{k^\star}$, then $\SKECD.\CDec(\skecd.\sk,
\widetilde{\sk})$ outputs the value $\skfe.\sk_y$ with high
probability. Since $\CDec(\sk, \ct = (\skecd.\sk, \skfe.\ct =
\SKFE.\Enc\allowbreak(\skfe.\msk, x^*)))$ outputs
$\SKFE.\Dec(\SKECD.\CDec(\skecd.\sk, \widetilde{\sk}), \skfe.\ct)$, we
have that it outputs $x^*$ with high probability from the decryption
correctness of $\SKFE$. Therefore, it must be the case that there
exists some index $i$ for which $\widetilde{\sk} \neq x[i]$. With
probability $\frac{1}{\ctlen}$, this happens to be the guessed value
$i^\star$.  In this case, $\qA$ must output $s_{i^\star}$ on register
$\qreg{S_i}$ such that $f(s_{i^\star}) = t_{i^\star, 1 - x[i^\star]} =
y^*$. This concludes the proof. \qed

Since we have proved selective single-ciphertext security and
key-testability, we can now state the following theorem:

\begin{theorem}
Assuming the existence of a BB84-based SKE-CD scheme and the existence
of OWFs, there exists a selective single-ciphertext KLA secure
SKFE-CR-SKL scheme satisfying the key-testability property.
\end{theorem}

\else
	\newpage
	 	\appendix
	 	\setcounter{page}{1}
 	{
	\noindent
 	\begin{center}
	{\Large SUPPLEMENTAL MATERIALS}
	\end{center}
 	}
	\setcounter{tocdepth}{2}
	\ifnum\noaux=1
 	\else
	\fi


\section{Omitted Preliminaries}\label{sec:omitted_prelim}
\paragraph{Notations and conventions.}
In this paper, standard math or sans serif font stands for classical algorithms (e.g., $C$ or $\algo{Gen}$) and classical variables (e.g., $x$ or $\keys{pk}$).
Calligraphic font stands for quantum algorithms (e.g., $\qalgo{Gen}$) and calligraphic font and/or the bracket notation for (mixed) quantum states (e.g., $\qstate{q}$ or $\ket{\psi}$).

Let $[\ell]$ denote the set of integers $\{1, \cdots, \ell \}$, $\secp$ denote a security parameter, and $y \seteq z$ denote that $y$ is set, defined, or substituted by $z$.
For a finite set $X$ and a distribution $D$, $x \chosen X$ denotes selecting an element from $X$ uniformly at random, and $x \chosen D$ denotes sampling an element $x$ according to $D$. Let $y \gets \algo{A}(x)$ and $y \gets \qalgo{A}(\qstate{x})$ denote assigning to $y$ the output of a probabilistic or deterministic algorithm $\algo{A}$ and a quantum algorithm $\qalgo{A}$ on an input $x$ and $\qstate{x}$, respectively.
PPT and QPT algorithms stand for probabilistic polynomial-time algorithms and polynomial-time quantum algorithms, respectively.
Let $\negl$ denote a negligible function.
For strings $x,y\in \bit^n$, $x\cdot y$ denotes $\bigoplus_{i\in[n]} x_i y_i$ where $x_i$ and $y_i$ denote the $i$th bit of $x$ and $y$, respectively. \fuyuki{Need to check if all the notations here are really used in this paper or not.}
\nikhil{I think so, except we didn't use ; for randomness. Also need to check if everything is covered here.} For random variables $X$ and $Y$, we use the notation $X \approx Y$ to denote that these are computationally indistinguishable. On the other hand, $X \approx_s Y$ denotes statistically indistinguishability between the random variables.

\paragraph{Compute-and-Compare Obfuscation.}
We define a class of circuits called compute-and-compare circuits as
follows:

\begin{definition}[Compute-and-Compare Circuits]\label{def:cc_circuits_searchability}
A compute-and-compare circuit $\cnc{P}{\lock,\msg}$ is of the form
\[
\cnc{P}{\lock,\msg}(x)\left\{
\begin{array}{ll}
    \msg&\textrm{if}\; P(x)=\lock\\
\bot&\text{otherwise}~
\end{array}
\right.
\]
where $P$ is a circuit, $\lock$ is a string called the lock value,
and $\msg$ is a message.
\end{definition}

We now introduce the definition of compute-and-compare obfuscation.
We assume that a program $P$ has an associated set of parameters $\pp_P$ (input size, output size, circuit size) which we do not need to hide.
\begin{definition}[Compute-and-Compare Obfuscation]\label{def:CCObf}
A PPT algorithm $\CCObf$ is an obfuscator for the family of distributions $D=\{D_\secp\}$ if the following holds:
\begin{description}
\item[Functionality Preserving:] There exists a negligible function
$\negl$ such that for all programs $P$, all lock values $\lock$, and
all messages $\msg$, it holds that

\begin{align}
\Pr[\forall x, \tlP(x)=\cnc{P}{\lock,\msg}(x) :
\tlP\la\CCObf(1^\secp,P,\lock,\msg)] \ge 1-\negl(\secp).
\end{align}
\item[Distributional Indistinguishability:] There exists an
efficient simulator $\Sim$ such that for all messages $\msg$, we have
\begin{align}
(\CCObf(1^\secp,P,\lock,\msg),\qaux)\approx(\CCSim(1^\secp,\pp_P,\abs{\msg}),\qaux),
\end{align}
where $(P,\lock,\qaux)\la D_\secp$.
\end{description}
\end{definition}

\begin{theorem}[\cite{FOCS:GoyKopWat17,FOCS:WicZir17}]
If the LWE assumption holds, there exists compute-and-compare obfuscation for all families of distributions $D=\{D_\secp\}$, where each $D_\secp$ outputs uniformly random lock value $\lock$ independent of $P$ and $\qaux$.
\end{theorem}


\ifnum\submission=1
\section{Proof of O2H Lemma with Auxiliary Quantum Oracle}\label{appsec:O2H_oracle}
\else
\fi
\begin{proof}[Proof of~\cref{lem:O2H}]
Let us consider an adversary $\widetilde{\qA}$ that receives as input
the description $\langle \cQ\rangle$ of $\cQ$\footnote{The O2H Lemma (Lemma \ref{lem:O2Hprev}) holds even if $z$ is exponentially large, so the description of $\cQ$ need not be concise.}, along with the input
$z$ used by $\qA$. Given oracle access to $H$ (likewise $G$) and $(z,
\langle \cQ\rangle)$ as input, $\widetilde{\qA}$ simply runs
$\qA^{H, \cQ}(z)$ (likewise $\qA^{G, \cQ}(z)$) by simulating its queries to $\cQ$
using the description $\langle \cQ\rangle$. Then, the O2H Lemma (Lemma
\ref{lem:O2Hprev}) implies the existence of an algorithm
$\widetilde{\qB}^{H}(z, \langle \cQ\rangle)$ that chooses $i \gets
[q]$, runs $\widetilde{\qA}^H(z, \langle \cQ\rangle)$, measures its
$i$-th query to $H$ and outputs the measurement outcome. Observe that
the algorithms $\widetilde{A}$ and $\widetilde{B}$ do not make use of
the description $\langle \cQ\rangle$ except for simulating the queries
made by $\qA$.  Consequently, there exists an algorithm
$\qB^{\cQ, H}(z)$ equivalent to $\widetilde{\qB}^{H}(z, \langle \cQ\rangle)$ that directly runs $\qA$ (instead of $\widetilde{\qA}$)
and simulates its oracle queries to $\cQ$ using its own access to
$\cQ$.
\end{proof}

	\setcounter{tocdepth}{2}
	\tableofcontents

\fi
\else
\fi

\end{document}